\documentclass[12pt,a4paper]{article}
\usepackage[truedimen,margin=30mm]{geometry} 

\usepackage{mathrsfs}
\usepackage{amssymb}
\usepackage{amsmath}
\usepackage{ascmac}
\usepackage{amsthm}
\usepackage[pdftex]{graphicx}
\usepackage{natbib}
\usepackage{setspace}
\usepackage{url}

\usepackage{color}
\usepackage{times}

\usepackage{bm}
\usepackage{bbm}

\usepackage{comment}

\usepackage{titlesec}
\titleformat*{\section}{\large\bfseries}
\titleformat*{\subsection}{\it}
\titleformat*{\subsubsection}{\it}

\newtheorem{df}{Definition}
\newtheorem{thm}{Theorem}
\newtheorem{lem}{Lemma}

\newcommand{\overbar}[1]{\mkern 1.5mu\overline{\mkern-1.5mu#1\mkern-1.5mu}\mkern 1.5mu}

\def\Kc{{\cal K}}
\def\pd{\partial}

\def\ga{{\gamma}}
\def\de{{\delta}}
\def\ep{{\varepsilon}}
\def\la{{\lambda}}
\def\si{{\sigma}}
\def\om{{\omega}}

\def\ka{{\kappa}}

\def\ta{{\tau}}
\def\gat{{\tilde{\gamma}}}
\def\ybt{{\tilde \y}}

\def\bbe{{\text{\boldmath $\beta$}}}

\def\bep{{\text{\boldmath $\varepsilon$}}}

\def\bsi{{\text{\boldmath $\sigma$}}}

\def\bth{{\text{\boldmath $\theta$}}}

\def\bxi{{\text{\boldmath $\xi$}}}

\def\bpsi{{\text{\boldmath $\psi$}}}

\def\bmu{{\text{\boldmath $\mu$}}}

\def\bze{{\text{\boldmath $\zeta$}}}

\def\Si{{\Sigma}}

\def\Om{{\Omega}}

\def\La{{\Lambda}}

\def\bSi{{\text{\boldmath $\Si$}}}

\def\bOm{{\text{\boldmath $\Om$}}}
\def\bLa{{\text{\boldmath $\La$}}}

\def\bPsi{{\text{\boldmath $\Psi$}}}

\def\a{{\text{\boldmath $a$}}}
\def\b{{\text{\boldmath $b$}}}
\def\c{{\text{\boldmath $c$}}}
\def\d{{\text{\boldmath $d$}}}

\def\m{{\text{\boldmath $m$}}}

\def\r{{\text{\boldmath $r$}}}
\def\s{{\text{\boldmath $s$}}}
\def\t{{\text{\boldmath $t$}}}
\def\u{{\text{\boldmath $u$}}}

\def\x{{\text{\boldmath $x$}}}
\def\y{{\text{\boldmath $y$}}}
\def\z{{\text{\boldmath $z$}}}
\def\A{{\text{\boldmath $A$}}}
\def\B{{\text{\boldmath $B$}}}
\def\C{{\text{\boldmath $C$}}}

\def\E{{\text{\boldmath $E$}}}

\def\H{{\text{\boldmath $H$}}}
\def\I{{\text{\boldmath $I$}}}

\def\O{{\text{\boldmath $O$}}}

\def\Q{{\text{\boldmath $Q$}}}
\def\R{{\text{\boldmath $R$}}}
\def\S{{\text{\boldmath $S$}}}
\def\T{{\text{\boldmath $T$}}}

\def\V{{\text{\boldmath $V$}}}

\def\X{{\text{\boldmath $X$}}}

\def\Z{{\text{\boldmath $Z$}}}

\def\Ic{{\cal I}}

\def\Lc{{\cal L}}

\def\Zc{{\cal Z}}

\def\non{{\nonumber}}

\DeclareMathOperator{\tr}{tr}
\DeclareMathOperator{\sgn}{sgn}

\def\diag{{\rm diag\,}}
\def\zero{{\text{\boldmath $0$}}}

\title{{\bf Outlier-Robust Bayesian Multivariate Analysis with Correlation-Intact Sandwich Mixture}\footnote{\today}} 

\date{}

\begin{document}

\maketitle
\doublespacing

\vspace{-1.5cm}
\begin{center}
Yasuyuki Hamura$^1$, Kaoru Irie$^2$ and Shonosuke Sugasawa$^3$

\bigskip
\noindent
$^1$Graduate School of Economics, Kyoto University\\
$^2$Faculty of Economics, The University of Tokyo\\
$^3$Faculty of Economics, Keio University
\end{center}

\vspace{3mm}
\begin{center}
{\bf \large Abstract}
\end{center}

Handling outliers is a fundamental challenge in multivariate data analysis because outliers may distort the structures of correlation or conditional independence. Although robust Bayesian inference has been extensively studied in univariate settings, theoretical results ensuring posterior robustness in multivariate models are scarce. We propose a novel scale mixture of multivariate normals called correlation-intact sandwich mixtures, in which the scale parameters are real values and follow an unfolded log-Pareto distribution. 
Our theoretical results on posterior robustness in multivariate settings emphasize that the use of a symmetric, super heavy-tailed distribution for scale parameters is essential for achieving posterior robustness against element-wise contamination. The posterior inference for the proposed model is feasible using the developed efficient Gibbs sampling algorithm. The superiority of the proposed method was further illustrated further in simulation and empirical studies using graphical models and multivariate regression in the presence of complex outlier structures. 

\bigskip\noindent
{\bf Key words}: Graphical model; Multivariate regression; Outlier-robust posterior inference; Super heavy-tailed distribution.

\newpage
\section{Introduction}

The handling of outliers is a crucial challenge in multivariate data analysis, where anomalies can arise in complex patterns across multiple dimensions.
Unlike in a univariate setting, outliers in multivariate data are not necessarily extreme in all coordinates; they may occur only in a subset of variables (element-wise contamination) or in specific combinations of variables, making them more difficult to detect and mitigate.
Such contamination can severely distort the estimation of the covariance or precision matrices, potentially creating spurious correlations or masking genuine dependencies.
Therefore, robust methods that explicitly account for the interplay between outliers and the underlying correlation structure are essential to ensure reliable multivariate statistical analysis.

Existing approaches to robust Bayesian inference are diverse.
One prominent line of work is based on general-purpose methodologies for model misspecification, such as the c-posterior \citep{miller2019robust}, general Bayes methods employing robust divergence \citep[e.g.][]{jewson2018principles,yonekura2023adaptation}, and Stein discrepancy-based approaches \citep{matsubara2022robust}, which provide principled ways to temper the influence of outliers on the posterior distribution.
Another stream of work leverages heavy-tailed distributions for both Gaussian outcomes \citep{gagnon2020new,hamura2022log}, and non-Gaussian outcomes \citep{hamura2024robust,gagnon2024robust}, allowing extreme observations to be accommodated within the assumed data-generating process.
While these methods and their robustness properties have been widely investigated in the univariate setting, far fewer studies have addressed multivariate cases from either a theoretical or practical perspective. 
Although there have been a few attempts at robust covariance estimation in the context of graphical modeling \citep{finegold2011robust,finegold2014robust,onizuka2023robust}, these studies largely focused on methodological development and empirical performance without establishing a general theoretical framework for posterior robustness. 
Unlike the univariate case, no valid theoretical results guarantee the robustness of posteriors for multivariate data, leaving a significant gap between the methodology and theory. 
In addition, a mere extension of the univariate method, especially the divergence-based method, does not adequately address the unique challenge in the multivariate setting; outliers may be localized to specific coordinates, yet still distort dependence structures, such as covariances or conditional independencies. 
\cite{katayama2018robust} and \cite{raymaekers2024cellwise} studied the estimation of the covariance matrix under cell-wise contamination from a non-Bayesian viewpoint; however, the full posterior inference allowing for a variety of prior choices is still limited.

To address the problem of multivariate analysis with potential outliers, we propose a new robust multivariate model, theoretically prove its posterior robustness, and provide a custom Markov chain Monte Carlo (MCMC) algorithm with illustrations using simulations and real data analyses. Our model is based on a classical contamination model, where the distribution of error terms is the mixture of two components for moderate observations and outliers. The component for outliers utilizes the log-Pareto distribution, which is super heavy-tailed and known to ensure posterior robustness in the univariate case (e.g. \citealt{gagnon2020new} and \citealt{hamura2022log}). 
We note that the Student's $t$-distribution, which has been well studied \citep{west1984outlier} and widely practiced in the context of robust statistics \citep{tak2019robust,da2020bayesian}, does not guarantee posterior robustness, as in the univariate case \citep{hamura2022log,hamura2024posterior,gagnon2023theoretical}. 
To handle element-wise contamination, we introduce latent variables into the covariance matrix based on variance-correlation decomposition, resulting in the product of the original covariance and the two diagonal matrices of latent variables in the sandwich form. Because this form can mitigate the effect of outliers on correlations, we name this model \textit{correlation-intact sandwich mixture (CSM)} of multivariate normals. 

To justify our proposal, we mathematically prove the posterior robustness of the CSM by identifying a set of sufficient conditions for robustness. In addition, we examine each condition assumed in the CSM model to show that the lack of these conditions leads to a loss of posterior robustness. A noteworthy finding is that the latent variable, which is assumed to be positive in the literature, must be real-valued to achieve robustness. Our proposal is further supported by discussions on other possible models and their properties. Finally, the MCMC algorithm is derived for the CSM model; the full conditionals of the latent variables become the multivariate truncated normal distributions and need to be sampled by the Hamiltonian Monte Carlo method, while most of the parameters and variables are conditionally conjugate and easy to be sampled.

The remainder of this paper is organized as follows.  
The model settings and details of the CSM model are presented in Section~\ref{sec:method}. 
Our main theorems of posterior robustness under multivariate models as well as a study of other possible models are provided in Section~\ref{sec:theory}. 
The MCMC algorithm for posterior inference under the CSM is explained in Section~\ref{sec:mcmc}. 
Finally, simulation and empirical studies using graphical models and multivariate linear regression models are conducted in Section~\ref{sec:num}. 
Many technical details, including the proofs of the theorems, are provided in the Supplementary Materials. 
The R code for implementing the proposed method is available in the GitHub repository (\url{https://github.com/sshonosuke/CSM}).

\section{Multivariate Models and Correlation-intact Sandwich Mixture}\label{sec:method}

\subsection{Multivariate linear models}

Let $\y_i=(y_{i1},\ldots,y_{ip})^{\top}$ be a $p$-dimensional vector for $i=1,\ldots,n$, where $n$ is the sample size and $p$ is the dimension of the observed vector. 
We consider a multivariate linear model, $\y_i \sim \mathrm{N}_p (\X _i \bbe , \bSi )$, where $\X_i$ is a $p\times q$ design matrix, $\bbe$ is a $q$-dimensional vector of regression coefficients and $\bSi$ is a $p\times p$ covariance matrix. 
Throughout this paper, we assume that each row of $\X_i$ is not zero, which can easily be guaranteed if, for example, an intercept is included. 
This class of models is simple, yet covers a wide range of multivariate models used in modern statistical applications.
Examples include the Gaussian graphical model with $\X _i \bbe=\zero$ (the all-zero vector), and reduced rank regression with $\bbe={\rm vec}(\B)$, where an appropriate rank condition is assumed for matrix $\B$. 
Our primary interest is inference on coefficient $\bbe$ and covariance matrix $\bSi$ in a Bayesian framework, where we introduce a joint prior distribution for $(\bbe ,\bSi )$.

A main drawback of the standard multivariate model is that the model is sensitive to outliers. 
For example, when at least one element of the residual $\y_i-\X_i\bbe$ is large under given $\bbe$, such information would have significant effects on the inference on $\bSi$, resulting in biased estimation. 
To accommodate such outliers in the sampling model, we introduce latent variables $\t _i = (t_{i1},\ldots,t_{ip})$ as follows: 
\begin{equation} \label{eq:lm}
\y_i | \t_i \sim  \mathrm{N}_p (\X _i \bbe , \V (\bSi, \t_i)),
\end{equation}
where the covariance matrix is now $\V (\bSi, \t_i)$, a function of base covariance matrix $\bSi$ and latent variables $\t_i$. 
There are several options in modeling $\V (\bSi, \t_i)$ (see Section~\ref{sec:others}), among which we aim to identify a subclass of posterior-robust models.

\subsection{Correlation-intact sandwich covariance mixture}
We specify the form of the covariance matrix in (\ref{eq:lm}) as follows:
\begin{equation} \label{eq:model_main}
    \V(\bSi ,\t_i ) = \T _i \bSi \T _i, \qquad \T _i= \mathrm{diag}(\t_i).
\end{equation}
Throughout this study, we call the model in (\ref{eq:lm}) using the covariance in (\ref{eq:model_main}) {\it Correlation-intact Sandwich Mixture of multivariate normals} (CSM). 
In this specification, a large $|t_{i,k}|$ inflates the marginal variance of the $k$-th variable of the $i$-th observation, allowing outlying $y_{i,k}$, while the values of $y_{i,k'}$ for $k'\not=k$ could be still moderate. 
This construction is based on the variance-correlation decomposition of $\bSi$ and is not novel by itself \citep[e.g.][]{ameen1984discount}, but has not been extensively studied due to the complexity of posterior computation. Exceptions include \cite{finegold2011robust}, \cite{choy2014bivariate}, \cite{griffin2024structured} and \cite{okudo2025shrinkage}, where $t_{i,k}^2$ has the polynomial density tail (e.g., the inverse-gamma distribution), with $t_{i,k} > 0$.  
By contrast, the CSM is unique in having distribution of $t_{i,k}$ to be super heavy-tailed and, most importantly, allowing $t_{i,k}$ to be negative. 
Specifically, the latent variables, or $t_{i,k}$ for $i=1,\dots ,n$ and $k=1,\dots ,p$, are mutually independent and identically distributed, whose marginal density is given by $\pi(\cdot)$ with the support,
\begin{equation*}
    t_{i,k} \in (-\infty , -1] \cup [1, \infty ) . 
\end{equation*}
The origin (zero) is excluded from the support of $t_{i,k}$ to avoid a degenerate distribution in (\ref{eq:lm}). 
Subsequently, the marginal variance of $y_{i,k}$ is $t_{i,k}^2 \Si _{k,k}$, where the $(k,k')$-entry of matrix $\bSi$ is denoted by $\Si_{k,k'}$, so the sign of $t_{i,k}$ is not essential in increasing each variance. 
However, the correlation between $y_{i,k}$ and $y_{i,k'}$ is $\sgn (t_{i,k} t_{i,k'}) \Si _{k,k'} / \sqrt{ \Si_{k,k} \Si_{k',k'} } $; its absolute value is independent of the latent variables, but its sign could change, depending on the signs of $t_{i,k}$ and $t_{i,k'}$.  
For example, even if $y_{i,k}$ and $y_{i,k'}$ are positively correlated in the sense that $\Si _{k,k'} > 0$, it is still possible to have positively-outlying $y_{i,k}$ and negatively-outlying or non-outlying $y_{i,k'}$. Figure~\ref{fig:onesidet} illustrates this property by the contour plots of the joint density of $\y_i$, as well as scatter plots of random samples of $\y_i$, with real-valued and positive $\t_i$. 
Using this model specification, we can distinguish the correlation from the simultaneous (non-)occurrence of outliers, which is crucial for achieving posterior robustness in a multivariate setting, as discussed in Section~\ref{sec:theory}.

\begin{figure}[htbp!]
\centering
\includegraphics[width=\linewidth]{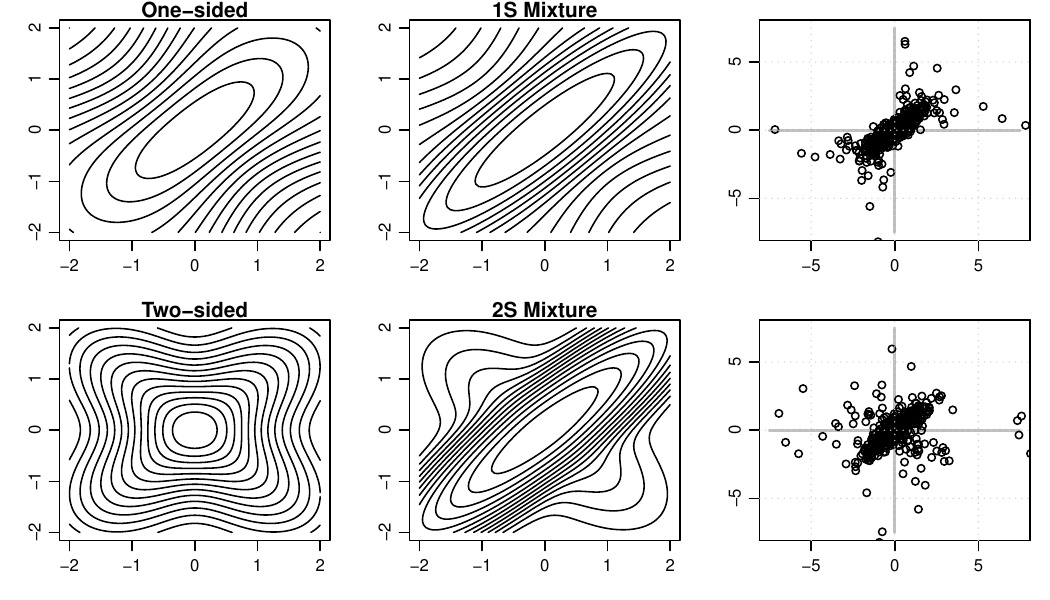}
\caption{Joint densities of $\y_i$ defined by (\ref{eq:lm}) and (\ref{eq:model_main}) in the bivariate case ($p=2$), where $t_{i,k}$ follows the one-sided (top) and double-sided (bottom) distributions, which are either the log-Pareto distribution in (\ref{eq:lpdens}) (left) or the contamination model in (\ref{eq:two}) (middle). The scatter plots of 1000 random variables generated from the contamination model are also shown (right). }
\label{fig:onesidet}
\end{figure}

\subsection{Mixing distribution of latent variable $\t_i$} \label{sec:mixing}

For the super heavy-tailed mixing distribution of $t_{i,k}$, we consider the class of log-Pareto-tailed densities. 
Note that the marginal distribution of $\y_i$ inherits the super heavy tails from the log-Pareto tails of $\t_i$ (see, for example, Proposition~1 in \citealt{hamura2022log}). 
Specifically, we consider the {\it unfolded (and shifted) log-Pareto} density given by 
\begin{equation}\label{eq:lpdens}
\pi_{\rm LP}(t; \ga) = \frac{\gamma}{2 |t| (1 + \log |t|)^{1 + \ga }} \mathbbm{1}[|t| > 1] , 
\end{equation}
where $\mathbbm{1}[\cdot]$ denotes an indicator function and $\gamma > 0$ is a pre-specified hyperparameter. 
As the name suggests, this density is an unfolded and shifted version of the one-sided log-Pareto density \citep[e.g.][]{cormann2009generalizing}.
A notable feature of density (\ref{eq:lpdens}) is that it has a super heavy tail (heavier than the Cauchy distribution) and is symmetric, both of which play essential roles in proving posterior robustness in Section~\ref{sec:theory}. 
Figure~\ref{fig:tdens} shows the density and distribution functions of the unfolded log-Pareto and inverse-$\chi_1$ distributions. As seen in the graphs, these densities exclude the neighborhood of zero and resembles the non-local prior for hypothesis testing \citep{johnson2010use}. The distribution function of the log-Pareto distribution is increasing but very slowly, illustrating its unique tail property. 
The value of hyperparameter $\gamma$ affects the tail behavior only through a logarithmic factor, hence does not greatly contribute to the posterior results. 
In our implementation, we specify a certain value of $\gamma$ ($\gamma=1$) and we use a simplified notation, $\pi_{\rm LP}(t)\equiv \pi_{\rm LP}(t; \gamma)$, in what follows.

\begin{figure}[htbp!]
\centering
\includegraphics[width=\linewidth]{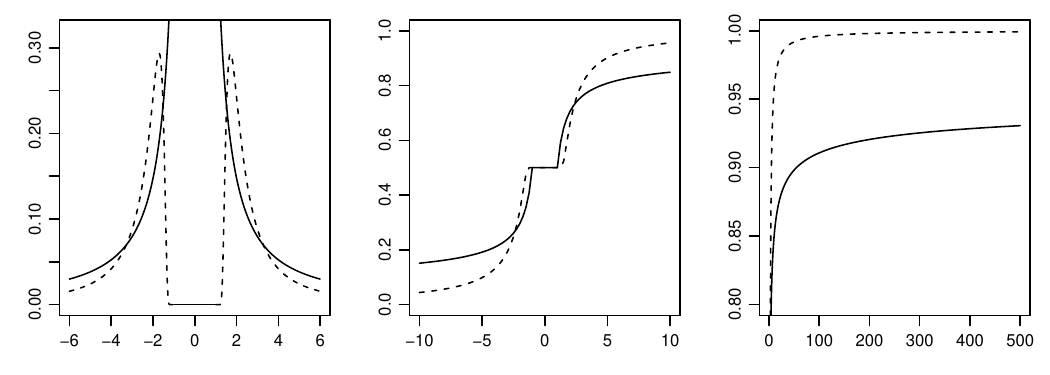}
\caption{ The density (left) and distribution functions (middle and right) of the unfolded, shifted log-Pareto (solid) and inverse-$\chi_1$ (dashed) distributions ($\ga = 1$). Note that $t$ is (unfolded) inverse-$\chi_1$ distributed when $t^{-2}\sim \chi^2_1$. }
\label{fig:tdens}
\end{figure}

Although tail heaviness is important for robustness, the direct use of the unfolded log-Pareto distribution as the distribution of $t_{i,k}$ can overestimate posterior uncertainty. 
Hence, we employ the following two-component mixture model for $t_{i,k}$:
\begin{equation} \label{eq:two}
\pi(t_{i,k}; \phi)=(1-\phi) \delta _{\{ 1 \}} (t_{i,k}) + \phi \pi_{\rm LP}( t_{i,k} ), 
\end{equation}
where $\phi \in (0,1)$ is an unknown proportion, $\delta_{\{1\}}$ is the point-mass distribution on $\{ t_{i,k}=1 \}$ and $\pi_{\rm LP}(t_{i,k})$ is the unfolded log-Pareto density in (\ref{eq:lpdens}). 
By marginalizing $t_{i,k}$ out, the distribution of $y_{i,k}$ is exactly the (point-mass) contamination model \citep{box1968bayesian}, where the second component is replaced with a super heavy-tailed distribution. Due to the second component, $\pi(t_{i,k}; \phi)$ is also log-Pareto tailed as long as $\phi\neq  0$. 
Here, $\phi$ governs both the efficiency and robustness of the posterior.
When $\phi$ is large, the heavy-tailed component $\pi_{\rm LP}(\cdot)$ is emphasized, allowing for more outliers but tending to overestimate posterior uncertainty.
Conversely, when $\phi$ is small, the model implicitly assumes that only a subset of the observations are outliers, thereby reducing posterior uncertainty.
In the proposed CSM, we place a prior on $\phi$, enabling data-dependent learning of $\phi$ and allowing the degree of robustness to be controlled adaptively.

The CSM consists of the model for $\y_i$ in equation (\ref{eq:lm}) and the sandwich form of covariance in equation (\ref{eq:model_main}), with the latent variables following the two-component density in equation (\ref{eq:two}). 
The model includes unknown parameters $(\bbe, \bSi)$ (main parameters of interest) and $\phi$ (proportion to control robustness), for which the prior distribution must be specified. 
As detailed in the next section, a wide range of priors are shown to achieve posterior robustness. However, for computational convenience, we use the independent, conditionally conjugate normal, inverse-Wishart and beta priors for $\bbe$, $\bSi$ and $\phi$, respectively, in the applications in Section~\ref{sec:num}.

\subsection{Remarks on other possible models for covariance} \label{sec:others}

There are several models of variance $\V(\bSi,\t_i)$ that differ from (\ref{eq:model_main}) and have been used in practice. We conclude that the three models discussed below are less flexible than CSM in terms of robustness against element-wise contamination. 

The first model scales the base variance by a univariate latent variable, $t_i \in \mathbb{R}$, as 
$\V(\bSi, t_i ) = t_i^2 \bSi$.
This model is typically adopted to construct a heavy-tailed multivariate distribution, such as a multivariate $t$-distribution \citep[e.g.][]{lange1989robust}, and has been used for its computational simplicity.
However, this model is viewed as unreasonable because the latent variable $t_i$ is common across all elements of $\y_i$, inflating all the marginal variances simultaneously. It is rare for all elements of $\y_i$ are outlying; increasing all the marginal variances for a handful of outliers in $\y_i$ is inefficient. Similar problems would arise to the model based on the Cholesky decomposition of $\bSi$, where the latent variables are multiplied to the eigenvalues of $\bSi$ and inflate not all but many marginal variances simultaneously. This inefficiency can also be observed in the divergence approach to multivariate responses with outliers mentioned in the introduction. We will confirm this observation in the numerical example in Section~\ref{sec:num}.  

The second model of variance is of the additive form, $\V(\bSi,\t_i) = \bSi + \T_i^2$ with $\T_i^2 = \mathrm{diag}(t_{i,1}^2,\dots, t_{i,p}^2)$.
This model can be obtained by adding an error vector as $\y _i = \X_i\bbe + \tilde{\bep}_i + \bep_i$, where $\bep_i$ and $\tilde{\bep}_i$ are mutually independent, $\bep_i \sim {\rm N}_p(\zero,\bSi)$, and $\tilde{\bep}_i \sim {\rm N}_p(\zero,\T_i^2)$. 
The third model is similar to the second one but is based on the correlation matrix as $\V(\bSi,\t_i) = \mathrm{diag}(\bsi ) (\bm{R}+\T_i) \mathrm{diag}(\bsi )$, where $\bsi ^2 = (\Si_{11},\dots,\Si_{pp})$ and $\bm{R}$ is the correlation matrix; $\bm{R} = \{ \diag (\bsi)^{-1} \} \bSi \{ \diag (\bsi)^{-1} \}$. The proposed CSM model differs from the second and third models mostly in its correlation structure. The correlation of $y_{i,1}$ and $y_{i,2}$ defined by (\ref{eq:model_main}) is $\sgn (t_{i,1}t_{i,2}) \Si_{12}/\sqrt{ \Si_{11}\Si_{22} }$, where the latent variables affect the correlation only through their signs. By contrast, the correlations under the second and third models are
\begin{equation*}
\frac{ \Si_{12} }{ \sqrt{ (\Si_{11}+t_{i,1}^2)(\Si_{22}+t_{i,2}^2) } } \qquad {\rm and} \qquad \frac{ \Si_{12} }{ \sqrt{ \Si_{11}\Si_{22}(1+t_{i,1}^2)(1+t_{i,2}^2)} },    
\end{equation*}
respectively. When $|t_{i,1}|$ is large, the correlation of the proposed CSM model is unchanged (except for its sign), whereas those of the two models are reduced to zero. With this structure, the proposed model can realize the separation of latent variables and model parameters, thus simplifying the proof of posterior robustness.

\section{Theory of Posterior Robustness for Multivariate Data}
\label{sec:theory}

In this section, we establish a general theory of posterior robustness for the multivariate linear model. 
In particular, we prove that the CSM achieves posterior robustness under certain regularity conditions, whereas typical multivariate heavy-tailed models (asymmetric and/or polynomial density tails for $t_{i,k}$) fail to hold this property.

\subsection{Multivariate outliers and robustness} \label{sec:def}

The posterior distribution of $(\bbe, \bSi)$ is of the form, 
\begin{equation}\label{eq:pos}
p(\bbe,\bSi | \y ) = \frac{ p_0(\bbe,\bSi) \prod _{i=1}^n p(\y_i|\bbe,\bSi)  }{ 
\int p_0(\bbe,\bSi) \prod _{i=1}^n p(\y_i|\bbe,\bSi) d(\bbe,\bSi) },
\end{equation}
where $\y = (\y_1,\dots ,\y_n)$ is a set of observations, $p_0(\bbe,\bSi)$ is a joint prior distribution and $p(\y_i|\bbe, \bSi)$ is the likelihood with the latent variables being marginalized and defined as 
\begin{equation}\label{eq:like}
p(\y_i|\bbe,\bSi) = \int {\rm N}_p(\y_i| \X_i\bbe,\T_i\bSi\T_i ) \pi(\t_i)d\t_i.
\end{equation}
Here, ${\rm N}_p(\cdot | \bmu, \bSi)$ denotes the density function of the multivariate normal distribution with mean $\bmu$ and variance $\bSi$ and $\pi(\t_i)$ is an arbitrary mixing density of latent variable $\t_i$. An outlier is defined as an extreme value conditioned in the posterior in (\ref{eq:pos}). For each $i=1,\dots ,n$, let $\Kc (i)$ and $\Lc (i)$ be the set of indices of non-outliers and outliers in the $i$-th observation; $\Kc(i) \cup \Lc(i) = \{ 1, \dots, p \}$ and $\Kc(i) \cap \Lc(i) = \emptyset$. Then, we consider the behavior of the posterior distribution (\ref{eq:pos}) when $|y_{i,k}|\to \infty$ for all $k\in\Lc(i)$ and $i\in \{ 1,\dots ,n\}$. 

For simplicity, we further model the outlier values as follows. By following the works of \cite{desgagne2015robustness} and other authors, we define the observed values conditioned in the posterior by $y_{i, k} = c_{i, k} + d_{i, k} \om$, with some constants $(c_{i,k},d_{i,k})$ and $\om > 0$, then define $\Kc (i)$ and $\Lc (i)$ as 
\begin{equation*}
\begin{split}
\Kc (i) = \{ \ k \in \{ 1,\dots, p \}  \ | \ d_{i,k} = 0 \ \}, \ \ \ \ \ \ 
\Lc (i) = \{ \ k \in \{ 1,\dots, p \} \ | \ d_{i,k} \neq 0 \ \} .
\end{split} 
\end{equation*}
By this definition, for all $i \in \{ 1,\dots, n \}$, as $\om \to \infty$, we have $|y_{i,k}|\to \infty$ for all $k \in \Lc (i)$, while $y_{i,k}$ is finite for $k\in \Kc(i)$. The concept of posterior robustness is defined as the behavior of the posterior density in response to such extreme outliers.

\begin{df}\label{df:pos}
The model is posterior robust if, for almost all $(c_{i,k},d_{i,k})$ for all $i$ and $k$, it holds that $p(\bbe ,\bSi | \y ) \to p( \bbe ,\bSi | \y_{\Kc} )$ as $\om \to \infty$, at each $(\bbe ,\bSi)$, where $\y _{\Kc } = \{ \ y_{i,k} \ | \ k\in \Kc (i), \ i\in \{ 1,\dots, n \} \ \}$ is the set of all the non-outliers.
\end{df}
This is a natural extension of the concept of posterior robustness used in the univariate case and requires complete rejection of outliers in the limit (\citealt{andrade2006bayesian,andrade2011bayesian}). Although we focus on the point-wise convergence of posterior densities, one can also consider convergences of other types such as uniform convergence on compact support and weak convergence (\citealt{desgagne2013full,desgagne2015robustness,gagnon2020new}). 

Because the outlying $y_{i,k}$ affects the posterior via the likelihood in (\ref{eq:like}), the robustness of this likelihood is also important. If the limit of $p(\y_i|\bbe,\bSi)$ is proportional to $p(\y_{i,\Kc(i)}|\bbe,\bSi)$ as $\om \to \infty$, the computation of the limiting posterior is greatly simplified.  

\begin{df}\label{df:like}
The model is likelihood robust if, for each $i=1,\dots,n$, there exists positive constant $C_i(\om)$ such that $\displaystyle \lim_{\om \to \infty} p(\y_i |\bbe,\bSi)/C_i(\om) = p(\y_{i,\Kc(i)} |\bbe,\bSi)$ at each $(\bbe,\bSi)$ for almost all $(c_{i,k},d_{i,k})$ for all $k$. 
\end{df}
Under the robustness of likelihood, we claim that the posterior becomes robust as 
\begin{equation*}
\begin{split}
    p(\bbe,\bSi | \y ) &= \frac{ p_0(\bbe,\bSi) \prod _{i=1}^n p(\y_i|\bbe,\bSi) / C_i(\om) }{ \int p_0(\bbe,\bSi) \prod _{i=1}^n p(\y_i|\bbe,\bSi) / C_i(\om) d(\bbe,\bSi) } \\ 
    &\to \frac{ p_0(\bbe,\bSi) \prod _{i=1}^n p(\y_{i,\Kc(i)}|\bbe,\bSi) }{ \int p_0(\bbe,\bSi) \prod _{i=1}^n p(\y_{i,\Kc(i)}|\bbe,\bSi) d(\bbe,\bSi) } = p(\bbe,\bSi | \y_{\Kc} ),
\end{split}
\end{equation*}
as $\om \to \infty$. 
In the denominator, we compute the limit by assuming the exchange of limit and integral. The proof of the posterior robustness is completed by justifying this exchange.

\subsection{Likelihood and posterior robustness}

The likelihood robustness holds under minimal assumptions for a wider class of distributions for $t_{i,k}$ with log-Pareto tails. 
The density, $\pi (\cdot)$, is log-Pareto-tailed if there exists constant $c>0$ and $C>0$ such that
\begin{equation} \label{eq:pareto}
    \lim _{|t|\to \infty} \pi (t) |t| \{ \log |t| \}^{1+c} = C. 
\end{equation}
Trivially, the unfolded log-Pareto density in (\ref{eq:lpdens}) satisfies the condition. Density tails of this type have two notable features crucial for guaranteeing robustness: 
\begin{itemize}
\item {\it Super heavy tails.}  The density tails are heavier than those of any $t$-distribution ($|t|^{1+\nu}$ for some $\nu > 0$ as $|t|\to\infty$). 

\item {\it Symmetry.} The two limits at $t\to\infty$ and $t\to -\infty$ are identical, indicating that the density is symmetric at tails. 
\end{itemize}
The density in (\ref{eq:pareto}) belongs to the class of log-regularly varying distributions \citep[e.g.][]{desgagne2013full}, which is an integral part of the sufficient condition for posterior robustness in univariate Gaussian models \citep{desgagne2013full,desgagne2015robustness}. The likelihood robustness does not hold without those features of density tails, as we will see in Sections~\ref{sec:neccesity-symmetry} and \ref{sec:necessity-tail}. 

The likelihood robustness holds for bounded, log-Pareto tailed mixing distributions.

\begin{thm} \label{thm:unnormalized}
Consider the multivariate model (\ref{eq:lm}) with the sandwich covariance structure (\ref{eq:model_main}), and assume that the mixing density $\pi(\cdot)$ for $t_{i,k}$ satisfies the tail condition (\ref{eq:pareto}) and is bounded.
Then, the likelihood robustness holds with choice of scaling constants, $C_i(\om ) = \prod_{k \in \Lc(i)} \pi (| y_{i, k} |)$ if $\Lc(i)\not=\emptyset$ and $C_i(\om ) = 1$ if $\Lc(i)=\emptyset$.
\end{thm}

The proof is provided in the Supplementary Materials.
To gain some insights in the role of symmetry condition of $\pi(\cdot)$, we study the bivariate case ($p=2$) and provide a sketch of the proof in the next subsection.

We next present our result on the posterior robustness under the CSM model, where we specify $\pi(\cdot)$ by the two-component mixture with the unfolded log-Pareto distribution. 
Below, the value of mixture weight $\phi$ is fixed, so the posterior robustness is proved to be valid for each value of $\phi$.

\begin{thm} \label{thm:normalized}
Assume that the joint prior $p_0(\bbe,\bSi)$ satisfies 
\begin{equation*}
     \mathbb{E} \Big[ (1 + \| \bbe \| ^{1 + \ka } ) \Big( 1 + \sum_{k = 1}^{p} {\si _k}^{\ka } \Big) \{ 1 + \sqrt{\tr ( \bSi ^{- 1} )} \} \Big] < \infty ,
\end{equation*}
for some $\kappa>0$, where $\sigma_k^2$ is the $k$-th diagonal element of $\bSi$ and the expectation is taken with respect to $p_0(\bbe,\bSi)$.
Then, the posterior robustness holds under the CSM model for each value of $\phi \in (0,1)$. 
\end{thm}
The proof is provided in the Supplementary Materials. As discussed in Section~\ref{sec:def}, the likelihood robustness implies the posterior robustness if the exchange of the limit ($\om\to\infty$) and integral (with respect to $(\bbe,\bSi)$) is justified. 
In using the dominated convergence theorem, however, constructing an integrable, upper-bounding function is more difficult than the univariate case. We find such a bounding function by viewing the integral as a risk function, in the spirit of the iterated expectation of risk difference methods (\citealt{kubokawa1994unified,ks1994}). 

Required moment conditions guarantee the integrability of the bounding function. A typical class of priors that satisfies the moment conditions is an independent prior with finite first moments of $\bbe$ and $\bSi^{-1/2}$. Examples of prior distributions in this class include independent normal and inverse-Wishart priors for $\bbe$ and $\bSi$, respectively.

\subsection{On necessity of symmetry for robustness}
\label{sec:neccesity-symmetry}

As explained in the previous section, a novel feature of the proposed model is that the scale parameter $t_{i,k}$ is not always positive, but can take a negative value. 
This feature is explicitly stated in the tail condition in (\ref{eq:pareto}) and assumed in Theorem~\ref{thm:unnormalized}. 
To demonstrate the necessity of a symmetric distribution of $t_{i,k}$ for likelihood robustness, we provide a sketch of the proof of Theorem~\ref{thm:unnormalized} in the bivariate case. 
Again, a rigorous proof of general $p$ is provided in the Supplementary Materials. 

We let $p=2$, fix $i \in \{ 1,\dots ,n \}$, and write $\y_i=(y_{i,1},y_{i,2})^\top$. 
We assume $\Lc(i)=\{ 1 \}$ and $\Kc(i) = \{ 2 \}$. That is, in the $i$-th observation, $y_{i,1}$ is an outlier and should be excluded, whereas $y_{i,2}$ is not outlying and should be conditioned in the posterior. 
The (scaled) likelihood of interest is written as 
\begin{equation}\label{eq:scaled-like}
    \frac{ p(\y_i|\bbe,\bSi) }{ \pi (|y_{i,1}|) } 
    = \int A(t_{i,2};\bbe,\bSi ) {\rm N}(  y_{i,2} | \x_{i,2}^{\top}\bbe , t_{i,2}^2 \Si _{22} )  \pi (t_{i,2}) dt_{i,2},
\end{equation}
where 
\begin{equation} \label{eq:bias}
    A(t_{i,2};\bbe,\bSi ) = \int {\rm N}\left(  \frac{ y_{i,1} - \x_{i,1}^{\top}\bbe }{ t_{i,1}}  \left|  \frac{ \mu _{1\cdot2} }{ t_{i,2}} , \Si_{1\cdot2} \right. \right)  \frac{ \pi (t_{i,1}) dt_{i,1} }{ \pi(|y_{i,1}|) |t_{i,1}| },
\end{equation}
with conditional moments $\mu _{1\cdot2} = \Si_{12}\Si_{22}^{-1}(y_{i,2}-\x_{i,2}^{\top}\bbe)$ and $\Si_{1\cdot2} = \Si _{11}- \Si_{12}\Si^{-1}_{22}\Si_{21}$. 
The term in the integrand, $A(t_{i,2};\bbe,\bSi )$, can be viewed as a ``bias'' term; if $A(t_{i,2};\bbe,\bSi )\to 1$ (or some constant) as $\om\to\infty$, then the scaled likelihood in (\ref{eq:scaled-like}) converges to $p(y_{i2}|\bbe,\bSi)$, concluding that the limit is the likelihood without the outlier as desired.

Since we defined $y_{i,1} = c_{i,1} + d_{i,1}\om$, we have $y_{i,1}/|y_{i,1}|\to \sgn (d_{i,1})$ as $\om \to \infty$. 
To compute the limit of (\ref{eq:bias}), we make a change-of-variable to the latent variable $t_{i,1}$ by $\xi _{i,1} = |y_{i,1}|/t_{i,1}$ with the Jacobian defined by $|d\xi _{i,1}/\xi_{i,1}| = |dt_{i,1}/t_{i,1}|$. The support of the new variable is $|\xi _{i,1}| < |y_{i,1}|$, thus covers the whole real line as $|y_{i,1}| \to \infty$. Then, we observe that 
\begin{equation*}
    \frac{ y_{i,1} - \x_{i,1}^{\top}\bbe - \mu_{1\cdot 2}/t_{i,2}}{ |y_{i,1}| / \xi _{i,1} } \to \sgn (d_{i,1}) \xi _{i,1} \qquad \mbox{as} \quad \om\to\infty.
\end{equation*}
Also, note that $\pi (|y_{i,1}|/\xi _{i,1})/ \{ \pi (|y_{i,1}|) |\xi_{i,1}| \} \to 1$ as $\om\to\infty$, that is, $\pi(\cdot)$ is regularly varying with order 1. 
Thus, assuming the exchange of the limit and integral, we have 
$$
A(t_{i,2};\bbe,\bSi ) \to 
\int_{-\infty}^{\infty} {\rm N}\left(\sgn (d_{i,1}) \xi_{i,1}|0, \Sigma_{1\cdot 2}\right) d\xi_{i,1} = 1 \qquad \mbox{as} \quad \om\to\infty,
$$
indicating likelihood robustness. 
In this evaluation, the role of the negative support of $t_{i,1}$ is clear in having $\xi _{i,1}$ being supported on the entire real line. The other features of $\pi(\cdot)$, the log-Pareto tails and symmetry in the tails, allow us to use the limiting property of the regularly varying function with order 1.

We next illustrate that the likelihood robustness does not hold when the mixing density $\pi(\cdot)$ is not symmetric. 
In such a situation, the bias term in (\ref{eq:bias}) converges not to a constant as in Theorem~\ref{thm:unnormalized}, but depends on $(t_{i,2},\bbe,\bSi)$ in the limit.
Specifically, we consider two scenarios: (I) a one-sided log-Pareto-tailed distribution and (II) two-sided but asymmetric log-Pareto-tailed distributions. 

\begin{itemize}
\item[(I)] {\it One-sided log-Pareto-tailed distribution.}  Let $\pi (t_{i,k})$ have positive support $\{t_{i,k}>1 \}$ with a log-Pareto tail. 
Then, the distribution of $\xi _{i,1} = |y_{i,1}|/t_{i,1}$ has support $0<\xi_{i,1}<|y_{i,1}|$, which covers only the positive half-line. As $\om \to \infty$, the bias term becomes 
\begin{equation*} 
    A(t_{i,2};\bbe,\bSi ) \to \int _0^{\infty} {\rm N}\left(   \xi _{i,1} \left| \sgn(d_{i,1}) \frac{ \mu _{1\cdot2} }{ t_{i,2}} , \Si_{1\cdot2} \right. \right)  d\xi _{i,1},
\end{equation*}
which is the distribution function of $ {\rm N}( \mu _{1\cdot2} / t_{i,2}, \Si_{1\cdot2})$ and depends on $(t_{i,2},\bbe,\bSi)$. As an exception, when $\Si _{12} = 0$, the expression above becomes constant and the robustness holds. In that case, $\mu _{1\cdot2}=0$ and $ A(t_{i,2};\bbe,\bSi )\to 1/2$. For general $\Si_{12}\not= 0$, however, the posterior robustness does not hold. This observation implies that the sign of $t_{i,k}$ controls the effect of outliers on correlation. 

\item[(II)] {\it Two-sided but asymmetric log-Pareto-tailed distribution.} 
We consider the following form of $\pi (t_{i,k})$: 
\begin{equation*}
\pi (t_{i,k}) \propto (1-w) \mathbbm{1}[ t_{i,k}<-1 ] \pi _{c'} (|t|) + w \mathbbm{1}[ t_{i,k}>1 ] \pi _c (|t|),
\end{equation*}
for some $w \in (0,1)$ and $c,c'>0$, where $\pi _c(t)$ is a one-sided log-Pareto tailed density on $(1,\infty)$ with order $c>0$. 
If $c=c'$ and $w=1/2$, then this distribution reduces to the symmetric one given in (\ref{eq:pareto}). 
If $c=c'$, it is immediate that the limit of $A(t_{i,2};\bbe,\bSi )$ is proportional to 
\begin{equation*} 
    (1-w)\int _{-\infty}^0 {\rm N}\left(  \xi_{i,1} \Big| \tilde{\mu}_{i,1\cdot 2}, \Si_{1\cdot2} \right)  d\xi _{i,1} 
    + w\int _0^{\infty} {\rm N}\left(  \xi _{i,1} \Big| \tilde{\mu}_{i,1\cdot 2}, \Si_{1\cdot2} \right)  d\xi _{i,1},
\end{equation*}
with $\tilde{\mu}_{i,1\cdot 2}=\sgn(d_{i,1})\mu_{1\cdot2}/ t_{i,2}$. This limit depends on $(t_{i,2},\bbe,\bSi)$ if $w\not=1/2$ and $\Si_{12}\not=0$. 
\end{itemize}

\subsection{On necessity of super heavy tails for robustness} \label{sec:necessity-tail}

We next consider the density with lighter tails than (\ref{eq:pareto}), denoted by $\pi _{c,c'}(t_{i,k})$, which satisfies $\pi _{c,c'}(|t|) |t|^{1+c'} \{ \log |t| \} ^{1+c} \to C$ as $|t|\to\infty$ for some constants $(c,c',C)$. 
If $c'=0$, then this is the log-Pareto tailed density in (\ref{eq:pareto}). 
This class of densities also includes the inverse-gamma distribution by setting $c=-1$ to delete the log-term, which implies the Student's $t$-distribution of $y_{i,k}$ with degree-of-freedom $c'$ \citep{finegold2011robust,choy2014bivariate}. This density satisfies $ \pi _{c,c'}(|y_{i,1}|/\xi _{i,1})/\pi (|y_{i,1}|) |\xi_{i,1}|^{1+c'} \to 1$ as $|y_{i,1}|\to\infty$, namely, it is regularly varying with order $1+c'$. The bias term in (\ref{eq:bias}) becomes: 
\begin{equation*} 
    A(t_{i,2};\bbe,\bSi ) \to \int _{-\infty}^{\infty} |\xi _{i,1}|^{c'} {\rm N}\left(  \xi _{i,1} \left| \sgn(d_{i,1}) \frac{ \mu _{1\cdot2} }{ t_{i,2}} , \Si_{1\cdot2} \right. \right)  d\xi _{i,1}\qquad \mbox{as} \quad \om\to\infty,
\end{equation*}
which is the $c'$-th absolute moment of ${\rm N}( \sgn(d_{i,1}) \mu _{1\cdot2} / t_{i,2} , \Si_{1\cdot2})$ and, if $c'>0$, depends on $(t_{i,2},\bbe,\bSi)$. Note that this dependence on $(t_{i,2},\bbe,\bSi)$ remains even if $\Si _{12} = 0$. In this case, $y_{i,1}$ and $y_{i,2}$ become independent, and the problem is reduced to a univariate case. This observation is consistent with the findings of the univariate analysis in which the density with the polynomial tail ($c'>0$) is unable to achieve the posterior robustness.

\subsection{Robustness of other covariance models} \label{sec:others_robust}

The alternative covariance models considered in Section~\ref{sec:others} are unable to achieve likelihood/posterior robustness without additional assumptions. Under the first model, $\V(\bSi, t_i ) = t_i^2 \bSi$, the robustness of likelihood does not hold when $p=2$ and $\bSi$ is diagonal. See the Supplementary Materials for further details. The second and third models complicate the posterior density as a function of $\t_i$, making the theoretical analysis difficult. Although not reported in detail here, we proved the robustness of these models, but required stronger assumptions than those in Theorems~1 and 2.

\section{Posterior computation of CSM}
\label{sec:mcmc}

\subsection{MCMC: difficulty and proposal}

The MCMC method for the proposed model should sample $(\bbe,\bSi,\phi)$ and $\t_{1:n}$ iteratively. Sampling $\bbe$, $\bSi$ and $\phi$ from their full conditionals is straightforward since the sampling model is simply Gaussian and linear and enables the use of priors and computational methods developed for such Gaussian linear models, including the conjugate normal and inverse-Wishart priors for $\bbe$ and $\bSi$. In sampling $\phi$, we introduce binary latent variable $z_{i,k}$ for each $i$ and $k$ by writing the contamination model in (\ref{eq:two}) as follows: 
\begin{equation*}
\begin{split}
    t_{i,k} | (z_{i,k}=0) \sim \delta _{\{1\}}(t_{i,k}), \ \ \ \ 
    t_{i,k} | (z_{i,k}=1) \sim \pi_{\rm LP}(t_{i,k}),  
\end{split}
\end{equation*}
and $\mathbb{P}[z_{i,k}=1] = 1-\mathbb{P}[z_{i,k}=0] =  \phi$. Then, the beta prior for $\phi$ becomes conditionally conjugate. 

The challenging part of the MCMC algorithm is the sampling of the latent variables $\t_{1:n}$ and $\z_{1:n}$. In the univariate case ($p=1$), the sampling of $t_1,\dots ,t_n$ is straightforward; they are conditionally independent and can be sampled in parallel from well-known distributions.  
In the multivariate case, however, the full conditional of $\t_i$ is no longer any well-known distribution due to the interaction term of $t_{i,k}$ and $t_{i,k'}$ for $k\not=k'$ in the likelihood. More importantly, although we still enjoy conditional independence across observations, we have correlations within the observations. That is, $\t_1,\dots, \t_n$ are mutually independent, but $t_{i,1},\dots, t_{i,p}$ are not independent for each $i$ conditional on the data and other parameters. Sampling each element of $\t_i$ individually from the two-component mixture is feasible but most likely results in the highly-autocorrelated Markov chains. The same within-observation dependence can be observed in $z_{i,1},\dots ,z_{i,p}$ for each $i$. From the viewpoint of computational statistics, this problem resembles a posterior computation using point-mass spike-and-slab priors for variable selection. To accelerate the convergence of the Markov chains, we considered the slice sampler in sampling $\t_i$, as well as introduced more latent variables to decrease the autocorrelation among $z_{i,1},\dots ,z_{i,p}$.

\subsection{Sampling of $\z_i$ with latent variable $\bth_i$}
To see the difficulty in sampling the state variables, observe that latent variable $\t_i$ appears in the likelihood only through $p(\y_i|\bbe,\bSi,\t_i) = \mathrm{N}(\y_i|\X_i\bbe,\T_i\bSi\T_i)$. Then the likelihood function has $t_{i,k}^{z_{i,k}} t_{i,k'}^{z_{i,k'}}$ in the exponential term, so it is clear that $z_{i,k}$ is not (conditionally) independent of $z_{i,k'}$. However, conditional on $\t_i$, this dependence is cut by another set of latent variables, which will be introduced below as $\bth_i$ (e.g., \citealt{zheng2025gibbs}). 

For simplicity, denote the ``standardization'' of $\y_i$ by $\tilde{\y}_i=(\tilde{y}_{i,1},\dots ,\tilde{y}_{i,p})^{\top}$, where $\tilde{y}_{i,k}$ is the $k$-th element of $( \y_i-\X_i\bbe )$ divided by $t_{i,k} ^{z_{i,k}}$. Note that $\tilde{y}_{i,k}$ depends on $z_{i,k}$, but not on $z_{i,k'}$ for $k'\not=k$. With this notation, the likelihood becomes
\begin{equation*}
    \mathrm{N}_p(\y_i|\X_i\bbe,\T_i\bSi\T_i) = \prod _{k=1}^p t_{i,k}^{-z_{i,k}} \ \mathrm{N}_p(\tilde{\y}_i|\zero,\bSi), 
\end{equation*}
which is known to be the marginal of the normal-normal model: 
\begin{equation*}
    \mathrm{N}_p(\tilde{\y}_i|\zero,\bSi) = \int \mathrm{N}_p(\tilde{\y}_i|\bth_i,c\I_p) \mathrm{N}_p(\bth_i|\zero,\bSi-c\I_p) d\bth _i, 
\end{equation*}
where $c>0$ is chosen so that $\bSi > c\I_p$. 
Also, the singular value decomposition of $\bSi$ is given by $\bSi = \H^{\top}\bLa \H$, where $\bLa$ is the diagonal matrix of the singular values, or $\bLa = \diag (\lambda _1,\dots, \lambda_p)$, and $\H$ is orthonormal, or $\H^{\top}\H = \I_p$. Subsequently, the integrand above can be computed as 
\begin{equation*}
\begin{split}
&\mathrm{N}_p(\tilde{\y}_i|\bth_i,c\I_p) \mathrm{N}_p(\bth_i|\zero,\bSi-c\I_p) \\
&\propto 
\exp \left\{ -\frac{1}{2} \left[ (\H \bth_i)^{\top} \tilde{\bLa} (\H\bth_i) -2 c^{-1}\tilde{\y}_i^{\top} \bth_i  \right] - \frac{1}{2c}\|\tilde{\y}_i\|^2 \right\},
\end{split}
\end{equation*}
where $\tilde{\bLa} = c^{-1}\I_p + \diag( (\lambda_1-c)^{-1}, \dots , (\lambda_p-c)^{-1} )$. In the above expression,  we obtain no cross term of $\tilde{y}_{i,k}$ and $\tilde{y}_{i,k'}$, showing that $z_{i,1},\dots ,z_{i,p}$ can be sampled in parallel. This form also shows that the full conditional of $\bth_i$ is a multivariate normal distribution, from which it is easy to simulate. In practice, we choose $c$ so that $c > \lambda _1$, where $\lambda_1$ is the smallest eigenvalue of $\bSi$; see the Supplementary Materials for details.

\subsection{Sampling of $\t_i$ with latent variable $\u_i$}

When computing the full conditional of $\t_i$, note that the likelihood does not depend on $t_{i,k}$ if $z_{i,k}=0$. Thus, for $k \in \{ 1,\dots ,p \}$ that satisfies $z_{i,k}=0$, the full conditional equals the prior. Sampling from the prior requires simulation from the unfolded log-Pareto distribution, which is in fact straightforward (see the Appendix). 

For those with $z_{i,k}=1$, we augment the log-regularly varying term in (\ref{eq:lpdens}) in the spirit of the slice sampler, as follows: 
\begin{equation*}
    \{ 1+\log|t_{i,k}| \} ^{-(1+\gamma)} = \int \mathbbm{1}[ \ 0 < u_{i,k} < \{ 1+\log|t_{i,k}| \} ^{-(1+\gamma)} \ ] \ du_{i,k}, 
\end{equation*}
and introduce another set of latent variables $\u_i = (u_{i,1},\dots,u_{i,p})^{\top}$, whose full conditional distributions are the independent uniform distributions. In sampling $t_{i,k}$, we work on its reciprocal $\tilde{t}_{i,k} = 1/t_{i,k}$ ($|\tilde{t}_{i,k}|<1$). Also, following the notation in the Supplementary Materials, we write the subvectors and submatrices by using $\Zc(i)=\{ k \in \{ 1,\dots, p \} | z_{i,k}=1 \}$ and $\overbar{\Zc}(i) = \{ 1,\dots, n \} \setminus \Zc(i)$. For example, the vector of $\tilde{t}_{i,k}$ with $z_{i,k}=1$ is $\tilde{\t}_{i,\mathcal{Z}(i)}$. To write down the density of the full conditional, let $\bm{1}$ be the all-one column vector of the appropriate dimension and $\bPsi _i = \diag (\y_i-\X_i\bbe) \bSi^{-1} \diag (\y_i-\X_i\bbe)$. Then, the full conditional of $\tilde{\t}_{i,\Zc(i)}$ has the density proportional to  
\begin{equation} \label{eq:rdist}
\begin{split}
    &\exp \left\{ -\frac{1}{2} \left[ \tilde{\t}_{i,\Zc(i)}^{\top} \bPsi _{i,\Zc(i),\Zc(i)} \tilde{\t}_{i,\Zc(i)} - \bm{1}^{\top} \bPsi _{i,\overbar{\Zc}(i),\Zc(i)} \tilde{\t}_{i,\Zc(i)} \right] \right\} \\
    &\qquad \times \prod _{ \{ k | z_{i,k} = 1 \} } \mathbbm{1}\!\left[ \ \exp\left\{ 1 - u_{i,k}^{-1/(1+\gamma ) } \right\} < |\tilde{t}_{i,k}| < 1 \  \right],    
\end{split}
\end{equation}
hence is a multivariate normal distribution truncated at bounded intervals. Because direct sampling from the truncated multivariate normal distribution is difficult, we use a Hamiltonian Monte Carlo-based sampler tailored for such a distribution in \cite{pakman2014exact}.

\subsection{The sketch of the proposed algorithm}

Note that the full conditionals of $\z_i$ and $\bth_i$ are unchanged even with the introduction of $\u_i$. The full conditional of $\u_i$ is dependent only on $\t_i$, while the density of $\t_i$ in (\ref{eq:rdist}) is obtained with $\bth_i$ being marginalized out. Thus, sampling from all the distributions we have derived is justified as a collapsed Gibbs sampler \citep{van2008partially}. An iteration of our MCMC method, while $(\bbe,\bSi,\phi)$ is conditioned, is summarized as follows. 

\begin{itemize}    
    \item Sampling from $p(\bth_i|-)$: The sample of $\bth_i$ is obtained via $\H \bth _i \sim \mathrm{N}_p( \tilde{\bLa} \H^{\top} \tilde{\y}_i , \tilde{\bLa}^{-1} )$. 

    \item  Sampling from $p(\z_i|-)$: The samples of $(z_{i,1},\dots, z_{i,p})$ are generated independently from the following Bernoulli distribution, 
    \begin{equation*}
        \mathbb{P}[ z_{i,k}=z ] \propto \frac{ \phi^z (1-\phi)^{1-z} }{ t_{i,k}^z } \exp \left\{  - \frac{\tilde{y}_{i,k}^2 - 2\tilde{y}_{i,k}\theta_{i,k}}{2c} \right\},
    \end{equation*}
    where $\tilde{y}_{i,k}$ is computed for each value of $z \in \{ 0, 1 \}$. 
    
    \item Sampling from $p(\u _i|-)$: sample $u_{i,k}$ from the uniform distribution on interval $[ 0 , \{ 1+\log|t_{i,k}| \} ^{-(1+\gamma)}]$ independently. 

    \item Sampling from $p(\t_i|-)$: The sample of $\tilde{\t}_i$ is generated from the truncated multivariate normal distribution in (\ref{eq:rdist}) using the Hamiltonian Monte Carlo method. Then, we set $t_{i,k}=1/\tilde{t}_{i,k}$. 
\end{itemize}

\subsection{Further details of posterior computation}

The above algorithm was further modified for efficiency and error handling. First, another set of latent variables is introduced to avoid numerical errors occasionally occurring during the HMC step. Second, the singular value decomposition is applied not to $\bSi$, but to the correlation matrix of precision $\bSi^{-1}$, before $\bth_i$ is introduced. Finally, the sign of $t_{i,k}$ is sampled with $z_{i,k}$ to accelerate the convergence. See Supplementary Materials for the complete algorithm.

\section{Numerical Study}\label{sec:num}

In this section, graphical and multivariate regression models are estimated in the framework of CSM and other robust models using simulated and real datasets. In all models and datasets, the prior for the model parameters is the independent, conditionally conjugate priors:  $\bbe\sim \mathrm{N}_q(\b_0, \B_0)$, $\bSi\sim {\rm IW}(\nu_0, \S_0)$ (inverse-Wishart distribution with $\mathbb{E}[\bSi^{-1}] = \nu _0 \S_0$ ), and $\phi\sim {\rm Beta}(a_0, b_0)$, where we set $\b_0=\zero$ and $\B_0=10^2 \I$ ($\I$ is the identity matrix); $\nu_0=p$ and $\S_0=(\nu_0+p+1)^{-1}\I$; and $(a_0,b_0)=(0.05,1)$.

\subsection{Simulation study 1: graphical modeling}
\label{sec:sim1}

First, we investigate the numerical performance of the proposed method in the context of Gaussian graphical modeling. 
We set $n=200$ (sample size), $p=12$ (dimension) and considered the true precision matrix $\bOm = \bSi^{-1}$ as $\Omega_{jj}=1$, $\Omega_{jk}=0.5$ for $|j-k|=1$, $\Omega_{jk}=0.25$ for $|j-k|=2$ and $\Omega_{jk}=0$ otherwise.  
The genuine (uncontaminated) data $\y_i^{\ast}$ was generated from $\mathrm{N}_p(\zero, \bOm^{-1})$.
To generate a contaminated observation $\y_i$, we first generated $z_i^{\ast}$ from the Bernoulli distribution with $\mathbb{P}[z_i^{\ast}=1] = \phi^{\ast}$, set $\y_i=\y_i^{\ast}$ if $z_i^{\ast}=0$, and if $z_i^{\ast}=1$, added 10 to the randomly-selected $c_i$ elements of $y_i^{\ast}$ before setting $\y_i=\y_i^{\ast}$. 
Regarding the choice of $c_i$, we adopted $c_i=1$ (Scenario 1), $c_i=2$ (Scenario 2) and $c_i= 1+c_i^{\ast}$ (Scenario 3), where $c_i^{\ast}$ was generated from the Poisson distribution with mean $1$. 
Note that $\phi^{\ast}$ does not represent the outlier ratio. 
For example, under Scenario 1 with $\phi^{\ast}=0.6$, the expected numbers of outliers is $n\phi^{\ast}(=120)$ among $np(=2400)$ observations, so the actual outlier ratio is $5\%$.

For the generated dataset, we applied the robust graphical models based on the proposed CSM model with $\y_i\sim \mathrm{N}_p(\zero, \T_i \bOm^{-1} \T_i)$. 
Additionally, we employed a one-sided log-Pareto distribution supported only on the positive half-line for the mixing distribution of $\t_i$ (Model (I) in Section~\ref{sec:neccesity-symmetry}), denoted by PCS. 
For comparison, we used the standard Bayesian Gaussian graphical modeling (GG), $\y_i\sim \mathrm{N}_p(\zero, \bOm)$, which would be highly sensitive to outliers. 
Moreover, we adopted Robust graphical modeling with the classical multivariate $t$-distribution (CT) ($\V(\bSi,t_i)=t_i^2\bSi$ and $t_i^2$ is inverse-gamma distributed), alternative $t$-distribution (AT) by \cite{finegold2011robust}, and Dirichlet $t$-distribution (DT) by \cite{finegold2014robust}.
The degree of freedom in the Student's $t$ distribution had the uniform prior on $[1,100]$ in CT and AT, or was fixed to $3$ in DT, as practiced in \cite{finegold2014robust}. 
We also applied robust Gaussian graphical modeling with $\gamma$-divergence (GD) by \cite{onizuka2023robust} with two choices of tuning parameters: $\gamma=0.05$ (GD1) and $\gamma=0.1$ (GD2).

For each model, we generated $1000$ posterior samples after discarding the first $1000$ samples as burn-in, computed the posterior means, and computed element-wise $95\%$ credible intervals for each (diagonal and upper-triangular) entry of $\bOm$. 
The performance of the listed models was evaluated using the mean squared error (MSE) of the posterior means, coverage probability (CP), and average interval length (AL) of the credible intervals, based on 200 replications. 

Table~\ref{tab:sim} reports the three performance measures averaged over all the elements of $\bOm$. 
Overall, the proposed CSM model exhibits the best performance in terms of MSE, and the degree of improvement became more significant when the proportion of outliers increased. 
Robust approaches based on the $t$-distribution (CT, AT, and DT) tend to provide better performance than the Gaussian model (GG), but the MSE is much larger than that of the CSM. This observation is consistent with the discussion in Sections~\ref{sec:necessity-tail}, in that the $t$-distributions are generally not posterior robust. 
The $\gamma$-divergence approaches (GD1 and GD2) have smaller MSEs than GG in many scenarios, exemplifying its robustness. However, GD1 has much larger MSEs than GD2, implying the sensitivity of the model performance to the choice of $\gamma$.
Regarding posterior uncertainty quantification, it is observed that the CPs of the proposed CSM are around the nominal level, even under heavy contamination.
While AT, DT, and GD2 provide reasonable CP values, the ALs of these methods are much larger than those of the CSM, indicating the efficiency of the CSM in full posterior inference.

A comparison of the proposed CSM and PCS shows that CSM significantly outperforms PCS in both point and interval estimations. 
In particular, the CPs of PCS are considerably smaller than the nominal level when we have multiple outliers in an observation vector $\y_i$ (in Scenarios 2 and 3). 
This numerical result confirms the theoretical properties in Section~\ref{sec:theory} on the (non-)robustness of the CSM (PCS) and supports the use of symmetric (real-valued) latent variables in the presence of outliers.

\begin{table}[htb!]
\caption{The values of MSE, CP and AL, averaged over 200 Monte Carlo replications, under Gaussian graphical models with $n=200$ and $p=12$. 
Values of MSE and AL are multiplied by 100. The smallest values of MSE are highlighted in bold.}
\label{tab:sim}
\centering

\bigskip
\begin{tabular}{ccccccccccccccc}
\hline
\multicolumn{2}{r}{Scenario} & $\phi^{\ast}$ &  & CSM & PCS & GG & CT & AT & DT & GD1 & GD2 \\
\hline
 & 1 & 0 &  & 0.54 & 0.54 & 0.54 & 0.53 & 0.50 & 13.15 & {\bf 0.44} & 0.46 \\
 & 1 & 0.2 &  & {\bf 0.51} & 0.61 & 6.68 & 2.39 & 2.59 & 4.91 & 2.89 & 0.64 \\
 & 1 & 0.4 &  & {\bf 0.55} & 0.72 & 9.06 & 4.67 & 14.12 & 2.63 & 7.07 & 1.53 \\
 & 1 & 0.6 &  & {\bf 0.56} & 0.80 & 10.10 & 6.85 & 15.12 & 3.02 & 9.26 & 5.24 \\
MSE & 2 & 0.2 &  & {\bf 0.48} & 0.79 & 9.00 & 2.92 & 14.29 & 4.04 & 0.71 & 0.55 \\
 & 2 & 0.4 &  & {\bf 0.58} & 1.19 & 10.74 & 5.26 & 16.25 & 3.08 & 8.26 & 0.95 \\
 & 2 & 0.6 &  & {\bf 0.64} & 1.55 & 11.42 & 7.38 & 19.01 & 4.46 & 10.97 & 2.25 \\
 & 3 & 0.4 &  & {\bf 0.61} & 1.18 & 10.65 & 5.19 & 16.51 & 3.16 & 7.00 & 0.90 \\
 & 3 & 0.6 &  & {\bf 0.67} & 1.72 & 11.38 & 7.39 & 18.68 & 4.25 & 10.75 & 1.95 \\
 \hline
 & 1 & 0 &  & 92.9 & 92.9 & 98.6 & 98.7 & 98.6 & 99.3 & 93.0 & 93.5 \\
 & 1 & 0.2 &  & 95.5 & 91.6 & 79.1 & 95.6 & 99.1 & 99.6 & 73.8 & 93.6 \\
 & 1 & 0.4 &  & 96.6 & 89.8 & 69.6 & 86.3 & 99.8 & 99.6 & 62.7 & 94.1 \\
 & 1 & 0.6 &  & 97.3 & 88.5 & 68.2 & 78.5 & 99.8 & 98.2 & 62.3 & 82.0 \\
CP ($\%$) & 2 & 0.2 &  & 97.0 & 88.0 & 70.8 & 94.9 & 99.7 & 99.9 & 92.7 & 95.0 \\
 & 2 & 0.4 &  & 97.9 & 79.5 & 68.2 & 85.4 & 99.7 & 99.9 & 64.2 & 94.3 \\
 & 2 & 0.6 &  & 98.6 & 77.8 & 68.2 & 78.6 & 99.8 & 96.4 & 58.3 & 95.5 \\
 & 3 & 0.4 &  & 97.7 & 79.2 & 68.2 & 85.6 & 99.9 & 99.9 & 63.1 & 94.1 \\
 & 3 & 0.6 &  & 98.6 & 69.1 & 68.2 & 78.6 & 99.8 & 97.1 & 56.8 & 96.0 \\
 \hline
 & 1 & 0 &  & 0.26 & 0.26 & 0.91 & 0.92 & 0.95 & 2.03 & 0.28 & 0.29 \\
 & 1 & 0.2 &  & 0.28 & 0.26 & 0.35 & 0.67 & 1.21 & 1.54 & 0.20 & 0.33 \\
 & 1 & 0.4 &  & 0.31 & 0.27 & 0.21 & 0.50 & 2.08 & 1.18 & 0.10 & 0.38 \\
 & 1 & 0.6 &  & 0.33 & 0.27 & 0.15 & 0.35 & 2.12 & 0.91 & 0.05 & 0.25 \\
AL & 2 & 0.2 &  & 0.30 & 0.27 & 0.26 & 0.68 & 2.08 & 1.47 & 0.32 & 0.33 \\
 & 2 & 0.4 &  & 0.35 & 0.27 & 0.14 & 0.51 & 2.16 & 1.05 & 0.13 & 0.42 \\
 & 2 & 0.6 &  & 0.39 & 0.27 & 0.10 & 0.37 & 2.25 & 0.75 & 0.04 & 0.57 \\
 & 3 & 0.4 &  & 0.35 & 0.27 & 0.15 & 0.51 & 2.18 & 1.11 & 0.15 & 0.40 \\
 & 3 & 0.6 &  & 0.40 & 0.27 & 0.10 & 0.36 & 2.24 & 0.80 & 0.04 & 0.55 \\
\hline
\end{tabular}
\end{table}

\subsection{Simulation study 2: multivariate regression}

Next, we demonstrate the proposed method using multivariate regression. 
We set $n=200$, $p=10$, and $q=10$; we also considered $p=5$, as reported in Supplementary Materials. 
For $i=1,\ldots,n$, the $p\times q$ covariate matrix $\X_i=(\x_{i1},\ldots,\x_{ip})^\top$ was constructed by generating $\x_{ik}\sim \mathrm{N}_q(\zero, \R(0.3))$ independently for $k=1,\ldots,p$, where $\R(\rho)$ is the $q\times q$ equi-correlation matrix with diagonal elements being 1 and off-diagonal elements being $\rho \in (0,1)$.
Then, a $p$-dimensional genuine (uncontaminated) response vector $\y_i^{\ast}$ was generated as 
$$
\y_i^{\ast}=\X_i\bbe + \bep_i, \ \ \ \bep_i\sim \mathrm{N}_p(\zero, \bSi), 
$$
where the true values of regression coefficients $\bbe=(\beta_1,\ldots,\beta_q)^\top$ were set as $\beta_1=0.5$, $\beta_2=1$, $\beta_3=-1$, and $\beta_6=0.5$, while the other coefficients were all zero. The true covariance matrix $\bSi$ was defined by $\Sigma_{k,k'}=(0.6)^{|k-k'|}$ for all $k,k'$. 
Then, we generated a contaminated observation $\y_i$ in the same manner as in Scenario 1 in Section~\ref{sec:sim1}; we generated binary $z_{i}^{\ast}$ with success probability $\phi^{\ast}$ and, if $z_{i}^{\ast}=1$, then we selected a single index $k \in \{ 1,\dots,p \}$ randomly and set $y_{i,k}=y_{i,k}^{\ast}+10$ and $y_{i,k'}=y_{i,k'}$ for $k'\not=k$. If $z_i^{\ast}=0$, then $\y_i=\y_i^{\ast}$. 
For the contamination ratio, we considered 7 cases: $\phi ^{\ast} \in \{0, 0.1,.\ldots,0.6\}$.

For each generated dataset, we applied the proposed CSM, standard Gaussian distribution (G), conventional multivariate $t$-distribution (CT), and an alternative $t$-distribution (AT). 
As in Section~\ref{sec:sim1}, 1000 posterior samples were stored after discarding the first 1000 samples as burn-in, to report posterior means and $95\%$ credible intervals of each element of $\bbe$ and $\bSi$. 
To evaluate the model performance, we computed the interval scores (IS) of credible intervals \citep{gneiting2007strictly}, in addition to MSE and CP. 
IS is the average length of the credible intervals but is further penalized when the interval fails to cover the true value; hence, a smaller IS is better. 
The MSEs, CPs and ISs were computed based on 200 replications, and averaged over elements in $\bbe$ and $\bSi$, respectively. 
The MSEs values are shown in Figure~\ref{fig:MLR-mse} with Monte Carlo error. 
While the four methods are comparable under $\phi^{\ast}=0$ (no contamination), as $\phi^{\ast}$ increases, the MSEs of both $\bbe$ and $\bSi$ of CT, AT and G increases significantly, whereas those for the posterior-robust CSM remained small. 
The CPs and ISs are shown in Figure~\ref{fig:MLR-cp}.
While the CPs of the CSM, CT, and G are around the nominal level, the ISs of the CSM are much smaller than those of the other models, indicating the efficiency of the CSM in covering true values without unnecessarily widening the intervals.

\begin{figure}[htbp!]
\centering
\includegraphics[width=\linewidth]{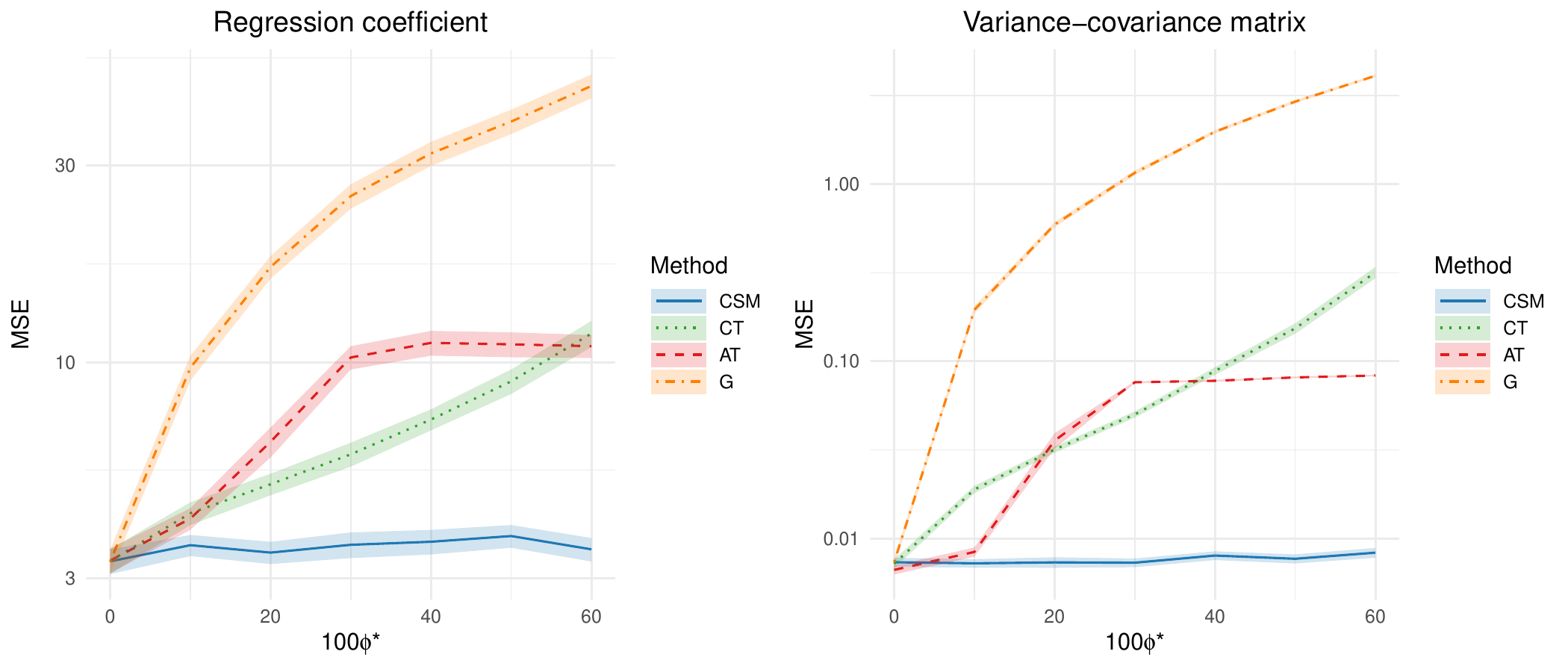}
\caption{ Mean squared errors (MSE) of posterior means of regression coefficients $\bbe$ (left) and variance-covariance matrix $\bSi$ (right) in multivariate linear regression with $p=10$. 
The shaded region represents estimated Monte Carlo errors (twice the standard deviation). 
}
\label{fig:MLR-mse}
\end{figure}

\begin{figure}[htbp!]
\centering
\includegraphics[width=\linewidth]{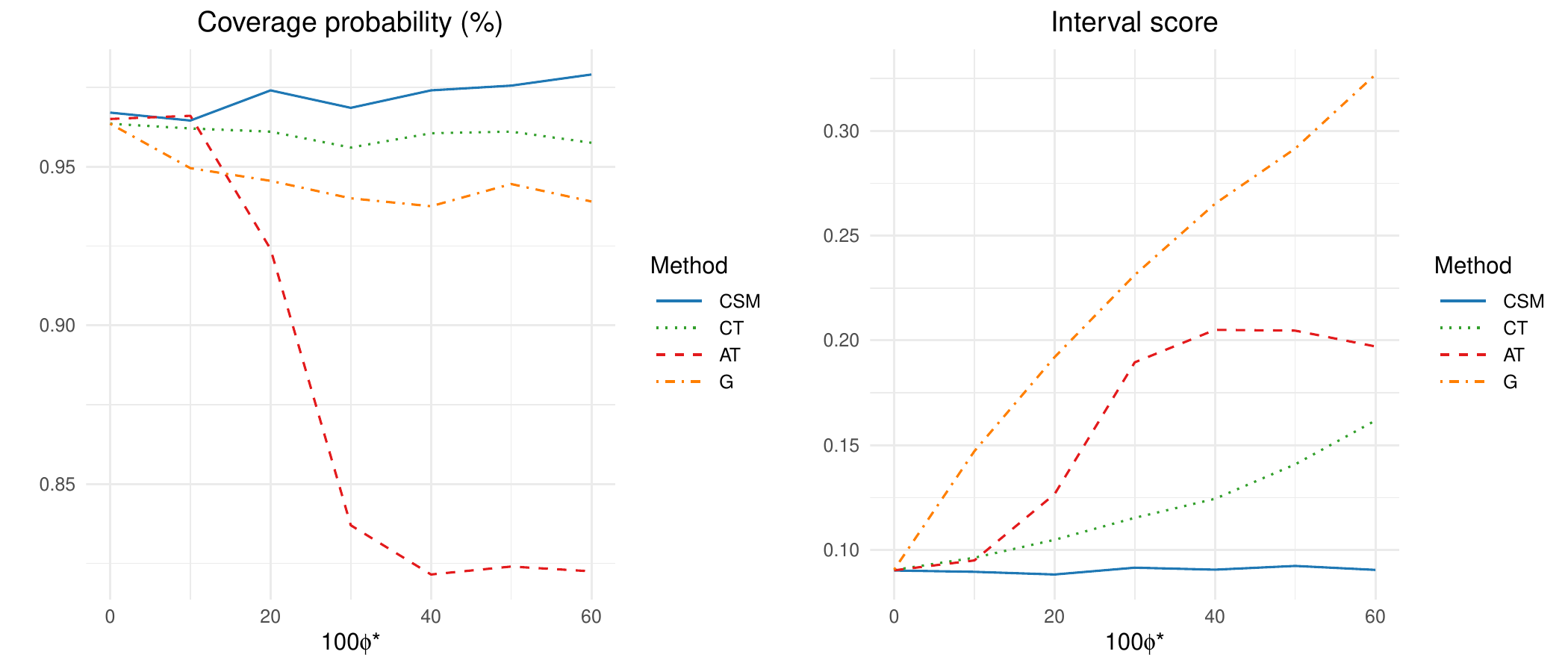}
\caption{ Coverage probability (CP) and interval score (IS) of $95\%$ credible intervals of $\bbe$ (left) and $\bSi$ (right) in multivariate linear regression with $p=10$.} 
\label{fig:MLR-cp}
\end{figure}

\subsection{Example: graphical modeling of gene expression data} \label{sec:gene}

We further demonstrate the usefulness of the CSM model by analyzing gene expression data, which has been used in studies on graphical models (for example, \citealt{finegold2011robust} and \citealt{onizuka2023robust}). 
We used the gene expression profiles of {\it E.~Coli} submitted by \cite{faith2007large}. Following \cite{yamada2014least} and \cite{hirose2017robust}, we selected $p=11$ features and normalized the raw values to construct the dataset of size $n=445$. 
We employed the four versions of the Gaussian graphical model ($\bbe = \zero$) as considered in Section~\ref{sec:sim1}: the CSM model, Bayesian $\gamma$-divergence with $\gamma=0.05$ (GD), alternative $t$-distribution (AT) and standard Gaussian distribution (G). 
We generated 5000 posterior samples after discarding the first 2000 burn-in samples.

Figure~\ref{fig:gene} summarizes the posterior means of precision matrix $\bOm = \bSi^{-1}$ in the form of the undirected graph, where the edge from node $k$ to $k'$ means the posterior mean of the $(k,k')$-entry of $\bOm$ whose $95\%$ credible interval does not cover $0$. 
The G model, which is sensitive to potential outliers, has many estimates that shrink to zero, and shows fewer edges in the graph. More edges become significant as the model becomes more robust; the model with the largest number of edges is GD, followed by CSM and AT. Considering the AT is not posterior robust, and that the GD tends to be ``overly robust,'' it is reasonable that the posterior result of CSM is in between those of these two models.

The CSM model is also useful for approximately evaluating the posterior probabilities of outliers, although this was not the primary objective of our study. 
Such a probability is naturally defined in the CSM as the posterior probability of having $z_{i,k}=1$, whereas we need to define a custom definition of the probability of outliers in the other models. 
For GD, we define the unnormalized weight as $w_i(\bOm)\propto {\rm N}_p(\y_i| \zero, \bOm^{-1})^{\gamma}$, being proportional to the power-discounted likelihood, then normalize them by setting $\sum_{i=1}^n w_i(\bOm)=n$. By construction, the normalized weight should be small if $\y_i$ contains outliers. 
Then, we define the posterior outlier probability as the probability of $w_i(\bOm)<0.01$.
For AT, we compute the posterior probability that the element-wise scale parameter is larger than $2$. 
In Figure~\ref{fig:gene-out}, we present the element-wise posterior outlier probabilities.
As expected from the model form, the CSM yielded almost zero probabilities for most entries, but significantly higher probabilities for several entries. This helps detect outliers in an element-wise manner. 
By contrast, the posterior probability under GD tends to be high for all elements of some columns. Thus the divergence-based posterior analysis could conclude that some vector $\y_i$ is entirely outlying, although it is most likely that many elements of $\y_i$ are non-outlying. 
The posterior probabilities under AT vary across the elements but differ from those of CSM in terms of sparseness. For non-outlying observations, the posterior probabilities under AT are small but not close to zero, leaving more uncertainty about potential outliers and resulting in a less precise parameter estimation. 
By considering elements whose posterior probability is greater than $0.5$, the CSM identified 71 outliers among $np (=4895)$ elements. Out of $n=445$ observational vectors (columns), 56 vectors had only a single outlier and 5 vectors had three outliers. 
The estimated number of outliers was also computed for each feature $k\in\{ 1,\dots, q \}$ and summarized in Table~\ref{tab:gene}.
This result indicates that the occurrence of outliers is not uniform across all features, but varies substantially depending on the dimension. 
According to the obtained result, some features contain multiple outlying elements, whereas others remain almost entirely outlier-free. 
Such heterogeneity highlights the importance of adopting an element-wise approach for outlier detection, as considered in CSM.

\begin{figure}[htbp!]
\centering
\includegraphics[width=0.45\linewidth]{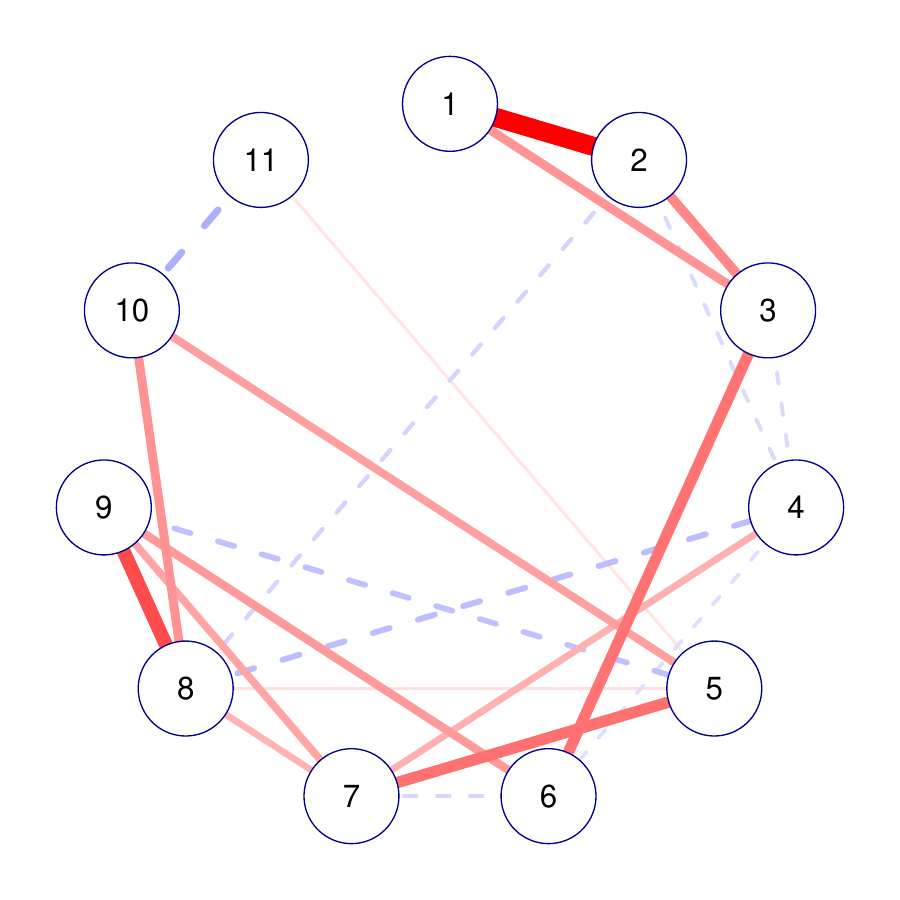}
\includegraphics[width=0.45\linewidth]{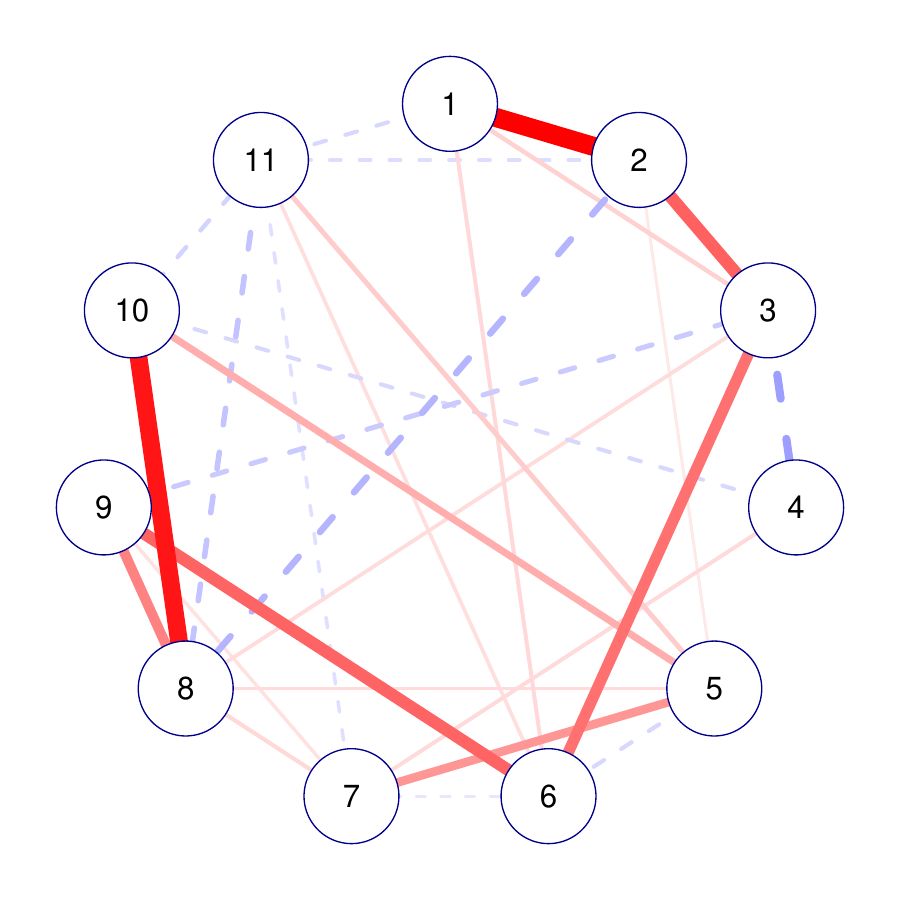}
\includegraphics[width=0.45\linewidth]{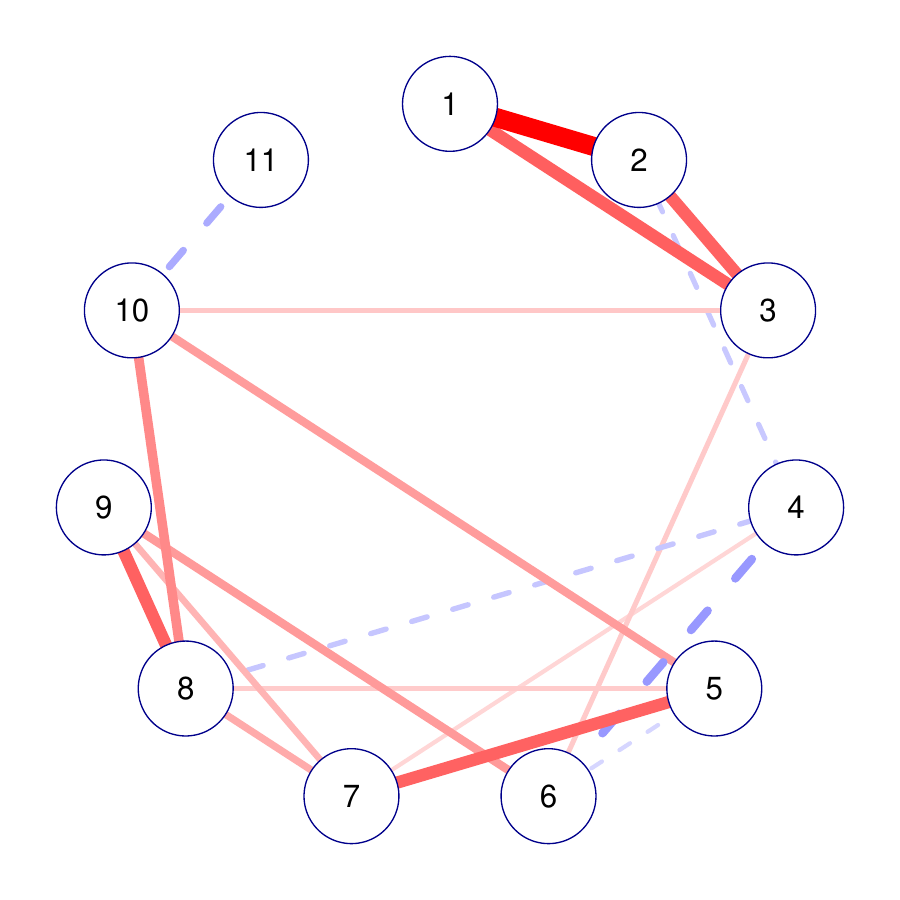}
\includegraphics[width=0.45\linewidth]{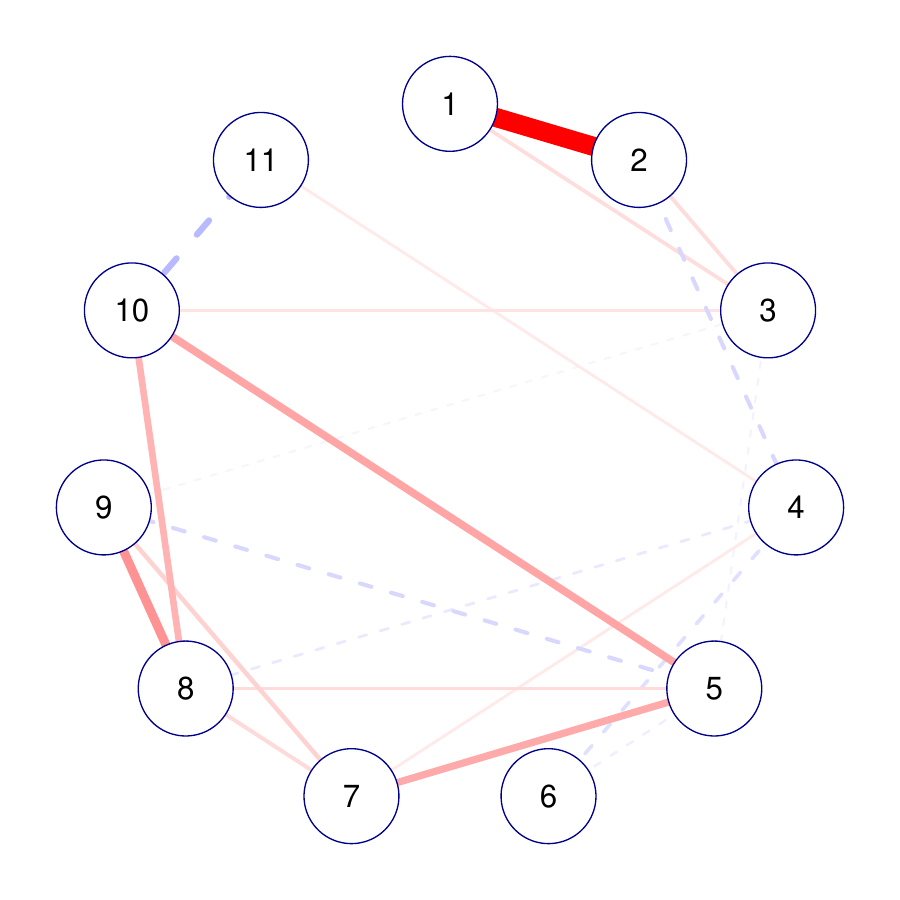}
\caption{ Detected non-null edges obtained from graphical modeling based on the correlation-intact sandwich mixture (CSM; top left), Bayesian $\gamma$-divergence (GD; top right), alternative $t$-distribution (AT; bottom-left) and multivariate normal distribution (G; bottom-right).
The dotted blue and solid red line correspond to positive and negative values, respectively. 
} 
\label{fig:gene}
\end{figure}

\begin{figure}[htbp!]
\centering
\includegraphics[width=\linewidth]{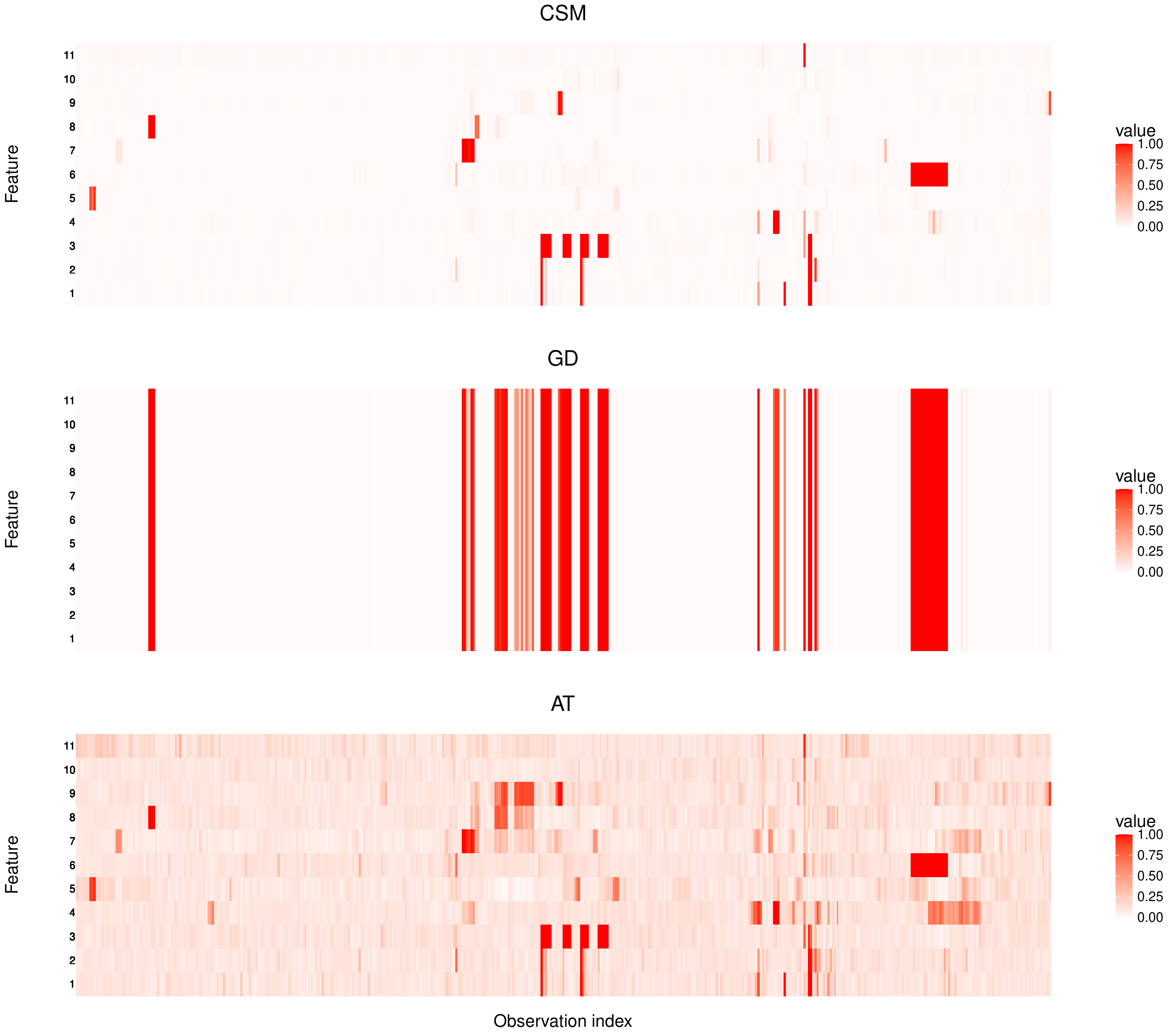}
\caption{ Posterior probability of outliers for each element of $p$-dimensional observations, obtained from CSM, GD and AT models } 
\label{fig:gene-out}
\end{figure}

\begin{table}[htb!]
\caption{The number of outliers in each feature. For each feature $k\in \{ 1,\dots ,p \}$ ($p=11$), we counted the number of $i \in \{ 1, \dots, n \}$ ($n=445$) that satisfies $\mathbb{P}[ z_{i,k} = 1 | \mathrm{data} ] > 0.5$.}
\label{tab:gene}
\centering

\bigskip
\begin{tabular}{ccccccccccccccc}
\hline
Feature & & 1 & 2 & 3 & 4 & 5 & 6 & 7 & 8 & 9 & 10 & 11 & Total\\
\hline
\verb+#+Outliers && 7 & 6 & 20 & 3 & 3 & 17 & 6 & 5 & 3 & 0 & 1 & 71\\
\hline
\end{tabular}
\end{table}

\section{Discussion}
Our theoretical framework does not allow for hierarchical and/or non-Gaussian models, so it would be valuable to extend our theory to posterior robustness of hierarchical and generalized linear models. 
In particular, a univariate hierarchical model considered in \cite{hamura2022log} can be viewed as a multivariate model with sample size $1$, and is partially discussed as the additive error model in Section~\ref{sec:others}. The spectrum of univariate and multivariate models is worth exploring further in terms of posterior robustness. 

To apply the proposed method, we revisited the graphical model in \cite{finegold2011robust} and observed clear differences in the posterior plots. Similar re-examination of analyses in the literature is promising, especially in a field where multivariate $t$ models have been employed, including dynamic stochastic volatility models \citep{choy2014bivariate}.

\section*{Appendix: Simulation from the unfolded log-Pareto distribution} 

To simulate from the unfolded log-Pareto distribution in (\ref{eq:lpdens}), one can utilize the integral representation of the density, 
\begin{equation*}
\pi_{\rm LP} (t;\gamma ) = \int_{0}^{\infty } {1 \over 2} \mathbbm{1}[ |t| > 1 ] {w \over |t|^{1 + w}} {\rm{Ga}} (w | \ga , 1) dw,
\end{equation*}
for all $t \in \mathbb{R} \setminus \{ 0 \} $. That is, the unfolded log-Pareto distribution is the shape mixture of unfolded Pareto distributions. To simulate a random number $t$ from this distribution, generate $w \sim \mathrm{Ga}(\ga,1)$, generate $\upsilon$ from the uniform distribution on $[0,1]$, set $|t| = (1-\upsilon)^{-1/w}$ (this $|t|$ is Pareto distributed), and take its negation with probability $1/2$.

\section*{Acknowledgement}
We thank Jyotishka Datta, Daichi Hiraki and Yasuhiro Omori for their comments. This work is supported by JSPS KAKENHI grant numbers: 25K21163, 22K13374, 24K21420 and 25H00546. 

\vspace{1cm}
\bibliographystyle{chicago}
\bibliography{refs}

\newpage
\setcounter{equation}{0}
\renewcommand{\theequation}{S\arabic{equation}}
\setcounter{section}{0}
\renewcommand{\thesection}{S\arabic{section}}
\setcounter{table}{0}
\renewcommand{\thetable}{S\arabic{table}}
\setcounter{figure}{0}
\renewcommand{\thefigure}{S\arabic{figure}}
\setcounter{lem}{0}
\renewcommand{\thelem}{S\arabic{lem}}

\setcounter{page}{1}

\begin{center}
{\LARGE {\bf 
Supplementary Materials of: \\``Outlier-Robust Bayesian Multivariate Analysis with Correlation-Intact Sandwich Mixture''
}}
\end{center}

\bigskip
In the Supplementary Materials, we provide the proof of the main theorems, derivation of the proposed Markov Chain Monte Carlo (MCMC) algorithm, and collect additional numerical and graphical results. To simplify the proof, we introduce additional notations in Section~\ref{sm:notation}. The proofs of Theorems~1 and 2 are given in Sections~\ref{sm:thm1} and \ref{sm:thm2}, respectively. For mathematical rigors, we provide the proof for the results in Sections~3.3, 3.4 and 3.5 of the main text in Section~\ref{sm:misc}. The improved MCMC algorithm is derived and summarized in Section~\ref{sm:mcmc}. Additional numerical results are listed in Section~\ref{sm:numerical}. The details of Figures 1 and 2 in the main text for reproducibility are found in Section~\ref{sm:fig}.

\section{Notations for the proofs} \label{sm:notation}

For vectors $\x = (x_1,\dots, x_d )^{\top }$ and $\y = (y_1,\dots,y_d)^{\top }$, the element-wise operations are applied as follows: $\x / \y = ( x_1 / y_1 , \dotsc ,x_d / y_d )^{\top }$, $\x ^{\y } = ( {x_1}^{y_1} , \dotsc , {x_d}^{y_d} )^{\top }$, and $|\x| = (|x_1|,\dots ,|x_d|)^{\top }$, while the norm of vector $\x$ is $\| \x \|$, defined by $\| \x \| ^2 = {x_1}^2+\cdots + {x_d}^2$. For matrix $\A$, we denote the Frobenius norm of $\A$ by $\| \A \|^2 = \tr (\A^{\top}\A)$ and, for square $\A$, the determinant by $|\A|$. For the subset of indices, say $\Ic \subset \{ 1, \dots, d \}$, the subvector of $\x$ is written as $\x_{\Ic} = (x_i) _{i\in\Ic}$. Likewise, the submatrix of $d{\times}d$-matrix $\A$ is $\A_{\Ic , \Ic }$. To be precise, for matrix $\A = ( a_{i, j} )_{i, j = 1}^{d}$, if $\Ic = \{ i_1,\dots, i_{d'} \}$ with $d' = |\Ic| \le d$ and $i_1<\cdots <i_{d'}$, then the $(j,k)$-th entry of $\A _{\Ic , \Ic } = ( a_{i, j} )_{i, j \in \Ic }$ is $a_{i_j,i_k}$. We also write the $d$-dimensional all-zero vector as $\bm{0}^{(d)}$, the identity matrix by $\I^{(d)}$ and all-zero square matrix by $\O ^{(d)}$. 
We denote the normal distribution with mean $\bmu$ and variance $\bSi$ by ${\rm{N}}_p  ( \bmu , \bSi )$. The sampling distribution of the proposed model is 
\begin{equation*}
    \y_i | (\t_i,\bbe,\bSi)  \stackrel{\rm ind.}{\sim} {\rm{N}}_p  ( \X _i \bbe , \T_i \bSi \T_i ), \qquad \T_i = \mathrm{diag}(\t_i). 
\end{equation*}
We also denote the density function of ${\rm{N}}_p  ( \bmu , \bSi )$ evaluated at $\y$ by ${\rm{N}}_p  ( \y | \bmu , \bSi )$. Note that this genuinely means the function of argument $(\y , \bmu , \bSi)$; for example, we have ${\rm{N}}_p  ( \y | \bmu , \T \bSi \T ) = {\rm{N}}_p  ( (\y-\bmu)/\t | \bm{0}^{(p)} , \bSi ) |\T|^{-1}$ for vector $\t$ and $\T = \mathrm{diag}(\t)$.

\section{Theorem~1} \label{sm:thm1}

\subsection{Lemmas for the proofs}

We assume that the density of latent variable $t_i$ has the log-Pareto tail. That is, there exist positive constants, $c>0$ and $C>0$, such that 
\begin{equation*}
    \lim _{ |t|\to\infty } \pi (t) | t | \{ \log (1+|t|) \} ^{1+c} = C.
\end{equation*}
We call the density with this property {\it log-Pareto-tailed density}. First, we review the basic formulae about this density, which have been routinely used in the studies of posterior robustness. 

\begin{lem}
\label{lem:rv} 
The log-Pareto-tailed density is regularly varying with order 1; for any non-zero constant $\xi$, 
\begin{equation*}
    \lim _{|t|\to\infty} \frac{ \pi(|t|/\xi ) }{ \pi(|t|){| \xi |}} = 1.
\end{equation*}
\end{lem}

\begin{proof}
    Given non-zero $\xi$, we have 
\begin{equation*}
\begin{split}
    &\lim _{|t|\to\infty} \frac{ \pi(|t|/\xi ) }{ \pi(|t|) |\xi| } \non \\
    &= \lim _{|t|\to\infty} \frac{ \pi(|t|/\xi ) ( |t|/|\xi |) \{ \log(1+|t|/|\xi| ) \} ^{1+c} }{ \pi(|t|) |t| \{ \log(1+|t|) \} ^{1+c} } \times \left\{ \frac{ \log(1+|t|/|\xi| )  }{ \log(1+|t| ) } \right\} ^{-(1+c)} \\
    &= \lim _{|t|\to\infty} \frac{ \pi(|t|/\xi ) ( |t|/|\xi| ) \{ \log(1+|t|/|\xi| ) \} ^{1+c} }{ \pi(|t|) |t| \{ \log(1+|t|) \} ^{1+c} } \times \left\{ \frac{ \log(|t|^{-1} + |\xi|^{-1} ) / \log|t| + 1 }{ \log(1+|t|^{-1} ) /\log|t| + 1 } \right\} 
    ^{-(1+c)} \\
    &= \frac{ C }{ C } \times \left( \frac{ 1 }{ 1 } \right) ^{-(1+c)} \text{,} \non 
\end{split}
\end{equation*}
which is equal to $1$. 
\end{proof}
It is noteworthy that, in proving the limit above, we used the symmetry in the tails of $\pi$. If $\pi$ has different tails, the limit above should depends on the sign of $\xi$. 

Next, for the sake of completeness, we comment on an trivial inequality that has been routinely used in the literature (e.g., Hamura et al. 2022, Lemma S1) 
\begin{lem}
\label{lem:basic} 
For any positive $a$ and $b$, we have
\begin{equation*}
        \frac{1 + \log (1+a)}{ 1 + \log(1+a/b) } \le 1+\log(1+b).
\end{equation*}
\end{lem}
\begin{proof}
We have 
\begin{equation*}
    \begin{split}
        \frac{1 + \log (1+a)}{ 1 + \log(1+a/b) } &\le \frac{1 + \log (1+a+a/b+b)}{ 1 + \log(1+a/b) }  \\
        &= \frac{1 + \log (1+a/b) + \log(1+b)}{ 1 + \log(1+a/b) } %
        = 1 + \frac{\log(1+b)}{ 1 + \log(1+a/b) }  \\
        &\le 1+\log(1+b) \text{,} 
    \end{split}
\end{equation*}
which is the desired result. 
\end{proof}

For the proof of Theorem~1, we additionally provide a property about the density tails of $\pi$. 
\begin{lem}
\label{lem:local} 
Let $\pi $ be a bounded log-Pareto-tailed density on $\mathbb{R}$. Then there exists $M > 0$ such that 
\begin{align}
\pi (|y| / \xi ) (|y| / | \xi|  ) \{ \log (1 + |y|) \} ^{1 + c} &\le M \min \{ \{ \log (1 + |y|) \} ^{1 + c} , \{ 1 + \log (1 + | \xi |) \} ^{1 + c} \} \non 
\end{align}
for all $\xi \neq 0$ and $y \in \mathbb{R}$. 
\end{lem}

\begin{proof}
Let $r_0 > 0$ be a sufficiently large constant. First, consider $|y|/|\xi|  \le r_0$. Since the density is bounded, there is some constant $M_1>0$ such that 
\begin{equation*}
    \pi ( |y|/ \xi ) ( |y|/|\xi| ) \le M_1 ,
\end{equation*}
thus we have 
\begin{equation*}
    \pi (|y| / \xi ) (|y| / |\xi| ) \{ \log (1 + |y|) \} ^{1 + c} \le M_1 \{ \log (1 + |y|) \} ^{1 + c}   , \qquad |y| / |\xi| \le r_0.
\end{equation*}
Next, consider $|y| / |\xi| > r_0$. Since $\pi(t)|t|\{ \log (1+|t|) \}^{1+c} \to C$ as $|t|\to \infty$, with $r_0$ being sufficiently large, there is constant $C'>0$ (which is close to $C$) such that 
\begin{equation*}
    \pi ( |y|/ \xi ) ( |y|/|\xi| ) \{ \log(1+|y|/|\xi|) \}^{1+c} \le C'.
\end{equation*}
Note also that, for $|y|/|\xi| > r_0$, 
\begin{equation*}
    \frac{1+\log(1+|y|/|\xi|)}{\log(1+|y|/|\xi|)} \le \frac{1+\log(1+r_0)}{\log(1+r_0)},
\end{equation*}
since the left hand side above is decreasing in $|y|/|\xi|$. 
Then, by using Lemma~\ref{lem:basic}, we have 
\begin{equation*}
\begin{split}
    &\pi (|y| / \xi ) (|y| / | \xi |) \{ \log (1 + |y|) \} ^{1 + c} \non \\
    &\le C' \left\{  \frac{ \log (1 + |y| ) }{ \log (1 + |y| / |\xi| )  } \right\} ^{1+c} \\
    &= C' \left\{  \frac{ 1 + \log (1 + |y| ) }{ 1 + \log (1 + |y| / |\xi| )  } \right\} ^{1+c} \left\{  \frac{ \log (1 + |y| ) }{ 1 + \log (1 + |y| )  } \right\} ^{1+c} \left\{  \frac{ 1 + \log (1 + |y| / |\xi| ) }{ \log (1 + |y| / |\xi| )  } \right\} ^{1+c}\\
    &\le C' \left\{  \frac{ 1 + \log (1 + |y| ) }{ 1 + \log (1 + |y| / |\xi| )  } \right\} ^{1+c} \times 1 \times \left\{  \frac{ 1 + \log (1 + r_0 ) }{ \log (1 + r_0 )  } \right\} ^{1+c}\\
    &\le M_2 \{ 1+\log(1+|\xi| ) \} ^{1+c}, \qquad |y| / |\xi | > r_0,
\end{split}
\end{equation*}
where $M_2=C' \{ ( 1 + \log (1 + r_0 ) \}^{1+c} / \{ \log (1 + r_0 )  \} ^{1+c}$. Combining the two cases, we obtain the inequality of this lemma, where $M=\max\{ M_1,M_2 \}$. 
\end{proof}

\subsection{Additional notations for Theorem~1}

Before proceeding to the proof of Theorem~1, we introduce notation for mathematical simplicity. We rearrange the elements of vector $\y_i$, so that the first $|\Lc(i)|$ elements are those of outliers, or their indices belong to $\Lc(i)$. Accordingly, we rearrange the rows of $\X$ and $\t_i$, and the rows and columns of $\bSi$. That is, we write 
\begin{equation*}
    \y _i = \begin{bmatrix}
        \y_{i,\Lc(i)} \\ \y_{i,\Kc(i)}
    \end{bmatrix}, \quad \X_i = \begin{bmatrix}
        \X_{i,\Lc(i)} \\ \X_{i,\Kc(i)}
    \end{bmatrix}, \quad \t _i = \begin{bmatrix}
        \t_{i,\Lc(i)} \\ \t_{i,\Kc(i)}
    \end{bmatrix}, \quad \bSi = \begin{bmatrix}
        \bSi _{\Lc(i),\Lc(i)} & \bSi _{\Lc(i),\Kc(i)}   \\
        \bSi _{\Kc(i),\Lc(i)} & \bSi _{\Kc(i),\Kc(i)}   
    \end{bmatrix},
\end{equation*}
where $\y_{i,\Lc(i)}$ is the $|\Lc(i)|$-dimensional vector whose elements diverge in the statement of posterior robustness. The rows of $\X_i$ and $\t_i$, and the rows and columns of $\bSi$, are rearranged accordingly, which defines the subvectors and submatrices above. This simplification can be done by row-switching transformation of normally distributed $\y_i$ and does not lose the generality of the proof. 

In the framework of the theory of posterior robustness, for each $i$ and $k$, there exists a pair of constants $(c_{i,k},d_{i,k})$ such that $y_{i,k}=c_{i,k} + d_{i,k}\om$. In our definition of outliers, $d_{i,k}\not=0$ if and only if $k \in \Lc(i)$. As $\om \to \infty$, we have 
\begin{equation*}
    \frac{y_{i,k}}{|y_{i,k}|} \to \sgn (d_{i,k}) = \begin{cases}
    1 & \mathrm{if} \ d_{i,k} > 0 \\
    -1 & \mathrm{if} \ d_{i,k} < 0 .
    \end{cases}
\end{equation*}
For notatinoal simplicity, we write 
\begin{equation*}
    \sgn ( \d_{i,\Lc(i)} ) = \lim _{\om \to \infty} \frac{ \y _{i,\Lc(i)} }{ | \y _{i,\Lc(i)} | }.
\end{equation*}

\subsection{Proof of Theorem~1}

Fix $i \in \{ 1, \dots , n \}$ and $(\bbe ,\bSi )$. The statement of the robustness of likelihood is as follows; $p( \y _i | \bbe , \bSi ) / C_i(\om) \to p( \y _{i,\Kc(i)} | \bbe , \bSi )$ as $\om \to \infty$ for some constants $C_i(\om)$. If $\Lc(i)=\emptyset$, then $C_i(\om ) =1$ and the likelihood robustness holds trivially. In what follows, we assume that $\Lc(i)\not=\emptyset$ and take the scaling constant as 
\begin{equation*}
    C_i(\om ) = \prod_{k \in \Lc(i)} \pi (| y_{i, k} |) \qquad {\rm if} \quad \Lc(i)\not=\emptyset.
\end{equation*}
Note that, for sufficiently large $\om$ (and $|y_{i,k}|$), density $\pi (| y_{i, k} |)$ is strictly positive since the log-Pareto tails of $\pi(\cdot)$ is assumed, which guarantees $C_i(\om)>0$. With this choice of $C_i(\om)$, the robustness of likelihood is equivalent to 
\begin{equation*}
    p( \y _i | \bbe , \bSi ) / \prod_{k \in \Lc(i)} \pi (| y_{i, k} |) \to p( \y _{i,\Kc(i)} | \bbe , \bSi ) , \qquad {\rm as} \quad \om \to \infty .
\end{equation*}
The likelihood of interest, or the left hand side above, is written as 
\begin{equation*} 
    \int_{\mathbb{R} ^p} {\rm{N}}_p  ( \y _i | \X _i \bbe , \T_i \bSi \T_i ) \Big\{ \prod_{k \in \Lc(i)} {\pi ( t_{i, k} ) \over \pi (| y_{i, k} |)} \Big\} \Big\{ \prod_{k \in \Kc(i)} \pi ( t_{i, k} ) \Big\} d{\t _i} .
\end{equation*}
We first compute its limit by exchanging the limit and integral, then justify this exchange by the dominated convergence theorem. To compute the limit, consider the change of variables $\bxi _{i, \Lc (i)} = |\y _{i, \Lc (i)}| / \t _{i, \Lc (i)}$ with $d\t_{i,\Lc(i)} / \t_{i,\Lc(i)} = - d\bxi _{i,\Lc(i)}/\bxi_{i,\Lc(i)}$. Then, we rewrite the integral of interest as 
\begin{equation}\label{eq:int_t}
\begin{split}
    &p( \y _i | \bbe , \bSi ) / \prod_{k \in \Lc(i)} \pi (| y_{i, k} |)  \\
    &= \int_{\mathbb{R} ^p} {\rm{N}}_p  \Big( \frac{ \y _i - \X _i \bbe }{\t_i} \Big| \bm{0}^{(p)} , \bSi \Big) \Big\{ \prod_{k = 1}^{p} {\pi ( t_{i, k} ) / | t_{i, k} | \over \pi (| y_{i, k} |)} \Big\} \Big\{ \prod_{k = 1}^{p} \pi ( t_{i, k} ) / | t_{i, k} |  \Big\} d{\t _i} \\
    &= \int_{\mathbb{R} ^p}  {\rm{N}}_p  \! \left( \left. \begin{bmatrix}
    \{ \y _{i, \Lc (i)} - \X _{i, \Lc (i)} \bbe \} / \{ |\y _{i, \Lc (i)} | / \bxi _{i, \Lc (i)} \} \\
    \{ \y _{i, \Kc (i)} - \X _{i, \Kc (i)} \bbe \} / \t _{i, \Kc (i)}
    \end{bmatrix} \right| \bm{0}^{(p)} , \bSi \right) \\ 
    &\quad \times \Big\{ \prod_{k \in \Lc (i)} {\pi (| y_{i, k} | / \xi _{i, k} ) \over \pi (| y_{i, k} |) | \xi _{i, k} |} \Big\} \Big\{ \prod_{k \in \Kc (i)} \pi ( t_{i, k} ) / | t_{i, k} | \Big\} d\bxi _{i, \Lc (i)} d\t _{i, \Kc (i)} ,
    \end{split}
\end{equation}
and compute its limit by using Lemma~\ref{lem:rv} as 
\begin{equation*}
\begin{split}
    &p( \y _i | \bbe , \bSi ) / \prod_{k \in \Lc(i)} \pi (| y_{i, k} |) \non \\
    &\to \int_{\mathbb{R} ^p} {\rm{N}}_p  \! \left( \left. \begin{bmatrix}
    \bxi _{i, \Lc (i)} \sgn\{ \d _{i, \Lc (i)} \}  \\
    \{ \y _{i, \Kc (i)} - \X _{i, \Kc (i)} \bbe \} / \t _{i, \Kc (i)}
    \end{bmatrix} \right| \bm{0}^{(p)} , \bSi \right) \Big\{ \prod_{k \in \Kc (i)} \pi ( t_{i, k} ) / | t_{i, k} |  \Big\} d\bxi _{i, \Lc (i)} d\t _{i, \Kc (i)} \\ 
    &= \int_{\mathbb{R} ^{| \Kc (i)|}} {\rm{N}}_{| \Kc (i)|} \Big( {\y _{i, \Kc (i)} - \X _{i, \Kc (i)} \bbe \over \t _{i, \Kc (i)}} \Big| \bm{0} ^{(| \Kc (i)|)} , \bSi _{\Kc (i), \Kc (i)} \Big) \Big\{ \prod_{k \in \Kc (i)} \pi ( t_{i, k} ) / | t_{i, k} | \Big\} d{\t _{i, \Kc (i)}} \\
    &= p( \y _{i, \Kc (i)} | \bbe , \bSi ),
    \end{split}
\end{equation*}
as $\om\to\infty$, or $|y_{i,k}|\to \infty$ for all $k\in\Lc(i)$, which shows the robustness of interest. 

In finding an integrable function that bounds the integrand in (\ref{eq:int_t}), we can assume that $|\y_{i,\Lc(i)}|$ is sufficiently large since we consider the limit of $\om \to \infty$. Note also that there exists $M_1>0$ such that $M_1 \I ^{(p)} \ge \bSi $ (any value which is larger than the eigenvalues of $\bSi$ serves as $M_1$) since $\bSi$ is fixed in this proof. First, we bound the normal density in (\ref{eq:int_t}) for large $|\y_{i,\Lc(i)}|$ as 
\begin{align}
&{\rm{N}}_p  \! \left( \left. \begin{bmatrix}
    \{ \y _{i, \Lc (i)} - \X _{i, \Lc (i)} \bbe \} /  \{ | \y _{i, \Lc (i)} | / \bxi _{i, \Lc (i)} \} \\
    \{ \y _{i, \Kc (i)} - \X _{i, \Kc (i)} \bbe \} / \t _{i, \Kc (i)}
\end{bmatrix} \right| \bm{0}^{(p)} , \bSi \right) , \qquad (\t_{i,\Lc(i)} = |\y _{i, \Lc (i)} | / \bxi _{i, \Lc (i)})\non \\
&\le {1 \over (2 \pi )^{p / 2}} {1 \over | \bSi |^{1 / 2}} \exp \Big[ - {1 \over 2 M_1} \Big\{ \Big\| {\y _{i, \Lc (i)} - \X _{i, \Lc (i)} \bbe \over |\y _{i, \Lc (i)}| / \bxi _{i, \Lc (i)}} \Big\| ^2 + \Big\| {\y _{i, \Kc (i)} - \X _{i, \Kc (i)} \bbe \over \t _{i, \Kc (i)}} \Big\| ^2 \Big\} \Big] \non \\
&\le {1 \over (2 \pi )^{p / 2}} {1 \over | \bSi |^{1 / 2}} \exp \Big[ - {1 \over 2 M_1} \Big\{ {1 \over 2^2} \| \bxi _{i, \Lc (i)} \| ^2 + \Big\| {\y _{i, \Kc (i)} - \X _{i, \Kc (i)} \bbe \over \t _{i, \Kc (i)}} \Big\| ^2 \Big\} \Big] ,\non 
\end{align}
for any non-zero $\bxi _{i, \Lc (i)}$ and $\t _{i, \Kc (i)}$. Next, by Lemma~\ref{lem:local}, there exists some positive constants, $M_2, M_3 > 0$, such that 
\begin{align*}
&\prod_{k \in \Lc (i)} {\pi (| y_{i, k} | / \xi _{i, k} ) \over \pi (| y_{i, k} |) | \xi _{i, k} |} = \prod_{k \in \Lc (i)} {\pi (| y_{i, k} | / \xi _{i, k} ) (| y_{i, k} | / | \xi _{i, k} |) \{ \log (1 + | y_{i, k} |) \} ^{1 + c} \over \pi (| y_{i, k} |) | y_{i, k} | \{ \log (1 + | y_{i, k} |) \} ^{1 + c} } \non \\
&\le \prod_{k \in \Lc (i)} { M_2 \{ 1 + \log (1 + | \xi _{i, k} |) \} ^{1 + c} \over \pi (| y_{i, k} |) | y_{i, k} | \{ \log (1 + | y_{i, k} |) \} ^{1 + c} } \le M_3 \prod_{k \in \Lc (i)} \{ 1 + \log (1 + | \xi _{i, k} |) \} ^{1 + c} \non 
\end{align*}
for any non-zero $\bxi _{i, \Lc (i)}$ and sufficiently large $|y_{i,k}|$. 
Then, a function that dominates the integrand in (\ref{eq:int_t}) is constructed using the two inequalities above as 
\begin{equation*}
    \begin{split}
    &{\rm{N}}_p  \! \left( \left. \begin{bmatrix}
    \{ \y _{i, \Lc (i)} - \X _{i, \Lc (i)} \bbe \} /  \{ | \y _{i, \Lc (i)} | / \bxi _{i, \Lc (i)} \} \\
    \{ \y _{i, \Kc (i)} - \X _{i, \Kc (i)} \bbe \} / \t _{i, \Kc (i)}
    \end{bmatrix} \right| \bm{0}^{(p)} , \bSi \right) \non \\
    &\times \Big\{ \prod_{k \in \Lc (i)} {\pi (| y_{i, k} | / \xi _{i, k} ) \over \pi (| y_{i, k} |) | \xi _{i, k} |} \Big\} \prod_{k \in \Kc (i)} \pi ( t_{i, k} ) / | t_{i, k} |  \\
    &\le {1 \over (2 \pi )^{p / 2}} {1 \over | \bSi |^{1 / 2}} \exp \Big[ - {1 \over 2 M_1} \Big\{ {1 \over 2^2} \| \bxi _{i, \Lc (i)} \| ^2 + \Big\| {\y _{i, \Kc (i)} - \X _{i, \Kc (i)} \bbe \over \t _{i, \Kc (i)}} \Big\| ^2 \Big\} \Big] \\
    &\quad \times M_3 \Big[ \prod_{k \in \Lc (i)} \{ 1 + \log (1 + | \xi _{i, k} |) \} ^{1 + c} \Big] \prod_{k \in \Kc (i)} \pi ( t_{i, k} ) / | t_{i, k} | \\
    &= {2^{|\Lc(i)|}M_1^{p/2}M_3 \over | \bSi  |^{1 / 2}} \ {\rm{N}} ( \bxi _{i, \Lc (i)} | \bm{0}^{(|\Lc(i)|)}, 4M_1\I^{(|\Lc(i)|)} ) \ {\rm{N}} ( \y_{i,\Kc(i)} | \X _{i, \Kc (i)} \bbe, \T_{i,\Kc(i)} \{  M_1\I^{(|\Kc(i)|)} \} \T_{i,\Kc(i)}  ) \\
    &\quad \times \Big[ \prod_{k \in \Lc (i)} \{ 1 + \log (1 + | \xi _{i, k} |) \} ^{1 + c} \Big] \prod_{k \in \Kc (i)} \pi ( t_{i, k} ) ,
    \end{split}
\end{equation*}
for any non-zero $\bxi _{i, \Lc (i)}$ and $\t _{i, \Kc (i)}$, where $\T_{i,\Kc(i)} = \diag (\t_{i,\Kc(i)})$. Note that this dominating function does not depend on $\y_{i,\Lc(i)}$. To see its integrability, observe that the integral of the function above is, without the constant, the product of 
\begin{align}
&\int_{\mathbb{R} ^{|\Lc(i)|}} {\rm{N}} ( \bxi _{i, \Lc (i)} | \bm{0}^{(|\Lc(i)|)}, 4M_1\I^{(|\Lc(i)|)} ) \Big[ \prod_{k \in \Lc (i)} \{ 1 + \log (1 + | \xi _{i, k} |) \}^{1 + c} \Big] d\bxi _{i, \Lc (i)} \non \\
&= \prod_{k \in \Lc (i)} \mathbb{E}[ \{ 1 + \log (1 + | \xi _{i, k} |) \} ^{1 + c} ] \non 
\end{align}
and 
\begin{align}
&\int_{\mathbb{R} ^{|\Kc(i)|}}{\rm{N}} ( \y_{i,\Kc(i)} | \X _{i, \Kc (i)} \bbe, \T_{i,\Kc(i)}\{ M_1\I^{(|\Kc(i)|)} \} \T_{i,\Kc(i)} ) \Big\{ \prod_{k \in \Kc (i)} \pi ( t_{i, k} )  \Big\} d\t _{i, \Kc (i)} \non \\
&= p( \y _{i, \Kc (i)} | \bbe , M_1\I^{(p)} ) \text{,} \non 
\end{align}
where the expectation is taken with respect to $\bxi_i \sim {\rm{N}} ( \bm{0}^{(|\Lc(i)|)}, 4M_1\I^{(|\Lc(i)|)} )$ and bounded by $\mathbb{E}[ (1+|\xi_{i,k}|)^{1+c} ] < \infty$, and $p( \y _{i, \Kc (i)} | \bbe , M_1\I^{(p)} )$ is the sampling density with $\bSi$ being replaced by $M_1\I^{(p)}$. Since $p( \y _{i, \Kc (i)} | \bbe , M_1\I^{(p)} )$ is the proper probability density function, its value is finite for almost all $\y_{i,\Kc(i)}=\bm{c}_{i,\Kc(i)}$ (with respect to the Lebesgue measure). Thus, the integral of the dominating function is shown to be finite, which justifies the exchange of the limit and integral and completes the proof.

\section{Theorem~2} \label{sm:thm2}

\subsection{Lemma for the proof of Theorem~2}

In addition to the lemmas used in the previous proof, we show in advance an inequality about covariance/correlation matrices. 

\begin{lem}
\label{lem:corr} 
For correlation matrix $\R \in \mathbb{R} ^{p \times p}$, we have $\R \le p \I ^{(p)}$.
\end{lem}

\begin{proof}
Let $\a = ( a_k )_{k = 1}^{p} \in \mathbb{R} ^p$ be an arbitrary $p$-vector. Then, observe that 
\begin{align}
{\a }^{\top } \R \a &\le \sum_{k = 1}^{p} \sum_{l = 1}^{p} | a_k || r_{k, l} || a_l | \le \sum_{k = 1}^{p} \sum_{l = 1}^{p} {{a_k}^2 + {a_l}^2 \over 2} = p \sum_{k = 1}^{p} {a_k}^2 = {\a }^{\top } (p \I ^{(p)} ) \a \text{,} \non 
\end{align}
where $r_{k, l}$ is the $(k, l)$ element of $\R $ for $k, l = 1, \dots , p$. 

\end{proof}

\subsection{Notations and model properties for the proof of Theorem~2}

\subsubsection*{Model review}

Unlike Theorem~1, we specify the distribution of latent variable $\t_i$ explicitly in the statement of Theorem~2. Conditional on $\t_i$, the sampling distribution is 
\begin{equation*}
    \y_i | \t_i \sim {\rm{N}}_p  ( \X _i \bbe , \T_i \bSi \T_i ), \qquad \T_i = \mathrm{diag}(\t_i),
\end{equation*}
and, for $\t_i = (t_{i,1},\dots,t_{i,p})^{\top}$, we assume that all the $t_{i,k}$'s are independently and identically distributed with density  
\begin{equation*}
    \pi(t) = (1-\phi) \pi_0(t) + \phi \pi _1(t),
\end{equation*}
where $\phi \in (0,1)$, $\pi_0 (\cdot ) = \delta _{\{1\}}(\cdot)$ is the point-mass distribution on $\{ t=1 \}$, and $\pi_1(\cdot)$ is the log-Pareto tailed. 
For each $t_{i,k}$, we introduce the binary variable $z_{i,k}$ such that 
\begin{equation*}
    \begin{cases}
    t_{i,k} = 1 & {\rm if} \ z_{i,k} = 0 \\
    t_{i,k} \sim \pi _1 & {\rm if} \ z_{i,k} = 1
    \end{cases},
\end{equation*}
which defines the conditional distribution $p(t_{i,k}|z_{i,k})$. In addition, $z_{i,k} \sim \mathrm{Bernoulli}(\phi)$, and $\z_1,\dots,\z_n$ are mutually independent, where $\z_i = (z_{i,1},\dots,z_{i,p})^{\top}$ for all $i$. 

\subsubsection*{Index set $\Zc(i)$}

We also introduce the index set, $\Zc(i)$ and $\overbar{\Zc}(i)$, by
\begin{equation*}
    \begin{split}
        \Zc (i) &= \{ \ k \in \{ 1,\dots ,p\} \ | \ z_{i,k}=1 \ \} \\
        \overbar{\Zc} (i) &= \{ \ k \in \{ 1,\dots ,p\} \ | \ z_{i,k}=0 \ \} = \{ 1,\dots ,p\} \setminus \Zc(i).
    \end{split}
\end{equation*}
That is, the $k$-th element of $\t_i$ is assigned to the log-Pareto-tailed component $\pi_1$ if $k \in \Zc(i)$, or set as $t_{i,k}=1$ if $k \in \overbar{\Zc}(i)$. 

Ideally, for $k\in \Lc(i)$ ($y_{i,k}$ is outlying), we should have $k \in \Zc(i)$ to make latent variable $t_{i,k}$ large to account for that outlier. In reality, however, we must consider the possibility of having $k \in \Lc(i)$ but $k\not\in \Zc(i)$ (or $k\in \overbar{\Zc}(i)$). Reflecting this observation, in the proof, two cases will be considered in the computation of functions of interest:
\begin{itemize}
    \item $\Lc(i) \subset \Zc(i)$ (or $\overbar{\Zc}(i)\cap \Lc(i) = \emptyset$.) 

    \item $\Lc(i) \not\subset \Zc(i)$ (or $\overbar{\Zc}(i)\cap \Lc(i) \not= \emptyset$.) 
\end{itemize}
Related to this notation, let $\z_{1:n} = (\z _1,\dots, \z_n)$, where $\z_i \in \{ 0,1 \}^p$ for each $i$. Then we define $\Zc$ by 
\begin{equation*}
    \Zc = \{ \ (\z_1,\dots ,\z_n) \ | \ \Lc(i) \subset \Zc(i) , \ \mathrm{for \ all \ } i.  \ \}.
\end{equation*}

Conditional on $\z_i$ (or $\Zc(i)$), we divide the vectors into two sub-vectors: those of $\Zc(i)$ and $\overbar{\Zc}(i)$. In doing so, we define the row switching matrix, $\E_i = \E(\z_i)$, by 
\begin{equation*}
    \y _i = \E_i \begin{bmatrix}
        \y_{i,\overbar{\Zc}(i)} \\
        \y_{i,\Zc(i)\cap\Kc(i)} \\
        \y_{i,\Zc(i)\cap\Lc(i)} 
    \end{bmatrix}. 
\end{equation*}
Note that this matrix satisfies $|\E_i|=1$ and $\E_i^{\top} = \E_i$ for any $\z_i\in \{0,1\}^p$. When $\Lc(i) \subset \Zc(i)$, we have $\Zc(i)\cap\Lc(i)=\Lc(i)$ and the third subvector becomes $\y_{i,\Lc(i)}$. With this notation, for example, we can rewrite the part of the normal density of interest as 
\begin{equation*}
    \T_i^{-1} (\y_i-\X_i\bbe) = \E_i \begin{bmatrix}
        \y_{i,\overbar{\Zc}(i)} - \X_{i,\overbar{\Zc}(i)}\bbe  \\
        \{ \y_{i,\Zc(i)\cap\Kc(i)} - \X_{i,\Zc(i)\cap\Kc(i)}\bbe \} / \t_{i,\Zc(i)\cap\Kc(i)}  \\
        \{ \y_{i,\Zc(i)\cap\Lc(i)} - \X_{i,\Zc(i)\cap\Lc(i)}\bbe \} / \t_{i,\Zc(i)\cap\Lc(i)} 
    \end{bmatrix}. 
\end{equation*}

\subsubsection*{Subset $D(\om)$ of the parameter space of $\bbe$ }

The difficulty in proving the robustness of the normalizing constant is the integral over $(\bbe , \bSi)$. In the proof of Theorem~1, we were allowed to fix the value of $(\bbe , \bSi)$, so the dominating function could depend on $(\bbe , \bSi)$ in any way and was constructed relatively easily. In the proof below, however, we must find the dominating function that is integrable over not only $\t_i$'s, but also $(\bbe , \bSi)$. To ease this difficulty, we split the parameter space into two subsets. Specifically, we define the subset of parameter space of $\bbe$ as 
\begin{equation*}
    D( \om ) = \bigcap_{i = 1}^{n} \bigcap_{k \in \Lc (i)} \{  \ \bbe \in \mathbb{R} ^q \  | \  | y_{i, k} - \x _{i, k}^{\top } \bbe | > \ep \om \ \} ,
\end{equation*}
for some sufficiently small $\ep > 0$, which will be given below. Note that, for $k\in \Lc(i)$, we can write $y_{i,k} = c_{i,k}+d_{i,k}\om$ and $d_{i,k}\not=0$. Then, the inequality in $D(\om)$ implies 
\begin{equation*}
    \left| \frac{c_{i,k} - \x _{i, k}^{\top } \bbe}{\om} + d_{i,k}  \right| > \ep.
\end{equation*}
Thus, for each $\bbe$, we have 
\begin{equation*}
    \left| \frac{c_{i,k} - \x _{i,k}^{\top } \bbe}{\om} + d_{i,k}  \right| \ge |d_{i,k}| - \left| \frac{c_{i,k} - \x _{i, k}^{\top } \bbe}{\om} \right| \ge \min _{1\le i \le n} \inf_{k\in\Lc(i)} |d_{i,k}| - \max _{1\le i \le n} 
 \sup _{k\in \Lc(i)} \left| \frac{c_{i,k} - \x _{i, k}^{\top } \bbe}{\om}  \right| ,
\end{equation*}
where the right hand side is independent of $i$ and $k$. This lower bound, or the right hand side, is positive and can be arbitrarily close to $\min _{1\le i \le n} \inf_{k\in\Lc(i)} |d_{i,k}|$ if $\om$ is sufficiently large. Hence, we choose an $\ep$ that satisfies
\begin{equation*}
    \min _{1\le i \le n} \inf_{k\in\Lc(i)} |d_{i,k}| > \ep > 0.
\end{equation*}
With this choice of $\ep$, for any value of $\bbe$, one can take $\om$ sufficiently large, so that $\bbe$ satisfies the inequality above for all $k\in \Lc(i)$ and $i \in \{ 1,\dots ,n \}$. This fact translates into the property of the indicator function,
\begin{equation*}
    \mathbbm{1}_{D(\om)} (\bbe) = \begin{cases}
        1 & \mathrm{if} \ \bbe \in D(\om) \\
        0 & \mathrm{if} \ \bbe \not\in D(\om),
    \end{cases}
\end{equation*}
that the limit is one;
\begin{equation*}
    \lim _{\om \to \infty} \mathbbm{1}_{D(\om)} (\bbe) = 1, \qquad \mathrm{for \ all \ } \bbe \in \mathbb{R}^q. 
\end{equation*}

If $\bbe \not\in D(\om)$, then vector $\bbe$ should not be close to the origin. For such $\bbe$, there exists $i \in \{ 1,\dots ,n\}$ and $k \in \Lc(i)$ such that 
\begin{equation*}
    | y_{i,k}-\x_{i,k}^{\top}\bbe |  \le \ep \om ,
\end{equation*}
where $y_{i,k} = c_{i,k} + d_{i,k}\om$. Since $| y_{i,k} | - | \x_{i,k}^{\top}\bbe | \le | y_{i,k}-\x_{i,k}^{\top}\bbe |$, we have 
\begin{equation*}
  \Big| \frac{c_{i,k}}{\om} + d_{i,k} \Big| - \ep  \le \frac{ | \x_{i,k}^{\top}\bbe | }{ \om } .
\end{equation*}
The left hand side is positive; for sufficiently large $\om$, 
\begin{equation*}
  \Big| \frac{c_{i,k}}{\om} + d_{i,k} \Big| - \ep  \ge |d_{i,k}|-\ep -  \Big| \frac{c_{i,k}}{\om}\Big| > 0,
\end{equation*}
since $|d_{i,k}|-\ep > 0$ by definition. Thus, we can take some constant $\ep _0 > 0$ such that $| c_{i,k} / \om + d_{i,k} | - \ep  \ge \ep_0$. Furthermore, by Cauchy-Schwarz inequality, we have $| \x_{i,k}^{\top}\bbe | \le \| \x_{i,k} \| \| \bbe \|$ and 
\begin{equation*}
    1 \le \left( \frac{1}{\ep_0}\max_{i,k} \|x_{i,k}\| \right) \frac{\|\bbe\|}{\om},
\end{equation*}
for sufficiently large $\om$ and any $\bbe \not\in D(\om)$. That is, if $\bbe \not\in D(\om)$, then $\bbe$ is ``larger than $\om$''. As $\om \to \infty$, the probability of having such $\bbe$ becomes negligible under normal circumstances, as rigorously stated later in the proof. For later use, we rewrite this inequality with arbitrary constant $\kappa > 0$ as 
\begin{equation} \label{eq:be_bound}
    1 \le M_2 \frac{\|\bbe\| ^{\kappa}}{\om ^{\kappa}} , \qquad {\rm where} \qquad M_2 = \frac{1}{\ep_0^{\kappa}}\max_{i,k} \|x_{i,k}\|^{\kappa}.
\end{equation}

\subsubsection*{Likelihoods}

Marginalizing the latent variables out, we have the likelihood as the product of: 
\begin{equation*}
    p( \y_i | \bbe , \bSi ) = \sum _{ \z_i \in \{ 0,1 \} ^p } p_{\z_i} \int _{\mathbb{R}^{|\Zc(i)|}} {\rm N}_{p}(\y_i|\X_i\bbe,\T_i\bSi\T_i) \Big\{ \prod_{k\in\Zc(i)} \pi_1 (t_{i,k}) \Big\} d\t_{i,\Zc(i)},
\end{equation*}
where $p_{\z_i} = \phi ^{|\Zc(i)|} (1-\phi )^{|\overbar{\Zc}(i)|}$ and $t_{i,k}=1$ for $k\in\overbar{\Zc}(i)$ in $\T_i$. The marginal likelihood, or the marginal density of $\y_{1:n}$, is obtained with prior $p_0(\bbe,\bSi)$ as 
\begin{equation*}
\begin{split}
    p( \y_{1:n} ) &= \int \left\{ \prod_{i=1}^n p( \y_i | \bbe , \bSi ) \right\} p_0(\bbe,\bSi) d(\bbe,\bSi) \\
    &= \sum _{ \z_{1:n} \in \{ 0,1 \} ^{np} } \Big( p_{\z} \int \Big[ p_0(\bbe,\bSi) \non \\
    &\quad \times \prod_{i=1}^n \int _{\mathbb{R}^{|\Zc(i)|}} {\rm N}_{p}(\y_i|\X_i\bbe,\T_i\bSi\T_i) \Big\{ \prod_{k\in\Zc(i)} \pi_1 (t_{i,k}) \Big\} d\t_{i,\Zc(i)} \Big] d(\bbe,\bSi) \Big) ,
\end{split}
\end{equation*}
where $p_{\z} = \prod _{i=1}^n p_{\z_i}$. To emphasize the integral, we also write as 
\begin{equation} \label{eq:marginal}
    p( \y_{1:n} ) = \sum _{ \z_{1:n} \in \{ 0,1 \} ^{np} }  \int f( \bbe,\bSi,\z_{1:n},\om) d(\bbe,\bSi) ,
\end{equation}
where
\begin{equation*}
    f( \bbe,\bSi,\z_{1:n},\om) = p_{\z} p_0(\bbe,\bSi) \prod _{i=1}^n \int _{\mathbb{R}^p} {\rm N}_{p}(\y_i|\X_i\bbe,\T_i\bSi\T_i) \Big\{ \prod_{k\in\Zc(i)} \pi_1 (t_{i,k}) \Big\} d\t_{i,\Zc(i)}.
\end{equation*}

\subsection{Proof of Theorem~2}

\subsubsection*{Overview of the proof}

The proof of Theorem~2 is summarized into three steps. 

\begin{itemize}
    \item {\bf Step 1.} Define the scaling constant, $C_i(\om)$, and show the robustness of likelihood. 

    Theorem~1 shows the robustness of likelihood with $C_i(\om ) = \prod_{k\in\Lc(i)}\pi (|y_{i,k}|)$. In the current model, however, $\pi$ is the mixture of two component, one of which is degenerated and does not have a density function with respect to the Lebesgue measure. Hence, we need the custom proof for this model with a different choice of $C_i(\om )$. 

    \item {\bf Step 2.} Compute the limit of the marginal likelihood by restricting the support of the parameters and latent variables to $D(\om )$ and $\Zc$. 

    For the marginal likelihood in equation (\ref{eq:marginal}), we argue that 
    \begin{equation*}
    \begin{split}
        \lim _{\om \to \infty} \sum _{ \z_{1:n} \in \Zc } \int \mathbbm{1}_{D(\om)}(\bbe) f( \bbe,\bSi,\z_{1:n},\om)  d(\bbe,\bSi) = p(\y_{1,\Lc(1)},\dots ,\y_{n,\Lc(n)}).
    \end{split}
    \end{equation*}
    The computation of this limit is straightforward by using the result of Step~1. The exchange of the limit and integral is justified by the dominated convergence theorem. 

    \item {\bf Step 3.} Show that the rest converges to zero. 

    What we need to show is, for each $\z_{1:n} \in \{ 0,1 \} ^{pn}$, 
    \begin{equation*}
    \begin{split}
        &\lim _{\om \to \infty} \int \mathbbm{1}[\bbe \not\in D(\om) \ \mathrm{or} \ \z_{1:n}\not\in\Zc ] f( \bbe,\bSi,\z_{1:n},\om)  d(\bbe,\bSi) \\ 
        =&\lim _{\om \to \infty} \int \{ \ \mathbbm{1}[\bbe \in D(\om) \ \mathrm{and} \ \z_{1:n}\not\in\Zc ] + \mathbbm{1}[\bbe \not\in D(\om) ] \ \} f( \bbe,\bSi,\z_{1:n},\om)  d(\bbe,\bSi) \\
        =& 0.
    \end{split}
    \end{equation*}
    This limit is bounded by some risk function, which is proved to converge to zero.

\end{itemize}

\subsubsection*{Proof of Theorem~2: Step 1}

We first fix $i \in \{ 1,\dots, n\}$ and $(\bbe,\bSi)$, and show $p( \y _i | \bbe , \bSi ) / C_i(\om ) \to p( \y _{i,\Kc(i)} | \bbe , \bSi )$ as $\om \to \infty$ for constant $C_i(\om )$ given by  
\begin{equation*}
    C_i(\om ) = \{ \phi C \} ^{|\Lc(i)|} \prod_{k \in \Lc(i)} | y_{i, k} |^{-1} \{ \log (1 + | y_{i, k} |) \} ^{-(1 + c)} ,
\end{equation*}
where $C = \lim _{t\to\infty} \pi_1(t)|t| ( \log |t|)^{1+c}$ is the tail index of the super heavily tailed component $\pi _1$. Noting that $t_{i,k}=1$ for $k \in \overbar{\Zc}(i)$, and that $t_{i,k} \sim \pi _1(t_{i,k})$ for $k \in \Zc(i)$, we write down the likelihood as 
\begin{equation*}
    \begin{split}
    p( \y_i | \bbe , \bSi ) / C_i(\om) &= \{ \phi C \} ^{-|\Lc(i)|}\sum _{ \z_i \in \{ 0,1 \} ^p } p_{\z_i} \int _{\mathbb{R}^{p}} {\rm N}_{p}(\y_i|\X_i\bbe,\T_i\bSi\T_i) \\
    &\qquad\qquad \times \Big\{ \prod _{k\in\Zc(i)} \pi (t_{i,k}) \Big\} \Big( \prod _{k \in \Lc(i)} | y_{i, k} | \{ \log (1 + | y_{i, k} |) \} ^{1 + c}  \Big) d\t_{i,\Zc(i)} \\
    &= \{ \phi C \} ^{-|\Lc(i)|}\sum _{ \z_i \in \{ 0,1 \} ^p } p_{\z_i} g( \y_i,\z_i ) ,
    \end{split}
\end{equation*}
where 
\begin{equation} \label{eq:gfunc_t}
\begin{split}
    g( \y_i,\z_i ) &= \int_{\mathbb{R} ^{|\Zc(i)|} } {\rm{N}}_p  \! \left( \left. \E_i \begin{bmatrix}
    \y _{i, \overbar{\Zc} (i)} - \X _{i, \overbar{\Zc} (i)} \bbe  \\
    \{ \y _{i, \Zc (i)} - \X _{i, \Zc (i)} \bbe \} / \t _{i, \Zc (i)}
    \end{bmatrix} \right| \bm{0}^{(p)} , \bSi \right) \\
    &\qquad \times \Big\{ \prod_{k \in \Zc (i)} \frac{ \pi _1 ( t_{i, k}  ) }{ |t _{i, k} | } \Big\} \prod_{k \in \Lc(i)} | y_{i, k} | \{ \log (1 + | y_{i, k} |) \} ^{1 + c}  d\t_{i,\Zc(i)}.
\end{split}
\end{equation}
As in the proof of Theorem~1, we set $\bxi _{i,\Zc(i)} = |\y_{i,\Zc(i)}|/\t_{i,\Zc(i)}$ and rewrite $g( \y_i,\z_i )$ as
\begin{equation} \label{eq:gfunc}
\begin{split}
    g( \y_i,\z_i ) &= \int_{\mathbb{R} ^{|\Zc(i)|} } {\rm{N}}_p  \! \left( \left. \E_i \begin{bmatrix}
    \y _{i, \overbar{\Zc} (i)} - \X _{i, \overbar{\Zc} (i)} \bbe  \\
    \{ \y _{i, \Zc (i)} - \X _{i, \Zc (i)} \bbe \} \bxi _{i, \Zc (i)} / |\y _{i, \Zc (i)}|
    \end{bmatrix} \right| \bm{0}^{(p)} , \bSi \right) \\
    &\qquad \times \Big\{ \prod_{k \in \Zc (i)} \frac{ \pi _1 ( |y_{i,k}|/\xi_{i, k}  ) }{ |\xi _{i, k} | } \Big\} \Big( \prod_{k \in \Lc(i)} | y_{i, k} | \{ \log (1 + | y_{i, k} |) \} ^{1 + c} \Big) d\bxi _{i,\Zc(i)}.
\end{split}
\end{equation}
Note that this function also depends on $(\bbe,\bSi)$, which is crucial at Steps~2 and 3, although we stick to this notation for notational convenience. 
To show the robustness of likelihood, we first assume the exchange of the limit and integral (with respect to $\bxi_{i,\Zc(i)}$) and compute the limit, then justify the exchange by the dominated convergence theorem. Finally, for later use, we define the $g$ function without outliers as 
\begin{equation*}
\begin{split}
    &g( \y_{i,\Kc(i)},\z_{i,\Kc(i)} ) \non \\
    &= \int_{\mathbb{R} ^{|\Zc(i)\cap\Kc(i)|} } {\rm{N}}_{|\Kc(i)|}  \! \left( \left. \begin{bmatrix}
    \displaystyle \y _{i, \overbar{\Zc} (i)\cap\Kc(i)} - \X _{i, \overbar{\Zc} (i)\cap\Kc(i)} \bbe  \\
    \displaystyle { \y _{i, \Zc (i)\cap\Kc(i)} - \X _{i, \Zc (i)\cap\Kc(i)} \bbe \over  |\y _{i, \Zc (i)\cap\Kc(i)}| / \bxi _{i, \Zc (i)\cap\Kc(i)} }
    \end{bmatrix} \right| \bm{0}^{(|\Kc(i)|)} , ( \E_i^{\top}\bSi\E_i )_{\Kc(i),\Kc(i)} \right) \\
    &\qquad \times \Big\{ \prod_{k \in \Zc (i)\cap\Kc(i)} \frac{ \pi _1 ( |y_{i,k}|/\xi_{i, k}  ) }{ |\xi _{i, k} | } \Big\} d\bxi _{i,\Zc(i)\cap\Kc(i)}. %
\end{split}
\end{equation*}

The likelihood of interest is the sum over $\z_i$, which is split into two parts: $\Lc(i) \subset \Zc(i)$ and $\Lc(i) \not\subset \Zc(i)$. We show that the former converges to $p( \y _{i,\Kc(i)} | \bbe , \bSi )$, while the latter converges to zero, which means the robustness of the likelihood as a whole. First, we assume $\Lc(i) \subset \Zc(i)$, or $\Lc(i)\cap \Zc(i) = \Lc(i)$, and compute the limit of the likelihood above. Using the change-of-variable, $t_{i,k} = |y_{i,k}| / \xi _{i,k}$ for $k \in \Lc(i)$, letting $\om \to \infty$, then marginalizing $\xi_{i,\Lc(i)}$ out, we have 
\begin{equation*}
    \begin{split}
    &g( \y_i ,\z_i ) \non \\
    &\to \int_{\mathbb{R} ^{|\Zc(i)|}} {\rm{N}}_p  \! \left( \left. \E_i \begin{bmatrix} \y _{i, \overbar{\Zc} (i)} - \X _{i, \overbar{\Zc} (i)} \bbe \\ \{ \y _{i, \Zc(i) \cap \Kc (i)} - \X _{i, \Zc (i) \cap \Kc (i)} \bbe \} \bxi _{i, \Zc(i) \cap \Kc (i)} / |\y _{i, \Zc(i) \cap \Kc (i)}|  \\  \sgn \{ \d _{i, \Lc (i)} \}  \bxi _{i, \Lc (i)} \end{bmatrix} \right| \bm{0}^{(p)} , \bSi \right) \\
    &\qquad \times C^{|\Lc(i)|} \Big\{ \prod_{k \in \Zc (i)\cap \Kc(i)} \frac{ \pi _1 ( |y_{i, k}|/\xi _{i,k} ) }{ |\xi _{i,k}| } \Big\} d\bxi_{i,\Zc(i)} \\ %
    &= C^{|\Lc(i)|} g( \y_{i,\Kc(i)},\z_{i,\Kc(i)} ).
    \end{split}
\end{equation*}

Next, suppose that $\Lc(i) \not\subset \Zc(i)$. That is, there exists $k \in \Lc(i)\cap\overbar{\Zc}(i)$, for which $|y_{i,k}|\to \infty$ and $t_{i,k}=1$. In the functional form of $g(\y_i,\z_i)$, such $y_{i,k}$ is seen in the argument of the normal density, or $\y _{i, \overbar{\Zc} (i)} - \X _{i, \overbar{\Zc} (i)} \bbe$, for which the normal density in $g(\y_i,\z_i)$ converges to zero at the exponential rate. In contrast, $| y_{i, k} | \{ \log (1 + | y_{i, k} |) \} ^{1 + c}$ diverges at the polynomial rate. Thus, the likelihood as a whole converges to zero. That is, 
\begin{equation*}
    g( \y_i,\z_i ) \to 0.
\end{equation*}

Combining these two results, and noting the equivalence between $\Lc(i)\subset\Zc(i)$ and $\z_{i,\Lc(i)} = (1,\dots,1)^{\top}$, we conclude
\begin{equation*}
    \begin{split}
    p( \y_i | \bbe , \bSi ) / C_i(\om) &= \{ \phi C \} ^{-|\Lc(i)|}\sum _{ \z_i \in \{ 0,1 \} ^p } \{ \ \mathbbm{1}[ \Lc(i) \subset \Zc(i) ] + \mathbbm{1}[ \Lc(i) \not\subset \Zc(i) ] \ \} \ p_{\z_i} \ g( \y_i,\z_i ) \\ 
    &\to \{ \phi C \} ^{-|\Lc(i)|}\sum _{ \z_i \in \{ 0,1 \} ^p } \mathbbm{1}[ \Lc(i) \subset \Zc(i) ] \ p_{\z_i}  C^{|\Lc(i)|} g( \y_{i,\Kc(i)},\z_{i,\Kc(i)} ) \\
    &= \phi ^{-|\Lc(i)|} \sum _{ \z_{i,\Kc(i)} \in \{ 0,1 \} ^{|\Kc(i)|} } \phi ^{|\Lc(i)|} p_{\z_{i,\Kc(i)}} g( \y_{i,\Kc(i)},\z_{i,\Kc(i)} ) \\
    &= p( \y_{i,\Kc(i)} | \bbe , \bSi ).
    \end{split}
\end{equation*}

To justify the computation above, we need to find an integrable function that dominates the integrand in (\ref{eq:gfunc}). Let $\bsi^2$ be the vector of the diagonal elements of $\bSi$. Then, by Lemma~\ref{lem:corr}, we have $\T_i\bSi \T_i \le p\cdot \diag(\bsi^2\t_i^2)$ and, for given $\z_i$,  
\begin{equation*}
\begin{split}
    &-\frac{1}{2} (\y_i-\X_i\bbe)^{\top} (\T_i\bSi\T_i)^{-1} (\y_i-\X_i\bbe) \non \\
    &\le -\frac{1}{2p} \left\{ \Big\| {\y _{i, \overbar{\Zc}(i)} - \X _{i, \overbar{\Zc}(i)} \bbe \over \bsi _{\overbar{\Zc}(i)} } \Big\| ^2 + \Big\| {\y _{i, \Zc(i)} - \X _{i, \Zc(i)} \bbe \over \bsi _{\Zc(i)} \t_{i,\Zc(i)} } \Big\| ^2 \right\} \\
    &\le -\frac{1}{2p} \left\{ \Big\| {\y _{i, \overbar{\Zc}(i)\cap\Lc(i)} - \X _{i, \overbar{\Zc}(i)\cap\Lc(i)} \bbe \over \bsi _{\overbar{\Zc}(i)\cap\Lc(i)} } \Big\| ^2 + \Big\| {\y _{i, \Zc(i)} - \X _{i, \Zc(i)} \bbe \over \bsi _{\Zc(i)} |\y_{i,\Zc(i)}| } \cdot \bxi_{i,\Zc(i)} \Big\| ^2 \right\} .
\end{split}
\end{equation*}
To bound the first term above in the exponential scale, for $k \in \overbar{\Zc}(i)\cap\Lc(i)$, we observe that 
\begin{equation*}
    |y_{i,k}| \{ \log ( 1+|y_{i,k}| ) \}^{1+c} \exp \Big\{ -\frac{1}{2p} \frac{(y_{i,k}-\x_{i,k}^{\top}\bbe)^2}{\sigma_k^2} \Big\} \to 0,
\end{equation*}
as $\om \to \infty$. Thus, for any constant $M_1>0$, if $\om$ is sufficiently large, we have
\begin{equation*}
    \Big( \prod _{k \in \overbar{\Zc}(i)\cap\Lc(i)} |y_{i,k}| \{ \log ( 1+|y_{i,k}| ) \}^{1+c}  \Big) \exp \Big\{ -\frac{1}{2p} \Big\| {\y _{i, \overbar{\Zc}(i)\cap\Lc(i)} - \X _{i, \overbar{\Zc}(i)\cap\Lc(i)} \bbe \over \bsi _{\overbar{\Zc}(i)\cap\Lc(i)} } \Big\| ^2 \Big\} \le M_1,
\end{equation*}
noting that $M_1$ is independent of $\om$. 
For the second term, for sufficiently large $\om$, we find the upper bound as 
\begin{equation*}
\begin{split}
    &-\frac{1}{2p} \Big\| {\y _{i, \Zc(i)} - \X _{i, \Zc(i)} \bbe \over \bsi _{\Zc(i)} |\y_{i,\Zc(i)}| } \cdot \bxi_{i,\Zc(i)} \Big\| ^2 \non \\
    &= -\frac{1}{2p} \left\{ \sum _{ k\in \Zc(i)\cap\Lc(i) } { (y _{i, k} - \x _{i,k}^{\top} \bbe )^2 \over \sigma _{k}^2 y_{i,k}^2 } \xi_{i,k}^2 +  \sum _{ k\in \Zc(i)\cap\Kc(i) } { (y _{i, k} - \x _{i,k}^{\top} \bbe )^2 \over \sigma _{k}^2 y_{i,k}^2 } \xi_{i,k}^2 \right\} \\
    &\le -\frac{1}{2p} \left\{ \sum _{ k\in \Zc(i)\cap\Lc(i) } 2^2 { \xi  _{i, k}^2 \over \sigma _k^2 } +  \sum _{ k\in \Zc(i)\cap\Kc(i) } { (y _{i, k} - \x _{i,k}^{\top} \bbe )^2 \over y_{i,k}^2 } { \xi _{i,k}^2 \over \sigma _k^2 } \right\} \\
    &\le -\frac{M_2}{2p} \left\| { \bxi _{i,\Zc(i)} \over \bsi _{\Zc(i)} } \right\|^2, \qquad M_2 = \min\left\{ 2^2 , \min _{ k\in \Zc(i)\cap\Kc(i) } { (y _{i, k} - \x _{i,k}^{\top} \bbe )^2 \over y_{i,k}^2 }\right\} .
\end{split}
\end{equation*}
Finally, by Lemma~\ref{lem:local}, there exists $M_3>0$ such that 
\begin{equation*}
    \prod_{k \in \Zc (i)} \frac{ \pi _1 ( |y_{i,k}|/\xi_{i, k}  )}{ |\xi_{i, k} | } | y_{i, k} | \{ \log (1 + | y_{i, k} |) \} ^{1 + c} \le M_3 \prod_{k \in \Zc (i)} \{ 1 + \log (1 + | \xi _{i, k} | ) \} ^{1 + c}.
\end{equation*}
Using all those inequalities, we bound the integrand in (\ref{eq:gfunc}) as
\begin{equation*}
\begin{split}
    &{\rm{N}}_p  \! \left( \left. \E_i\begin{bmatrix}
    \y _{i, \overbar{\Zc} (i)} - \X _{i, \overbar{\Zc} (i)} \bbe  \\
    \{ \y _{i, \Zc (i)} - \X _{i, \Zc (i)} \bbe \} \bxi_{i,\Zc(i)} / |\y _{i, \Zc (i)}|
    \end{bmatrix} \right| \bm{0}^{(p)} , \bSi \right) \non \\
    &\times \Big( \prod_{k\in\Lc(i)} | y_{i, k} | \{ \log (1 + | y_{i, k} |) \} ^{1 + c} \Big) \prod_{k \in \Zc (i)} \frac{ \pi _1 ( |y_{i, k}|/\xi_{i,k}  ) }{ |\xi _{i, k} | } \\
    &\le \frac{M_1M_3}{(2\pi)^{p/2}|\bSi|^{1/2}}  \exp \left\{ -\frac{M_2}{2p} \left\| { \bxi_{i,\Zc(i)} \over \bsi _{\Zc(i)} } \right\|^2 \right\} \non \\
    &\quad \times \Big[ \prod_{k \in \Zc (i)}  \{ 1 + \log (1 + | \xi_{i, k} | ) \} ^{1 + c} \Big] M_4 \prod_{k \in \Zc(i)\cap\Kc(i)} \frac{1}{ |y_{i,k}| \{ 1 + \log (1 + | y_{i, k} | ) \} ^{1 + c} } \\
    &= M_5 \exp \left\{ -\frac{M_2}{2p} \left\| { \bxi_{i,\Zc(i)} \over \bsi _{\Zc(i)} } \right\|^2 \right\} \prod_{k \in \Zc (i)} \{ 1 + \log (1 + | \xi _{i, k} | ) \} ^{1 + c},
\end{split}
\end{equation*}
for all $\bxi_{i,\Zc(i)}$, with some constants $M_4 , M_5 >0$. The last expression is independent of $\om$. To see its integrability, observe that 
\begin{equation*}
\begin{split}
    &\int  \exp \left\{ -\frac{M_2}{2p} \left\| \bxi _{i,\Zc(i)} \right\|^2 \right\} \Big[ \prod_{k \in \Zc (i)}  \{ 1 + \log (1 + |\xi _{i,k}|) \} ^{1 + c} \Big] d\bxi_{i,\Zc(i)} \non \\
    &= \mathbb{E}\left[ \prod_{k \in \Zc (i)}  \{ 1 + \log (1 + |\xi _{i,k}|) \} ^{1 + c} \right] < \infty \text{,} 
\end{split}
\end{equation*}
where $\bxi_{i,\Zc(i)}$ in the expectation is normally distributed with zero mean and variance $(p/M_1)\I^{( |\Zc(i)| )}$, thus the expectation above is finite. This completes the justification by the dominated convergence theorem.

\subsection*{Proof of Theorem~2: Step 2}

To show the robustness of the normalizing constant in (\ref{eq:marginal}), we express the scaled version of $p(\y_{1:n})$ by using $g(\y_i,\z_i)$ in (\ref{eq:gfunc}) as 
\begin{equation*}
\begin{split}
    \frac{p(\y_{1:n})}{ \prod _{i=1}^n C_i(\om) } &= \int p_0(\bbe,\bSi) \Big\{ \prod _{i=1}^n \frac{p(\y_i|\bbe,\bSi)}{C_i(\om)} \Big\} d(\bbe,\bSi) \\
    &= \int p_0(\bbe,\bSi) \left\{ \prod _{i=1}^n \{ \phi C \} ^{-|\Lc(i)|} \sum _{\z_i \in \{ 0,1 \}^p } p_{\z_i} g(\y_i,\z_i) \right\} d(\bbe,\bSi).
\end{split}
\end{equation*}
Here, we restrict the support of $\bbe$ to $D(\om)$, and that of $\z_{1:n}$ to $\Zc$ (i.e., we assume that $\Lc(i)\subset\Zc(i)$ for all $i$), to compute its limit. As $\om \to \infty$, if the exchange of the limit and integral is allowed, then the result of Step~1 applies directly as
\begin{align}
&\qquad \int \mathbbm{1}[\bbe\in D(\om)] p_0(\bbe,\bSi) \left\{ \prod _{i=1}^n \{ \phi C \} ^{-|\Lc(i)|} \sum _{ \{ \z_i : \Lc(i)\subset \Zc(i) \} } p_{\z_i} g(\y_i,\z_i) \right\} d(\bbe,\bSi) \label{eq:dom_st2} \\
&\to \int p_0(\bbe,\bSi) \left[ \prod _{i=1}^n \Big\{  \sum _{\z_{i,\Kc(i)} \in \{ 0,1 \}^{|\Kc(i)|} } p_{\z_{i,\Kc(i)}} g(\y_{i,\Kc(i)},\z_{i,\Kc(i)}) \Big\} \right] d(\bbe,\bSi) \non \\
&= \int p_0(\bbe,\bSi) \Big\{ \prod _{i=1}^n p(\y_{i,\Kc(i)}|\bbe,\bSi) \Big\} d(\bbe,\bSi) \non \\
&= p(\y_{1,\Kc(1)},\dots, \y_{n,\Kc(n)}) .\non
\end{align}
Below, we justify this exchange. 

Before providing the dominating function, we show several useful inequalities. First, for any given constant $\kappa >0$, the function of $x$ on $[0,\infty)$ defined below is bounded by some constant $M_1 >0$;
\begin{equation} \label{eq:xfunc1}
    \frac{ \{ 1+\log(1+x) \} ^{1+c} }{ 1+ x^{\kappa} } \le M_1,
\end{equation}
because the left hand side is continuous at all $x \in [0,\infty)$ and finite at $x = 0$ and $x \to \infty$. 
Similarly, 
\begin{equation*}
    \exp (- x^{\kappa/2} ) x,
\end{equation*}
is also bounded on $[0,\infty)$ by some constant $M_2>0$. For this reason, for any constant $a>0$, we have 
\begin{equation} \label{eq:xfunc2}
\begin{split}
    \exp (- x^{\kappa/2} ) (1+ax) &\le M_2 (x^{-1} + a) \qquad \mbox{(for any $x>0$)} \\
    &\le M_2(1 + a) \qquad\quad \mbox{(for any $x>1$)}.
\end{split}
\end{equation}
Using the same reasoning, we show that there are some positive constants, $M_3$ and $M_4$, such that 
\begin{equation*}
\begin{split}
    &\prod _{k\in \Lc(i)} \Big[ \frac{ \pi _1(t_{i,k}) }{ |t_{i,k}| } |y_{i,k}| \{ \log (1+|y_{i,k}|) \}^{1+c} \Big] \\
    &\le M_3 \prod _{k\in \Lc(i)} \Big[ \frac{  \{ 1 + \log (1+|y_{i,k}|/|t_{i,k}|) \}^{1+c} }{ |t_{i,k}|^2 } |y_{i,k}| \Big] \qquad \mbox{(by Lemma \ref{lem:local})} \\
    &\le M_4 \ \om ^{|\Lc(i)|} \prod _{k\in \Lc(i)} \frac{  \{ 1 + \log (1+\om/|t_{i,k}|) \}^{1+c} }{ |t_{i,k}|^2 } \qquad \mbox{(by the boundedness as the function of $\om$)} 
\end{split}
\end{equation*}
Second, it is immediate from Lemma~\ref{lem:corr} ($\bSi \le p\cdot \diag ( \bsi^2 )$) that 
\begin{equation} \label{eq:quad}
    - {1 \over 2} \Big( {\y _i - \X _i \bbe \over \t _i} \Big) ^{\top } \bSi ^{- 1} \Big( {\y _i - \X _i \bbe \over \t _i} \Big) \le - {1 \over 4 p} \Big\| {\y _i - \X _i \bbe \over \bsi \t _i} \Big\| ^2 \le - {1 \over 4 p} \Big\| {\y _{i,\Lc(i)} - \X _{i,\Lc(i)} \bbe \over \bsi _{\Lc(i)} \t _{i,\Lc(i)}} \Big\|^2 .
\end{equation}
Since we assume $\bbe \in D(\om)$, by definition, for any $\varepsilon >0$ that satisfies $\varepsilon < |d_{i,k}|$ for all $i\in\{1,\dots,n\}$ and $k\in\Lc(i)$, and for any sufficiently large $\om >0$, we have  
\begin{equation*}
    | y_{i,k} - \x_{i,k}^{\top} \bbe | > \varepsilon \om,
\end{equation*}
for all $i\in\{1,\dots,n\}$ and $k\in\Lc(i)$. Then, the right hand side in (\ref{eq:quad}) is bounded further as 
\begin{equation*}
\begin{split}
    - {1 \over 4 p} \Big\| {\y _{i,\Lc(i)} - \X _{i,\Lc(i)} \bbe \over \bsi_{\Lc(i)} \t _{i,\Lc(i)}} \Big\|^2 &= - {1 \over 4 p} \sum_{k\in\Lc(i)} \left( { y_{i,k} - \x _{i,k}^{\top} \bbe \over \sigma_k t_{i,k} } \right) ^2 \le - {1 \over 4 p} \sum_{k\in\Lc(i)} \left( { \varepsilon\om \over \sigma_k t_{i,k} } \right) ^2.
\end{split}
\end{equation*}
Thus, we have 
\begin{align*}
&\exp \Big\{ - {1 \over 4 p} \Big( {\ep \om \over \si _k t_{i, k}} \Big) ^2 \Big\} \{ 1 + \log (1 + \om / | t_{i, k} |) \} ^{1 + c} & & \\
&\le M_1 \exp \Big\{ - {1 \over 4 p} \Big( {\ep \om \over \si _k t_{i, k}} \Big) ^2 \Big\} \{ 1 + ( \om / | t_{i, k} |)^{\ka } \} &&\mbox{(by Equation (\ref{eq:xfunc1}))} \\
&\le M_1M_2 \Big[ 1 + \Big\{ (4 p)^{\ka / 2} {{\si _k}^{\ka } \over \ep ^{\ka }} \Big\} \Big] &&\mbox{(by Equation (\ref{eq:xfunc2}))}  \\
&\le M_1 M_2 \max \{ 1, (4 p)^{\ka / 2} / \ep ^{\ka } \} (1 + {\si _k}^{\ka } ). 
\end{align*}

With these inequalities in mind, we find a function that dominates the likelihood function. This time, the function must bound the likelihood for all $(\bbe,\bSi)$. 
Below, we use the definition of $g(\y_i,\z_i)$ as the integral of $\t _{i,\Zc(i)}$, given in (\ref{eq:gfunc_t}). 
Furthermore, for notational clarity, we write $t_{i,k}=1$ and $\pi_{z_{i,k}}(t_{i,k})=1$ for $k\in\overbar{\Zc}(i)$ ($z_{i,k}=0$), and $\pi_{z_{i,k}}(t_{i,k})=\pi _1 (t_{i,k})$ for $k\in\Zc(i)$ ($z_{i,k}=1$). 
Then, we fix $\ka _1 > 0$ and evaluate 
\begin{align}
&\int_{\mathbb{R} ^p} {1 \over (2\pi)^{p/2} | \bSi |^{1 / 2}} \exp \Big\{ - {1 \over 2} \Big( {\y _i - \X _i \bbe \over \t _i} \Big) ^{\top } \bSi ^{- 1} \Big( {\y _i - \X _i \bbe \over \t _i} \Big) \Big\} \non \\
&\quad \times \Big[ \prod_{k = 1}^{p} {\pi_{z_{i,k}} ( t_{i, k} ) \over | t_{i, k} |} \Big]  \prod_{k \in \Lc(i)} | y_{i, k} | \{ \log (1 + | y_{i, k} |) \} ^{1 + c} d{\t _i} \non \\
&\le M_5 \om ^{| \Lc (i)|} \int_{\mathbb{R} ^p} {1 \over | \bSi |^{1 / 2}} \exp \Big\{ - {1 \over 4} \Big( {\y _i - \X _i \bbe \over \t _i} \Big) ^{\top } \bSi ^{- 1} \Big( {\y _i - \X _i \bbe \over \t _i} \Big) \Big\} \Big\{ \prod_{k \in  \Kc (i)} {\pi _{z_{i, k}} ( t_{i, k} ) \over | t_{i, k} |} \Big\} \non \\
&\qquad\qquad\qquad \times \exp \Big( - {1 \over 4 p} \Big\| {\y _i - \X _i \bbe \over \bsi \t _i} \Big\| ^2 \Big) \prod_{k \in \Lc (i)} {\{ 1 + \log (1 + | \om | / | t_{i, k} |) \} ^{1 + c} \over | t_{i, k} |^2} d{\t _i} \non \\
&\le M_6 \om ^{| \Lc (i)|} \int_{\mathbb{R} ^p} {1 \over | \bSi |^{1 / 2}} \exp \Big\{ - {1 \over 4} \Big( {\y _i - \X _i \bbe \over \t _i} \Big) ^{\top } \bSi ^{- 1} \Big( {\y _i - \X _i \bbe \over \t _i} \Big) \Big\} \Big\{ \prod_{k \in  \Kc (i)} {\pi _{z_{i, k}} ( t_{i, k} ) \over | t_{i, k} |} \Big\} \non \\
&\qquad\qquad\qquad \times \exp \Big\{ - {1 \over 4 p} \sum_{k \in \Lc (i)} \Big( {\ep \om \over \si _k t_{i, k}} \Big) ^2 \Big\} \prod_{k \in \Lc (i)} {\{ 1 + \log (1 + \om / | t_{i, k} |) \} ^{1 + c} \over | t_{i, k} |^2} d{\t _i} \non \\
&\le M_7 \om ^{| \Lc (i)|} \int_{\mathbb{R} ^p} {1 \over (2 \pi )^{p / 2}} {1 \over 2^{p / 2}} {1 \over | \bSi |^{1 / 2}} \exp \Big\{ - {1 \over 4} \Big( {\y _i - \X _i \bbe \over \t _i} \Big) ^{\top } \bSi ^{- 1} \Big( {\y _i - \X _i \bbe \over \t _i} \Big) \Big\} \non \\
&\qquad\qquad\qquad \times \Big\{ \prod_{k \in \Kc (i)} {\pi _{z_{i, k}} ( t_{i, k} ) \over | t_{i, k} |} \Big\} \Big\{ \prod_{k \in \Lc (i)} {1 + {\si _k}^{\ka _1} \over | t_{i, k} |^2} \Big\} d{\t _i} \non 
\end{align}
for some positive $M_5$, $M_6$, and $M_7$. In the last expression above, we can integrate $\t_{i,\Lc(i)}$ out (by the change of variables from $\t_{i,\Lc(i)}$ to $\{ \y_{i,\Lc(i)} - \X_{i,\Lc(i)}\bbe \} /\t_{i.\Lc(i)}$ with Jacobian $\prod _{k\in\Lc(i)} |y_{i,k} - \x_{i,\Lc(i)}^{\top}\bbe |^{-1}$) to obtain the marginal normal density (kernel) of $\y_{i,\Kc(i)}$. By rearranging the terms, noting that we have assumed $\bbe\in D(\om)$, we have 
\begin{align*}
&g(\y_i,\z_i) \non \\
&\le M_7 \Big\{ \prod_{k \in \Lc (i)} (1 + {\si _k}^{\ka _1} ) \Big\} \Big\{ \prod_{k \in \Lc (i)} {\om \over | y_{i, k} - \x _{i, k}^{\top } \bbe |} \Big\} \int_{\mathbb{R} ^{| \Kc (i)|}} \Big( \Big\{ \prod_{k \in \Kc (i)} {\pi _{z_{i, k}} ( t_{i, k} ) \over | t_{i, k} |} \Big\} \non \\
&\quad \times {1 / (4 \pi )^{|\Kc (i)| / 2} \over | \bSi _{\Kc (i), \Kc (i)} |^{1 / 2}} \exp \Big[ - {1 \over 4} \Big\{ {\y _{i, \Kc (i)} - \X _{i, \Kc (i)} \bbe \over \t _{i, \Kc (i)}} \Big\} ^{\top } \bSi_{\Kc (i), \Kc (i)}^{- 1} \Big\{ {\y _{i, \Kc (i)} - \X _{i, \Kc (i)} \bbe \over \t _{i, \Kc (i)}} \Big\} \Big] \Big) d{\t _{i, \Kc (i)}} \non \\
&\le \frac{ M_7 %
}{\varepsilon^{|\Lc(i)|}} \Big\{ \prod_{k \in \Lc (i)} (1 + {\si _k}^{\ka _1} ) \Big\} \non \\
&\quad \times \int_{\mathbb{R} ^{| \Kc (i)|}} {\rm{N}}_{| \Kc (i)|} ( \y _{i, \Kc (i)} | \X _{i, \Kc (i)} \bbe , 2 \T _{i, \Kc (i)} \bSi_{\Kc (i), \Kc (i)} \T _{i, \Kc (i)} ) \Big\{ \prod_{k \in \Kc (i)} \pi _{z_{i, k}} ( t_{i, k} ) \Big\} d{\t _{i, \Kc (i)}} \\
&\le M_8\Big\{ \prod_{k \in \Lc (i)} (1 + {\si _k}^{\ka _1} ) \Big\} g(\y_{i,\Kc(i)},\z_{i,\Kc(i)}),
\end{align*}
for some $M_8>0$. Finally, the integrand of interest in (\ref{eq:dom_st2}) is bounded as 
\begin{align}
&\mathbbm{1}[\bbe\in D(\om)] p_0(\bbe,\bSi) \prod _{i=1}^n \left\{ \{ \phi C \} ^{-|\Lc(i)|} \sum _{ \{ \z_i : \Lc(i)\subset \Zc(i) \} } p_{\z_i} g(\y_i,\z_i) \right\} \non \\
    &\le M_9 \ \tilde{p}_0(\bbe,\bSi) \prod _{i=1}^n \left\{ \sum _{ \{ \z_{i,\Kc(i)} \in \{0,1\}^p \} } p_{\z_{i,\Kc(i)}} g(\y_{i,\Kc(i)},\z_{i,\Kc(i)}) \right\} , \non 
\end{align}
for some $M_9>0$, where 
\begin{equation*}
    \tilde{p}_0(\bbe,\bSi) \propto \Big\{ \prod _{i=1}^n \prod _{k\in\Lc(i)} (1 + {\si _k}^{\ka _1} ) \Big\} p_0(\bbe,\bSi).
\end{equation*}
To see the integrability of the dominating function provided above, view its integral with respect to $(\bbe,\bSi)$ as the marginal likelihood of the following hierarchical model:
\begin{equation*}
    \begin{split}
        \y_{i,\Kc(i)} | (\bbe,\bSi,\t_{i,\Kc(i)}) &\stackrel{\rm ind.}{\sim} {\rm{N}}_{| \Kc (i)|} ( \X _{i, \Kc (i)} \bbe , 2 \T _{i, \Kc (i)} \bSi_{\Kc (i), \Kc (i)} \T _{i, \Kc (i)} ) \\
        (\bbe,\bSi) &\sim \tilde{p}_0(\bbe,\bSi),
    \end{split}
\end{equation*}
with the independent, two-component mixture model for $t_{i,k}$. If $\tilde{p}_0(\bbe,\bSi)$ is a proper probability density, then the marginal likelihood, or the marginal distribution of $\y_{i,\Kc(i)}$, of this model must be proper, hence its density function must be finite at almost all $\y_{i,\Kc(i)}$ (or, equivalently, almost all $c_{i,k}$ for all $i \in \{1,\dots,n\}$ and $k\in \Kc(i)$), which proves the integrability of the dominating function. For $\tilde{p}_0(\bbe,\bSi)$ to be proper, we need the original prior, $p_0(\bbe,\bSi)$, to have the following moment to be finite (and non-zero): 
\begin{equation*}
    \mathbb{E}\left[ \prod _{i=1}^n \prod _{k\in\Lc(i)} (1 + {\si _k}^{\ka _1} ) \right] < \infty, \qquad {\rm where} \quad (\bbe,\bSi) \sim p_0(\bbe,\bSi).
\end{equation*}
This follows from the condition assumed in the statement of the theorem when $\ka _1$ is sufficiently small compared with $\ka $.

\subsection*{Proof of Theorem~2: Step 3}

The last step of the proof is to compute the part of the density ignored in Step~2. That is, we will show 
\begin{equation} \label{eq:dom_st3}
\int \sum _{  \z_{1:n} \in \{ 0,1 \}^{np} } \mathbbm{1}[ \ \bbe\not\in D(\om) \ {\rm or} \ \z_{1:n} \not\in \Zc \ ] p_0(\bbe,\bSi) \Big\{ \prod _{i=1}^n p_{\z_i} g(\y_i,\z_i) \Big\} d(\bbe,\bSi) \to 0,
\end{equation}
as $\om \to \infty$. For this purpose, it is sufficient to find the upper bound which converges to zero. In doing so, we start by bounding $g(\y_i,\z_i)$. Note that, by Lemma~\ref{lem:corr}, 
\begin{equation*}
    \bSi^{-1} = \frac{1}{2} \bSi^{-1} + \frac{1}{2} \bSi^{-1} \ge \frac{1}{2} \bSi^{-1} + \frac{1}{2} p^{-1} \I ^{(p)},
\end{equation*}
for any $\bSi>0$. Then, the quadratic term in the exponential function in $g(\y_i,\z_i)$ given in (\ref{eq:gfunc}) is bounded as,
\begin{align}
&- {1 \over 2} \begin{pmatrix} \displaystyle \y _{i, \overbar{\Zc}(i)} - \X _{i, \overbar{\Zc}(i)} \bbe \\ \displaystyle {\y _{i, \Zc(i)} - \X _{i, \Zc(i)} \bbe \over |\y _{i, \Zc(i)}| / \bxi _{i, \Zc(i)}} \end{pmatrix} ^{\top } \E_i^{\top } \bSi ^{- 1} \E_i \begin{pmatrix} \displaystyle \y _{i, \overbar{\Zc}(i)} - \X _{i, \overbar{\Zc}(i)} \bbe \\ \displaystyle {\y _{i, \Zc(i)} - \X _{i, \Zc(i)} \bbe \over | \y _{i, \Zc(i)} | / \bxi _{i, \Zc(i)}} \end{pmatrix} \non \\
&\le - {1 \over 4} \begin{pmatrix} \displaystyle \y _{i, \overbar{\Zc}(i)} - \X _{i, \overbar{\Zc}(i)} \bbe \\ \displaystyle {\y _{i, \Zc(i)} - \X _{i, \Zc(i)} \bbe \over |\y _{i, \Zc(i)}| / \bxi _{i, \Zc(i)}} \end{pmatrix} ^{\top } \E_i^{\top } \bSi ^{- 1} \E_i \begin{pmatrix} \displaystyle \y _{i, \overbar{\Zc}(i)} - \X _{i, \overbar{\Zc}(i)} \bbe \\ \displaystyle {\y _{i, \Zc(i)} - \X _{i, \Zc(i)} \bbe \over | \y _{i, \Zc(i)} | / \bxi _{i, \Zc(i)}} \end{pmatrix} \non \\
&\quad - {1 \over 4 p} \Big\| {\y _{i, \overbar{\Zc}(i) \cap \Lc (i)} - \X _{i, \overbar{\Zc}(i) \cap \Lc (i)} \bbe \over \bsi _{\overbar{\Zc}(i) \cap \Lc (i)}} \Big\| ^2 . \label{eq:decomp_st3}
\end{align}
The first term serves as the likelihood in the dominating function with $\bSi$ being replaced with $2\bSi$, while the second term has $\om$ and contribute to the convergence of the whole dominating function to zero, as detailed in what follows. 

Note that the indicator function that restricts $\bbe$ and $\z_{1:n}$ is decomposed as 
\begin{equation*}
\mathbbm{1}[ \ \bbe\not\in D(\om) \ {\rm or} \ \z_{1:n} \not\in \Zc \ ] = \mathbbm{1}[ \ \bbe\in D(\om) \ {\rm and} \ \z_{1:n} \not\in \Zc \ ] + \mathbbm{1}[ \ \bbe\not\in D(\om) \ ] 
\end{equation*}
First, assume that $\bbe\in D(\om)$ and $\z_{1:n} \not\in \Zc$. That is, for all $i$, we have $\Lc(i) \not\subset \Zc(i)$ or, equivalently, $\overbar{\Zc}(i)\cap\Lc(i)\not=\emptyset$. 
Then, the second term in (\ref{eq:decomp_st3}) is bounded in the exponential scale as 
\begin{align}
&\exp \Big\{ - {1 \over 4 p} \Big\| {\y _{i, \overbar{\Zc}(i) \cap \Lc (i)} - \X _{i, \overbar{\Zc}(i) \cap \Lc (i)} \bbe \over \bsi _{\overbar{\Zc}(i) \cap \Lc (i)}} \Big\| ^2 \Big\} \non \\ 
\le& \exp \Big\{ - {\ep ^2 \om ^2 \over 4 p} \Big\| {1 \over \bsi _{\overbar{\Zc}(i)\cap \Lc (i)}} \Big\| ^2 \Big\} \non \\
=& \Big\{ {\om ^2 \over \om ^2} \cdot  { \| 1/\bsi _{\overbar{\Zc}(i) \cap \Lc (i)} \| ^2 \over \| 1/\bsi _{\overbar{\Zc}(i) \cap \Lc (i)} \|^2  }  \Big\} ^{\ka / 2}  \exp \Big\{ - {\ep ^2 \om ^2 \over 4 p} \Big\| {1 \over \bsi _{\overbar{\Zc}(i) \cap \Lc (i)}} \Big\| ^2 \Big\} \non \\
=& \Big\{ {1 \over \om ^2} / \Big\| {1 \over \bsi _{\overbar{\Zc}(i) \cap \Lc (i)}} \Big\| ^2 \Big\} ^{\ka / 2} \Big\{ \om ^2 \Big\| {1 \over \bsi _{\overbar{\Zc}(i) \cap \Lc (i)}} \Big\| ^2 \Big\} ^{\ka / 2} \exp \Big\{ - {\ep ^2 \om ^2 \over 4 p} \Big\| {1 \over \bsi _{\overbar{\Zc}(i) \cap \Lc (i)}} \Big\| ^2 \Big\} \label{eq:rest_st3}.
\end{align}
In the last expression, the first term is bounded as, by taking arbitrary $k_0 \in \overbar{\Zc}(i) \cap \Lc(i)$,
\begin{equation*}
    \Big\{ {1 \over \om ^2} / \Big\| {1 \over \bsi _{\overbar{\Zc}(i) \cap \Lc (i)}} \Big\| ^2 \Big\} ^{\ka / 2} \le  {1 \over \om ^{\ka }} \Big( \frac{1}{ \sigma_{i,k_0}^2 } \Big)^{- \ka / 2} \le {1 \over \om ^{\ka }} \sum _{k=1}^p \sigma _{i,k}^{\ka }.
\end{equation*}
The rest of the expression in (\ref{eq:rest_st3}) is also bounded by some constant $M_1>0$. To see this, observe that the function defined as 
\begin{equation*}
    f(x) = x^{\ka / 2} \exp (-c_0 x) , \qquad x>0,
\end{equation*}
for some constant $c_0>0$ is maximized at $x = \ka / (2 c_0 )$. 
Use this result with $x=\om^2\|1/\bsi _{\overbar{\Zc}(i) \cap \Lc (i)} \|^2$ and $c_0 = \ep ^2 / (4 p)$ and conclude that 
\begin{equation*}
    \Big\{ \om ^2 \Big\| {1 \over \bsi _{\overbar{\Zc}(i) \cap \Lc (i)}} \Big\| ^2 \Big\} ^{\ka / 2} \exp \Big\{ - {\ep ^2 \om ^2 \over 4 p} \Big\| {1 \over \bsi _{\overbar{\Zc}(i) \cap \Lc (i)}} \Big\| ^2 \Big\} \le \left( \frac{2p \ka }{\ep^2} \right)^{\ka / 2} \exp (- \ka / 2) . 
\end{equation*}
Writing the constant in the right hand side by $M_1 = ( 2p \ka / \ep^2 )^{\ka / 2} \exp (- \ka / 2)$, we have 
\begin{equation*}
\mathbbm{1}[ \ \bbe\in D(\om) \ {\rm and} \ \z_{1:n} \not\in \Zc \ ] \exp \Big\{ - {1 \over 4 p} \Big\| {\y _{i, \overbar{\Zc}(i) \cap \Lc (i)} - \X _{i, \overbar{\Zc}(i) \cap \Lc (i)} \bbe \over \bsi _{\overbar{\Zc}(i) \cap \Lc (i)}} \Big\| ^2 \Big\} \le M_1 {1 \over \om ^{\ka }} \sum _{k=1}^p \sigma _{i,k}^{\ka }.    
\end{equation*}
Next, assume $\bbe \notin D( \om )$. Then, simply by the property of $D(\om)$ given in (\ref{eq:be_bound}), we have 
\begin{equation*}
\mathbbm{1}[ \ \bbe\not\in D(\om) \ ] \prod_{i=1}^n \exp \Big\{ - {1 \over 4 p} \Big\| {\y _{i, \overbar{\Zc}(i) \cap \Lc (i)} - \X _{i, \overbar{\Zc}(i) \cap \Lc (i)} \bbe \over \bsi _{\overbar{\Zc}(i) \cap \Lc (i)}} \Big\| ^2 \Big\} \le M_2 \frac{ \| \bbe \|^{\ka } }{\om ^{\ka }} .
\end{equation*}
for some $M_2 > 0$; the exponential terms in the left hand side are bounded by one. 

Now, combining these results, we bound the integral of interest in (\ref{eq:dom_st3}) as   
\begin{equation} \label{eq:bound_st3}
\begin{split}
&\int \sum _{  \z_{1:n} \in \{ 0,1 \}^{np} } \mathbbm{1}[ \ \bbe\not\in D(\om) \ {\rm or} \ \z_{1:n} \not\in \Zc \ ] p_0(\bbe,\bSi) \Big\{ \prod _{i=1}^n p_{\z_i} g(\y_i,\z_i) \Big\} d(\bbe,\bSi) \\
&\le M_1 M_2 \int \prod _{i=1}^n \sum _{\z_i \in \{ 0,1 \}^{p} } p_{\z_i}  \\
&\qquad\qquad\qquad \times \int_{\mathbb{R} ^{|\Zc(i)|} } 2^{p/2} {\rm{N}}_p  \! \left( \left. \E_i \begin{bmatrix}
    \y _{i, \overbar{\Zc} (i)} - \X _{i, \overbar{\Zc} (i)} \bbe  \\
    \{ \y _{i, \Zc (i)} - \X _{i, \Zc (i)} \bbe \} \bxi _{i, \Zc (i)} / |\y _{i, \Zc (i)}|
    \end{bmatrix} \right| \bm{0}^{(p)} , 2\bSi \right) \\
&\qquad\qquad\qquad \times \Big\{ \prod_{k \in \Zc (i)} \frac{ \pi _1 ( |y_{i,k}|/\xi_{i, k}  ) }{ |\xi _{i, k} | } \Big\} \prod_{k \in \Lc(i)} | y_{i, k} | \{ \log (1 + | y_{i, k} |) \} ^{1 + c} d\bxi _{i,\Zc(i)} \\
&\qquad\qquad \times \left( \frac{ \| \bbe \|^{\ka } }{\om ^{\ka }} + {1 \over \om ^{\ka }} \sum _{k=1}^p \sigma _{i,k}^{\ka } \right) p_0(\bbe,\bSi) d(\bbe,\bSi) \\
&\le M_3 \frac{1}{\om ^{\ka }} \int \Big\{ \prod _{i=1}^n \sum _{  \z_i \in \{ 0,1 \}^{p} } p_{\z_i} g_2(\y_i,\z_i) \Big\} \tilde{p}_1(\bbe,\bSi) d(\bbe,\bSi) \\
\end{split}
\end{equation}
for some $M_3>0$, where $g_2(\y_i,\z_i)$ is the same as $g(\y_i,\z_i)$ except that $\bSi$ in the original $g(\y_i,\z_i)$ is replaced with $2\bSi$, and $\tilde{p}_1(\bbe,\bSi)$ is given by 
\begin{equation*}
    \tilde{p}_1(\bbe,\bSi) \propto (1 + \| \bbe \| ^{\ka } ) \Big( 1 + \sum_{k = 1}^{p} {\si _k}^{\ka } \Big) p_0(\bbe,\bSi), 
\end{equation*}
which is a proper probability density by assumption. %
The bound we obtained here is the product of $1/\om ^{\ka }$ and the marginal likelihood of the proposed model with variance $2\bSi$ and prior $\tilde{p}_1(\bbe,\bSi)$. The latter could depend on $\om$, but, as will be verified in the following, it is bounded above by some constant, thus we can conclude that this bound converges to zero as a whole. Below, we compute $g_2(\y_i,\z_i)$ as the function of $\om$. 

We bound $g_2(\y_i,\z_i)$ further as, for sufficiently large $\om$,  
\begin{equation*} 
\begin{split}
&g_2(\y_i,\z_i) \\
&\le M_4 \int_{\mathbb{R} ^{| \Zc (i) |}} {1 \over | \bSi |^{1 / 2}} \exp \Big\{ - {1 \over 4} \begin{pmatrix} \displaystyle \y _{i, \overbar{\Zc}(i)} - \X _{i, \overbar{\Zc}(i)} \bbe \\ \displaystyle {\y _{i, \Zc (i)} - \X _{i, \Zc (i)} \bbe \over | \y _{i, \Zc (i)} | / \bxi _{i, \Zc (i)}} \end{pmatrix} ^{\top } \E_i^{\top } \bSi ^{- 1} \E_i \begin{pmatrix} \displaystyle \y _{i, \overbar{\Zc}(i)} - \X _{i, \overbar{\Zc}(i)} \bbe \\ \displaystyle {\y _{i, \Zc (i)} - \X _{i, \Zc (i)} \bbe \over | \y _{i, \Zc (i)} | / \bxi _{i, \Zc (i)}} \end{pmatrix} \Big\} \\
&\quad \times \{ \om ( \log \om )^{1 + c} \} ^{| \Lc (i)|} \prod_{k \in \Zc (i)} \frac{ \pi _1 (| y_{i, k} | / \xi _{i, k} ) }{ | \xi _{i, k} | }  d{\bxi _{i, \Zc (i)}} \\
&= M_4  ( \log \om )^{(1+c)| \Lc (i)|}  \int_{\mathbb{R} ^{| \Zc (i) |}} h ( \t _{i, \Zc (i)} , \z_i ; \bbe , \bSi , \om ) \Big\{ \prod_{k \in \Zc (i)} \pi _1 ( t_{i, k} ) \Big\} d{\t _{i, \Zc (i)}} 
\end{split}
\end{equation*}
for some $M_4 > 0$, where 
\begin{align}
&h ( \t _{i, \Zc (i)} , \z_i ; \bbe , \bSi , \om ) \non \\
&= \exp \Big[ - {1 \over 4} \begin{pmatrix} \displaystyle \y _{i, \overbar{\Zc} (i)} - \X _{i, \overbar{\Zc} (i)} \bbe \\ \displaystyle { \{ \y _{i, \Zc (i)} - \X _{i, \Zc (i)} \bbe \} / \t _{i, \Zc (i)}} \end{pmatrix} ^{\top } \E_i^{\top } \bSi ^{- 1} \E_i \begin{pmatrix} \displaystyle \y _{i, \overbar{\Zc} (i)} - \X _{i, \overbar{\Zc} (i)} \bbe \\ \displaystyle { \{ \y _{i, \Zc (i)} - \X _{i, \Zc (i)} \bbe \} / \t _{i, \Zc (i)}} \end{pmatrix} \Big] \non \\
&\quad \times {\om ^{| \Lc (i)|} \over | \bSi |^{1 / 2}} \times \prod_{k \in \Zc (i)} { 1 \over | t_{i, k} |}, \non 
\end{align}
for $\t _{i, \Zc (i)} \in \mathbb{R} ^{| \Zc (i) |}$. Using this notation, we can upper bound the last expression of (\ref{eq:bound_st3}) as, with some constant $M_4>0$, 
\begin{equation*}
\begin{split}
    &M_4 \frac{ ( \log \om )^{ (1+c)\sum _{i=1}^n | \Lc (i)|}}{\om ^{\ka / 3}} \sum _{  \z_{1:n} \in \{ 0,1 \}^{np} } \Big( p_{\z_{1:n}} \\
    &\times {1 \over \om ^{2 \ka / 3}} \int \Big[ \prod _{i=1}^n \int_{\mathbb{R} ^{| \Zc (i) |}} h ( \t _{i, \Zc (i)} , \z_i ; \bbe , \bSi , \om ) \Big\{ \prod_{k \in \Zc (i)} \pi _1 ( t_{i, k} ) \Big\} d{\t _{i, \Zc (i)}} \Big] \tilde{p}_1(\bbe,\bSi) d(\bbe,\bSi) \Big) ,
\end{split}
\end{equation*}
Since the first line above converges to zero, it is sufficient to show that the second line is at least bounded in the limit. That is, we prove that 
\begin{align}
&R_{\Z } ( \om ) \non \\
&\equiv {1 \over \om ^{\ka / 3}} \int \Big[ \prod _{i=1}^n \int_{\mathbb{R} ^{| \Zc (i) |}} h ( \t _{i, \Zc (i)} , \z_i ; \bbe , \bSi , \om ) \Big\{ \prod_{k \in \Zc (i)} \pi _1 ( t_{i, k} ) \Big\} d{\t _{i, \Zc (i)}} \Big] \tilde{p}_1(\bbe,\bSi) d(\bbe,\bSi) \non \\
&= o( \om ^{\ka / 3} ) , \non
\end{align}
as $\om \to \infty $. In showing this equation below, we note that $R_{\Z } ( \om )$ is viewed as a risk function, or the expectation of a certain loss function with respect to $(\bbe,\bSi,\t_{1:n,\Zc(1:n)})$ where $\t_{1:n,\Zc(1:n)} = (\t_{1,\Zc(1)},\dots,\t_{n,\Zc(n)})$. To be precise, we write 
\begin{equation*}
    R_{\Z } ( \om ) = \mathbb{E}\left[ \ L( \t _{1:n, \Zc (1:n)} , \z_{1:n} ; \bbe , \bSi , \om ) \ \right] ,
\end{equation*}
where $(\bbe,\bSi)$ and $\t_{1:n,\Zc(1:n)}$ are mutually independent, $(\bbe,\bSi) \sim \tilde{p}_1(\bbe,\bSi)$ and $\t_{i,k} \stackrel{\rm ind.}{\sim} \pi_1(t_{i,k})$ for $i\in \{ 1,\dots, n\}$ and $k \in \Zc(i)$. The loss function, $L$, is defined as (with $\om$ being replaced with general argument variable $\tau \in [1,\infty)$), 
\begin{align}
&L( \t _{1:n, \Zc (1:n)} , \z_{1:n} ; \bbe , \bSi , \tau ) \non \\
&= {1 \over \tau ^{\ka / 3}} \prod _{i=1}^n h ( \t _{i, \Zc (i)} , \z_i ; \bbe , \bSi , \tau ) \non \\
&= {1 \over \ta ^{\ka / 3}} \prod_{i = 1}^{n} \Big( \Big\{ \prod_{k \in \Zc(i)} {1 \over | t_{i, k} |} \Big\} {\ta ^{| \Lc (i)|} \over | \bSi |^{1 / 2}} \non \\
&\quad \times \exp \Big[ - {1 \over 4} \begin{pmatrix} \displaystyle \y _{i, \overbar{\Zc} (i)} - \X _{i, \overbar{\Zc} (i)} \bbe \\ \displaystyle { \{ \y _{i, \Zc (i)} - \X _{i, \Zc (i)} \bbe \} / \t _{i, \Zc (i)}} \end{pmatrix} ^{\top } \E_i^{\top } \bSi ^{- 1} \E_i \begin{pmatrix} \displaystyle \y _{i, \overbar{\Zc} (i)} - \X _{i, \overbar{\Zc} (i)} \bbe \\ \displaystyle { \{ \y _{i, \Zc (i)} - \X _{i, \Zc (i)} \bbe \} / \t _{i, \Zc (i)}} \end{pmatrix} \Big] \Big) \non \\
&= {1 \over \ta ^{\ka / 3}} \prod_{i = 1}^{n} \Big[ \Big\{ \prod_{k \in \Zc(i)} {1 \over | t_{i, k} |} \Big\} {\ta ^{| \Lc (i)|} \over | \bSi |^{1 / 2}} \non \\
&\quad \times \exp \Big\{ - {1 \over 4} \Big( {\ta \d _i + \c _i - \X _i \bbe \over \t _i} \Big) ^{\top } \bSi ^{- 1} \Big( {\ta \d _i + \c _i - \X _i \bbe \over \t _i} \Big) \Big\} \Big] \non ,
\end{align}
where $y_{i,k}$ is defined as $y_{i,k}=c_{i,k} + \tau d_{i,k}$ for $k\in \Lc(i)$ in the second line (or the original $y_{i,k}$ being evaluated at $\om = \tau$) and $t_{i,k} = 1$ for $k\in\overbar{\Zc}(i)$ in the third line.

We bound $R_{\Z } ( \om )$ in the spirit of the integral expression of risk difference method (Kubokawa, 1994; Kubokawa and Saleh, 1994), 
by writing $L( \t _{1:n, \Zc (1:n)} , \z_{1:n} ; \bbe , \bSi , \om )$ as 
\begin{align}
&L( \t _{1:n, \Zc (1:n)} , \z_{1:n} ; \bbe , \bSi , \om )  - L( \t _{1:n, \Zc (1:n)} , \z_{1:n} ; \bbe , \bSi , 1 )  \non \\
&= \int_{1}^{\om } {\pd L( \t _{1:n, \Zc (1:n)} , \z_{1:n} ; \bbe , \bSi , \tau )  \over \pd \ta } d\ta \non \\
&= \int_{1}^{\om } \Big( L( \t _{1:n, \Zc (1:n)} , \z_{1:n} ; \bbe , \bSi , \tau ) \non \\
&\quad \times \Big[ - {\ka / 3 \over \ta } + \sum_{i = 1}^{n} \Big\{ {| \Lc (i)| \over \ta } - {1 \over 2} \Big( {\d _i \over \t _i} \Big) ^{\top } \bSi ^{- 1} \Big( {\d _i \ta + \c _i - \X _i \bbe \over \t _i} \Big) \Big\} \Big] \Big) d\ta \non .
\end{align}
Noting that $t_{i,k}\ge 1$ due to the support of the prior, the quadratic term in the above expression is bounded as 
\begin{align}
&- {1 \over 2} \Big( {\d _i \over \t _i} \Big) ^{\top } \bSi ^{- 1} \Big( {\d _i \ta + \c _i - \X _i \bbe \over \t _i} \Big) \non \\
&\le {1 \over 2 \ta } \Big( {\c _i - \X _i \bbe \over \t _i} \Big) ^{\top } \bSi ^{- 1} \Big( {\d _i \ta + \c _i - \X _i \bbe \over \t _i} \Big) \non \\
&\le {1 \over 2 \ta } \Big\| {\c _i - \X _i \bbe \over \t _i} \Big\| \| ( \bSi ^{- 1 / 2} )^{\top } \| \Big\| \bSi ^{- 1 / 2} \Big( {\d _i \ta + \c _i - \X _i \bbe \over \t _i} \Big) \Big\| \non \\
&\le {1 \over 2 \ta } ( \| \c _i \| + \| \X _i \| \| \bbe \| ) \| \bSi ^{- 1 / 2} \| \Big\| \bSi ^{- 1 / 2} \Big( {\d _i \ta + \c _i - \X _i \bbe \over \t _i} \Big) \Big\| \non .
\end{align}
Then, we have  
\begin{align}
&R_{\Z } ( \om ) - R_{\Z } (1)  \non \\
&\le \mathbb{E} \Big[ \int_{1}^{\om } \Big[ L( \t _{1:n, \Zc (1:n)} , \z_{1:n} ; \bbe , \bSi , \tau ) \sum_{i = 1}^{n} \Big\{ {| \Lc (i)| \over \ta } \non \\
&\quad + {1 \over 2 \ta } ( \| \c _i \| + \| \X _i \| \| \bbe \| ) \| \bSi ^{- 1 / 2} \| \Big\| \bSi ^{- 1 / 2} \Big( {\d _i \ta + \c _i - \X _i \bbe \over \t _i} \Big) \Big\| \Big\} \Big] d\ta \Big] \non \\
&= \mathbb{E} \Big[ \int_{1}^{\om } \Big[ { L( \t _{1:n, \Zc (1:n)} , \z_{1:n} ; \bbe , \bSi , \tau ) \over \tau } \sum_{i = 1}^{n} \Big\{ | \Lc (i)| \non \\
&\quad + {1 \over 2 } ( \| \c _i \| + \| \X _i \| \| \bbe \| ) \| \bSi ^{- 1 / 2} \| \Big\| \bSi ^{- 1 / 2} \Big( {\d _i \ta + \c _i - \X _i \bbe \over \t _i} \Big) \Big\| \Big\} \Big] d\ta \Big] \non \\
&\le M_5 \sum_{i = 1}^{n} \mathbb{E} \Big[ \int_{1}^{\infty} \Big[ { L( \t _{1:n, \Zc (1:n)} , \z_{1:n} ; \bbe , \bSi , \tau ) \over \tau } \non \\
&\quad \times (1 + \| \bbe \| ) \Big\{ 1 + \sqrt{\tr ( \bSi ^{- 1} )} \Big\} \Big\| \bSi ^{- 1 / 2} \Big( {\d _i \ta + \c _i - \X _i \bbe \over \t _i} \Big) \Big\| \Big] d\ta \Big] \non ,
\end{align}
for some $M_5 > 0$. Thus $R_{\Z } (\om)$ is bounded by the sum of $R_{\Z } (1)$ and the expectation term in the last expression above. These two terms are finite for almost all $\{ (c_{i,k},d_{i,k}) \} _{i, k}$ as desired, as will be shown below. The first component of the sum, $R_{\Z}(1)$, is proportional to the marginal likelihood, or the density of $\y_{1:n}$, of the following model:
\begin{equation*}
    \begin{split}
        \y_i | (\bbe,\bSi,\t_i) &\stackrel{\rm ind.}{\sim} {\rm{N}}_{p} ( \X _{i} \bbe , 2 \T _{i} \bSi \T _{i} ) \\
        (\bbe,\bSi) &\sim \tilde{p}_1(\bbe,\bSi),
    \end{split}
\end{equation*}
where $t_{i,k}=1$ for $k\in \overbar{\Zc}(i)$ and $t_{i,k}\stackrel{\rm ind.}{\sim} \pi _1(t_{i,k})$ for $k\in \Zc(i)$. Similarly to our discussion made at the end of Step~2, this model consists of proper probability densities, hence the marginal density of $\y_{1:n}$, or $R_{\Z}(1)$, must be finite almost everywhere. The second component, or the expectation term, is also understood as the marginal likelihood of some synthetic model. To read off the corresponding model, write down the summand of the second term as 
\begin{equation*}
\begin{split}
    &\mathbb{E} \Big[ \int_{1}^{\infty} \Big[ { L( \t _{1:n, \Zc (1:n)} , \z_{1:n} ; \bbe , \bSi , \tau ) \over \tau } \non \\
    &\times (1 + \| \bbe \| ) \Big\{ 1 + \sqrt{\tr ( \bSi ^{- 1} )} \Big\} \Big\| \bSi ^{- 1 / 2} \Big( {\d _i \ta + \c _i - \X _i \bbe \over \t _i} \Big) \Big\| \Big] d\ta \Big] \\
    &= \int \Big[ \int_{1}^{\infty} \Big\{ \Big( \prod _{i=1}^p \int _{\mathbb{R}^{|\Lc(i)|}} \Big[ \Big\{ \prod_{k \in \Zc(i)} {\pi _1 ( t_{i, k} ) \over | t_{i, k} |} \Big\} \non \\
    &\quad \times {\ta ^{| \Lc (i)|} \over | \bSi |^{1 / 2}} \exp \Big\{ - {1 \over 4} \Big( {\ta \d _i + \c _i - \X _i \bbe \over \t _i} \Big) ^{\top } \bSi ^{- 1} \Big( {\ta \d _i + \c _i - \X _i \bbe \over \t _i} \Big) \Big\} \\
    &\quad \times \Big\| \bSi ^{- 1 / 2} \Big( {\d _i \ta + \c _i - \X _i \bbe \over \t _i} \Big) \Big\| \Big] d\t_i \Big) \frac{1}{\tau ^{1 + \ka / 3}} \Big\} d\tau \non \\
    &\quad \times (1 + \| \bbe \| ) \Big\{ 1 + \sqrt{\tr ( \bSi ^{- 1} )} \Big\} \tilde{p}_1(\bbe,\bSi) \Big] d(\bbe,\bSi) \\
    &= M_6 \int \int_{1}^{\infty} \Big[ \prod _{i=1}^p \int _{\mathbb{R}^{|\Lc(i)|}}  \tilde{f}_i( \d_i | \bbe,\bSi,\t_i,\tau ) \Big\{ \prod_{k \in \Zc(i)} \pi _1 ( t_{i, k} ) \Big\} d\t _{i,\Zc(i)} \Big] \tilde{p}_2( \tau, \bbe, \bSi ) d\tau d(\bbe,\bSi),
\end{split}
\end{equation*}
with some constant $M_6>0$, where the synthesis prior and sampling densities, $\tilde{p}_2$ and $\tilde{f}_i$'s, are defined by 
\begin{equation*}
    \tilde{p}_2( \tau, \bbe, \bSi ) \propto \frac{1}{\tau^{1 + \ka / 3}} (1 + \| \bbe \|^{1 + \ka } ) \Big( 1 + \sum_{k = 1}^{p} {\si _k}^{\ka } \Big) \Big\{ 1 + \sqrt{\tr ( \bSi ^{- 1} )} \Big\} p_0(\bbe,\bSi),
\end{equation*}
which is a proper probability density by assumption, and
\begin{equation*}
\begin{split}
    \tilde{f}_i ( \d_i | \bbe,\bSi,\t_i,\tau ) &\propto \Big\{ \prod_{k \in \Zc(i)} {1 \over | t_{i, k} |} \Big\} {\ta ^{| \Lc (i)|} \over | \bSi |^{1 / 2}} \Big\| \bSi ^{- 1 / 2} \Big( {\d _i \ta + \c _i - \X _i \bbe \over \t _i} \Big) \Big\| \non \\
    &\quad \times \exp \Big\{ - {1 \over 4} \Big( {\ta \d _i + \c _i - \X _i \bbe \over \t _i} \Big) ^{\top } \bSi ^{- 1} \Big( {\ta \d _i + \c _i - \X _i \bbe \over \t _i} \Big) \Big\} .
\end{split}
\end{equation*}
To confirm that $\tilde{f}$ is a probability density, we introduce fictitious observation, $\tilde{y}_i$, whose element is given by 
\begin{equation*}
    \tilde{y}_{i,k} = \begin{cases}
    d_{i,k} & {\rm if} \ k\in \Lc(i)\\
    c_{i,k} & {\rm if} \ k\in \Kc(i)
    \end{cases},
\end{equation*}
and define location and scale variable, $\bm{m}_i$ and $\tilde{\T}_i^{1/2} = \diag (\tilde{\t}_i)$, as 
\begin{equation*}
    m_{i,k} = \begin{cases}
    ( \x_{i,k}^{\top}\bbe - c_{i,k} ) / \tau & {\rm if} \ k\in \Lc(i)\\
    \x_{i,k}^{\top}\bbe & {\rm if} \ k\in \Kc(i)
    \end{cases}, \qquad \tilde{t}_{i,k} = \begin{cases}
    t_{i,k}/\tau & {\rm if} \ k\in \Lc(i)\\
    1 & {\rm if} \ k\in \Kc(i)
    \end{cases}.
\end{equation*}
Then we write $\tilde{f}$ as a function of $\tilde{\y}_i$ as  
\begin{align}
&\tilde{f}_i ( \d_i | \bbe,\bSi,\t_i,\tau ) \non \\
    &\propto {1 \over | \tilde{\T}_i\bSi^{1 / 2} |} \Big\|  \bSi ^{- 1 / 2} \tilde{\T}_i^{-1/2} \Big( \tilde{\y}_i - \bm{m}_i \Big) \Big\| \exp \Big\{ - {1 \over 4} \Big( \tilde{\y}_i - \bm{m}_i \Big) ^{\top } ( \tilde{\T}_i \bSi \tilde{\T}_i ) ^{- 1} \Big( \tilde{\y}_i - \bm{m}_i \Big) \Big\} . \non 
    \end{align}
In fact, this is the density of the multivariate version of the power truncated normal distribution with order $2$, location $\bm{m}_i$ and scale $2\tilde{\T}_i \bSi \tilde{\T}_i$, hence is a proper probability density. Thus, risk difference $R_{\Z}(\om)-R_{\Z}(1)$ is shown to be the marginal likelihood of the following hierarchical model: 
\begin{equation*}
    \begin{split}
        \d_i | (\bbe,\bSi,\t_{i,\Zc(i)},\tau ) &\stackrel{\rm ind.}{\sim} {\rm{PTN}}_{p} ( 2 , \bm{m}_i, 2\tilde{\T}_i \bSi \tilde{\T}_i) \\
        (\tau, \bbe,\bSi) &\sim \tilde{p}_2(\tau, \bbe,\bSi),
    \end{split}
\end{equation*}
and the same model applies to $t_{i,K}$ as in the first term. For more details about the power truncated normal distributions, see the next subsection. 

To summarize Step~3, we have shown the convergence in (\ref{eq:dom_st3}) by bounding the left hand side of (\ref{eq:dom_st3}) by $R_{\Z}(\om)$, which converges to zero. 

\subsection*{The power truncated normal distribution}

The univariate power truncated normal (PTN) distribution is denoted by $\mathrm{PTN}(p,m,c^2)$ with parameter $p>0$, location $m\not=0$ and scale $c>0$ (or $c^2$) and defined by the proper probability density, 
\begin{equation*}
    \mathrm{PTN}(x;p,m,c^2) \propto x^{p-1} \exp \left\{ - \frac{(x-m)^2}{2c^2} \right\}, \qquad x>0.
\end{equation*}
See, for example, He et al. (2019). 
This density is proper; the normalizing constant abbreviated above is the $(p{-}1)$-th moment of the corresponding normal distribution, $\mathrm{N}(m,c^2)$, which is finite and non-zero. Now, for $\x = (x_1,\dots ,x_n)^{\top}$, let $x_i\stackrel{\rm ind}{\sim}\mathrm{PTN}(p,m_i,1)$. Then, for positive definite matrix $\C$, we define $\y$ by $\y = \C^{1/2}\x$. The density of $\y$ is proper, and we denote the distribution of $\y$ by $\mathrm{PTN}_n(p,\m,\C)$ with $\m = (m_1,\dots ,m_n)^{\top}$.

\section{Proofs related to Sections~3.3, 3.4 and 3.5} \label{sm:misc}

\subsection{The exchange of the limit and integral in Secions~3.3 and 3.4}

In the analyses with the one-sided and asymmetric distributions of $t_{i,k}$ in Section~3.3 and the thin-tailed one in Section~3.4, we assumed the exchange of the limit ($\omega\to\infty$) and the integral (in $t_{i,k}$). The exchange has been justified in the previous section for the proposed model, but not for the three models here. In this subsection, we confirm the condition of the dominated convergence theorem for these models. 

First, we construct a general class of densities that covers the proposed model and the three models mentioned here. Consider the density is given by 
\begin{equation*}
	\begin{split}
		\pi ^{\ast}(t ) = w \mathbbm{1}[ t > 1 ] &{C_{\ga _1, \gat _1 } \over |t|^{1 + \gat _1}} {1 \over \{ 1 + \log |t| \} ^{1 + \ga _1}} \\
		&+ (1-w) \mathbbm{1}[ t < - 1 ] {C_{\ga _2, \gat _2} \over |t|^{1 + \gat _2}} {1 \over \{ 1 + \log |t| \} ^{1 + \ga _2}},
	\end{split}
\end{equation*}
with hyperparameters $(w,\ga_1,\ga_2,\gat_1,\gat_2)$. Here, $C_{\ga , \gat}$ is the normalizing constant and given by 
\begin{equation*}
		\frac{1}{C_{\ga , \gat }} = \int _1^{\infty} {1 \over  t^{1 + \gat }} {1 \over \{ 1 + \log (t) \} ^{1 + \ga }} dt,
\end{equation*}
which is finite if either (i) $\gat > 0$ and $\ga \ge -1$ or (ii) $\gat \ge 0$ and $\ga > 0$. By choosing the hyperparameters appropriately, we recover the models in Sections 3.3 and 3.4 as: 
\begin{itemize}
	\item (Proposed) $w=1/2$, $\gat_1=\gat_2=0$ and $\ga_1=\ga _2=\ga \in (0,\infty)$. 
	\item (One-sided log-Pareto) $w=1$, $\gat_1=0$ and $\ga_1 \in (0,\infty)$. 
	\item (Asymmetric) $w\in (0,1)$, $\gat_1=\gat_2=0$ and $\ga_1 , \ga_2 \in (0,\infty)$. 
	\item (Thin-tailed) $w=1/2$, $\gat_1=\gat_2=\gat > 0$ and $\ga_1=\ga _2=\ga \ge - 1$.  
\end{itemize}
Denote the unfolded log-Pareto density by $\pi (t; \gamma)$ (with parameter $\gamma$), for which we have verified the exchange of the limit and integral already. The general density here can be written as 
\begin{equation*}
	\begin{split}
		\pi ^{\ast}(t ) &= \frac{ wC_{\ga _1, \gat _1 } }{\ga _1/2} \mathbbm{1}[ t > 1 ] |t|^{-\gat_1} \pi(t;\ga_1) + \frac{ (1-w)C_{\ga _2, \gat _2 } }{\ga _2/2} \mathbbm{1}[ t < -1 ] |t|^{-\gat_2} \pi(t;\ga_2) \\
        &= \tilde{w}_1 \mathbbm{1}[ t > 1 ] |t|^{-\gat_1} \pi(t;\ga_1) + \tilde{w}_2  \mathbbm{1}[ t < -1 ] |t|^{-\gat_2} \pi(t;\ga_2) ,
	\end{split}
\end{equation*}
where $\tilde{w}_1 = wC_{\ga _1, \gat _1 } / (\ga _1/2)$ and $\tilde{w}_2 = (1-w) C_{\ga _2, \gat _2 } / (\ga _2/2)$.

We first review the use of the dominated convergence theorem in the proof of Theorem~1. If $t_{i,k}$'s follow $\pi$ independently, as verified in (\ref{eq:int_t}), the function of interest can be written as
\begin{equation*} 
	\begin{split}
		p( \y _i | \bbe , \bSi ) / \prod_{k \in \Lc(i)} \pi (| y_{i, k} |) %
		&= \int_{\mathbb{R} ^p} {\rm{N}}_p  \! \left( \left. \begin{bmatrix}
			\{ \y _{i, \Lc (i)} - \X _{i, \Lc (i)} \bbe \} / \{ |\y _{i, \Lc (i)} | / \bxi _{i, \Lc (i)} \} \\
			\{ \y _{i, \Kc (i)} - \X _{i, \Kc (i)} \bbe \} / \t _{i, \Kc (i)}
		\end{bmatrix} \right| \bm{0}^{(p)} , \bSi \right) \\ 
		&\quad \times \Big\{ \prod_{k \in \Lc (i)} {\pi (| y_{i, k} | / \xi _{i, k} ) \over \pi (| y_{i, k} |) | \xi _{i, k} |} \Big\} \Big\{ \prod_{k \in \Kc (i)} \pi ( t_{i, k} ) / | t_{i, k} | \Big\}  d\bxi _{i, \Lc (i)} d\t _{i, \Kc (i)} ,		
	\end{split}
\end{equation*}
then the integrand is bounded. In particular, the density of $t_{i,k}$ for outlying $k\in\Lc(i)$ is bounded as 
\begin{equation*}
	{\pi (| y_{i, k} | / \xi _{i, k} ) \over \pi (| y_{i, k} |) | \xi _{i, k} |} \le M  \{ 1 + \log (1 + | \xi _{i, k} |) \} ^{1 + \ga },
\end{equation*}
for some constant $M>0$. The integrability of the dominating function is confirmed by the moment, 
\begin{equation*}
	\mathbb{E}[ \{ 1 + \log (1 + | \xi _{i, k} |) \} ^{1 + \ga } ] \le \mathbb{E}[ (1 + | \xi _{i, k} |) ^{1 + \ga } ] < \infty,
\end{equation*}
where $\xi _{i,k}$ follows some normal distribution. This proof is still valid if the upper bound in the right hand side of the inequality above is of the form,
\begin{equation*}
	M  |\xi _{i, k} |^{\gat}  \{ 1 + \log (1 + | \xi _{i, k} |) \} ^{1 + \ga },
\end{equation*}
whose expected value is also finite. 

Now, suppose that $t_{i,k}$'s follow $\pi^{\ast}$ independently. It is sufficient to bound the ratio similarly by a function of $\xi_{i,k}$ whose expectation with respect to $\xi_{i,k}$ following the normal distribution is finite. The ratio with $\pi^{\ast}$ is computed as 
\begin{equation*}
\begin{split}
    {\pi ^{\ast}(| y_{i, k} | / \xi _{i, k} ) \over \pi ^{\ast}(| y_{i, k} |) | \xi _{i, k} |} &= \mathbbm{1}[ %
    | y_{i, k} | / \xi _{i, k} > 1 ] |\xi _{i,k}|^{\gat_1} {\pi (| y_{i, k} | / \xi _{i, k} ; \ga _1) \over \pi (| y_{i, k} |; \ga _1) | \xi _{i, k} |} \\ 
    &\ + \mathbbm{1}[ %
    | y_{i, k} | / \xi _{i, k} < - 1 ] \frac{\tilde{w}_2}{\tilde{w}_1} | y_{i, k} |^{ %
    \gat _1 - \gat _2} |\xi _{i,k}|^{%
    \gat _2}  {\pi (| y_{i, k} | / \xi _{i, k} ; \ga _2) \over \pi (| y_{i, k} |; \ga _1) | \xi _{i, k} |}.
\end{split}
\end{equation*}
If the density is symmetric ($\gat_1=\gat_2=\gat$ and $\ga_1=\ga_2=\ga$) and covers the thin-tailed case, then 
\begin{equation*}
\begin{split}
    {\pi ^{\ast}(| y_{i, k} | / \xi _{i, k} ) \over \pi ^{\ast}(| y_{i, k} |) | \xi _{i, k} |} &\le \Big\{ \mathbbm{1}[  y_{i, k} | / \xi _{i, k} > 1] + \frac{1-w}{w} \mathbbm{1}[ | y_{i, k} | / \xi _{i, k} < - 1] \Big\} \\
    &\qquad \times M |\xi_{i,k}|^{\gat} \{ 1+\log(1+|\xi_{i,k}|) \} ^{1+\ga},
\end{split}
\end{equation*}
for all $\xi_{i,k}$. The rest of the proof is almost identical to the proof of Theorem~1.  

In the case of the asymmetric tails ($\gat_1=\gat_2=0$ and $\ga_1\not=\ga_2$), we focus on the case of $\ga _1 < \ga_2$, which is sufficient for the discussions in Section~3.3. For $\xi _{i,k}>0$, the ratio is similarly bounded by $\{ 1+\log(1+|\xi _{i,k}|) \} ^{1+\ga_1}$ up to a constant, covering the one-sided case. For $\xi _{i,k}<0$, 
\begin{equation*}
    \begin{split}
        {\pi (| y_{i, k} | / \xi _{i, k} ; \ga _2) \over \pi (| y_{i, k} |; \ga _1) | \xi _{i, k} |} &= {\pi (| y_{i, k} | ; %
        \ga _2 ) \over \pi (| y_{i, k} |; %
        \ga _1 ) } {\pi (| y_{i, k} | / \xi _{i, k} ; \ga _2) \over \pi (| y_{i, k} |; \ga _2) | \xi _{i, k} |} \\
        &\le M' { \{ \log(1+|y_{i,k}| ) \} ^{-\ga _2} \over \{ \log(1+|y_{i,k}| ) \} ^{-\ga _1} } \{ 1+\log(1+|\xi _{i,k}|) \} ^{1+\ga_2} \\
        &\le M' \{ 1+\log(1+|\xi _{i,k}|) \} ^{1+\ga_2}, 
    \end{split}
\end{equation*}
for sufficiently large $|y_{i,k}|$, covering the asymmetric case.

\subsection{Outlier-sensitivity of the scaled variance model in Section 3.5}

Under the scaled variance, $V(t_i,\bSi) = t_i^2 \bSi$, it is difficult to achieve the likelihood robustness. Here, we observe the lack of such robustness can be seen in a simple example of $p=2$ with the specific choice of scaling constant $C_i(\om)$, parameter $(\bbe,\bSi)$, and patterns of outliers $\{ \Lc(i),\Kc(i) \}$. 

Set $p=2$, fix $i \in \{ 1,\dots ,n \}$ and consider $\Lc(i)=\{ 1,2 \}$ and $y_{i,1}=y_{i,2}=\omega$ ($d_{i,1}=d_{i,2}=1$ and $c_{i,1}=c_{i,2}=0$). 
Assume the log-Pareto tailed distribution for $v_i=t_i^2$ as 
\begin{equation*}
	p(v_i) =  \mathbbm{1}[ v_i > 0 ] \frac{ 1 }{ (1 + v_i) } \frac{ 1 }{ \{ 1 + \log (1 + v_i ) \} ^{1 + \ga } }.
\end{equation*}
We consider the scaled likelihood with $C_i(\om) = \pi (|y_{i,k}|) / \sqrt{ 2\pi} = p(\om ^2) / \sqrt{ 2\pi}$. The likelihood robustness requires the convergence of the scaled likelihood to some constant for all $\bSi$. Here, we prove that this does not hold at $\bSi = \diag \{ \sigma _1^2, \sigma _2^2 \}$.

By the change-of-variables $\tilde{v}_i = v_i / \om ^2$, the scaled likelihood can be computed as 
\begin{align*}
	&{p( \y _i | \bbe , \bSi ) \over p( \om ^2 ) / (2 \pi )}\\
	&= \int_{0}^{\infty } {1 \over ( \si _{1}^{2} \si _{2}^{2} )^{1 / 2} v_i} \exp \Big[ - {1 \over 2 v_i} \Big\{ {( y_{i, 1} - {\x _{i, 1}}^{\top } \bbe )^2 \over \si _{1}^{2}} + {( y_{i, 2} - {\x _{i, 2}}^{\top } \bbe )^2 \over \si _{2}^{2}} \Big\} \Big] {p( v_i ) \over p( \om ^2 )} d{v_i} \non \\
	&= \int_{0}^{\infty } {1 \over ( \si _{1}^{2} \si _{2}^{2} )^{1 / 2} {\tilde{v}_i}^2} \exp \Big[ - {1 \over 2 \tilde{v}_i} \Big\{ {(1 - {\x _{i, 1}}^{\top } \bbe / \om )^2 \over \si _{1}^{2}} + {(1 - {\x _{i, 2}}^{\top } \bbe / \om )^2 \over \si _{2}^{2}} \Big\} \Big] {p( \om ^2 \tilde{v}_i ) \tilde{v}_i \over p( \om ^2 )} d{\tilde{v}_i} \text{.} \non 
\end{align*}
Since $(1 - {\x _{i, 1}}^{\top } \bbe / \om )^2 / \si _{1}^{2} \ge (1 / 2)^2 / \si _{1}^{2}$ and 
\begin{align*}
	{p( \om ^2 \tilde{v}_i ) \tilde{v}_i \over p( \om ^2 )} &= {(1 + \om ^2 ) \tilde{v}_i \over 1 + \om ^2 \tilde{v}_i} \Big\{ {1 + \log (1 + \om ^2 ) \over 1 + \log (1 + \om ^2 \tilde{v}_i )} \Big\} ^{1 + \ga } \le 2 \{ 1 + \log (1 + 1 / \tilde{v}_i ) \} ^{1 + \ga } ,
\end{align*}
as $\om \to \infty $, it follows from the dominated convergence theorem and the regularly varying $p(\cdot )$ that 
\begin{align*}
	{p( \y _i | \bbe , \bSi ) \over p( \om ^2 ) / (2 \pi )} &\to \int_{0}^{\infty } {1 \over ( \si _{1}^{2} \si _{2}^{2} )^{1 / 2} {\tilde{v}_i}^2} \exp \Big\{ - {1 \over 2 \tilde{v}_i} \Big( {1 \over \si _{1}^{2}} + {1 \over \si _{2}^{2}} \Big) \Big\} d{\tilde{v}_i} = 2 {( \si _{1}^{2} \si _{2}^{2} )^{1 / 2} \over \si _{1}^{2} + \si _{2}^{2}} , 
\end{align*}
as $\om \to \infty $. 
The limit in the right-hand side depends on $\bSi $ and is not a constant.

\section{MCMC} \label{sm:mcmc}

The MCMC algorithm in Section~4 of the main text is further sophisticated for the use in Section~5, as detailed in this section. Note that the notations used in this section differs from those in Section~4. 

\subsection{Model review}

As in the main text, we consider the reciprocal of the latent variable, defined as $t_{i,k} = (1 / \tilde{t} _{i, k} )^{z_{i, k}}$ for $\tilde{t} _{i, k} \in (- 1, 1)$ and $z_{i, k} \in \{ 0, 1 \} $. The joint distribution for $( \tilde{t} _{i, k} , z_{i, k} )$ is given by  
\begin{align*}
p( \tilde{t} _{i, k} , z_{i, k} ) &= \mathbbm{1}[ | \tilde{t} _{i, k} | < 1 ] {\ga \over 2} {1 \over | \tilde{t} _{i, k} |} {1 \over \{ 1 + \log (1 / | \tilde{t} _{i, k} |) \} ^{1 + \ga }} \times (1 - \phi )^{1 - z_{i, k} } {\phi }^{z_{i, k}} \text{.}
\end{align*}
With the conjugate priors $\phi \sim {\rm{Beta}} (a_0, b_0)$ and $(\bbe,\bSi) \sim p_0(\bbe,\bSi)$, the joint posterior distribution of $( \tilde{\t } , \z , \phi , \bbe , \bSi )$ given $\y _1 , \dots , \y _n$ is 
\begin{align*}
p( \tilde{\t } , \z , \phi , \bbe , \bSi &| \y _1 , \dots , \y _n ) \propto p_0 ( \bbe , \bSi ) \phi ^{a - 1} (1 - \phi )^{b - 1} \\
&\quad \times \prod_{i = 1}^{n} \Big( \Big[ \prod_{k = 1}^{p} \Big\{ \mathbbm{1}[ | \tilde{t} _{i, k} | < 1] {1 \over | \tilde{t} _{i, k} |} {1 \over (1 - \log | \tilde{t} _{i, k} |)^{1 + \ga }} (1 - \phi )^{1 - z_{i, k} } {\phi }^{z_{i, k}} \Big\} \Big] \\
&\quad \times {\prod_{k = 1}^{p} | \tilde{t} _{i, k} |^{z_{i, k}} \over | \bSi |^{1 / 2}} \exp \Big[ - {1 \over 2} \Big\{ {\y _i - \X _i \bbe \over (1 / \tilde{\t } _i )^{\z _i}} \Big\} ^{\top } \bSi ^{- 1} {\y _i - \X _i \bbe \over (1 / \tilde{\t } _i )^{\z _i}} \Big] \Big) \text{.} 
\end{align*}
From this expression, it is immediate that sampling $\bbe $ and $\bSi $ separately from their full conditional distributions is straightforward if $p_0 ( \bbe , \bSi )$ is the independent normal and inverse-Wishart prior. 
The full conditional of $\phi $ is ${\rm{Beta}} \big( \phi \big| \sum_{i = 1}^{n} \sum_{k = 1}^{p} z_{i, k} + a_0, \sum_{i = 1}^{n} \sum_{k = 1}^{p} (1-z_{i, k}) + b_0 \big) $. 

\subsection{The normal augmentation and Sampling of $\bth_i$}

Before applying the augmentation by normal distributions, we consider the singular value decomposition of the scaled precision matrix. Let $\psi _k = \sqrt{ (\bSi^{-1})_{k,k} }$, where $(\bSi^{-1})_{k,k}$ is the $k$-th diagonal element of matrix $\bSi^{-1}$. Then, define $\Q = \{ \diag (\bpsi ) \} ^{-1} \bSi ^{-1} \{ \diag (\bpsi ) \} ^{-1}$; this is the correlation matrix of $\bSi^{-1}$ in the sense that the diagonal elements of $\Q$ are all ones. Finally, for some orthonormal $\H$ and diagonal $\bLa$, we set $\Q = \H \bLa \H^{\top}$, where $\bLa = \diag (\lambda _1,\dots, \lambda _p)$ with $\lambda _1 \ge \cdots \ge \lambda _p > 0$. We also write $\bLa = \diag (\lambda _k)$.

Let $\ybt _i = \bpsi (\y _i - \X_i\bbe) / {\t_i}^{\z _i}$. Then, with $c>0$ such that $c^{-1}\I^{(p)} < \bLa^{-1}$, we have 
\begin{equation*}
\begin{split}
    {\rm N}_p( \y_i | \X_i\bbe , \T_i\bSi\T_i ) &= {\rm N}_p( \tilde{\y}_i | \bm{0} , \Q ^{-1} ) \prod _{k=1}^p {\psi _k \over |t_{i,k}|} \\ 
    &= \int {\rm N}_p( \tilde{\y}_i | \bth_i , c^{-1}\I^{(p)} ) {\rm N}_p( \bth_i | \bm{0} , \Q ^{-1} - c^{-1}\I^{(p)} ) d\bth_i \prod _{k=1}^p {\psi _k \over |t_{i,k}|} .
\end{split}
\end{equation*}
The exponential term in the right hand side has the quadratic form, 
\begin{equation*}
     \bth_i^{\top} [ (\Q^{-1} - c^{-1} \I) ^{-1} + c\I ] \bth_i = (\H ^{\top} \bth_i)^{\top} [ \diag \{ c^2 / (c - \lambda _k  ) \}  ] (\H ^{\top} \bth_i),
\end{equation*}
hence the full conditional is $c \bth_i \sim {\rm N}_p( \H \diag \{ c - \lambda _k \} \H^{\top} \tilde{\y}_i, \H \diag \{ c - \lambda _k \} \H^{\top} )$. %

In application, we set $c = c_0 + \lambda _1$ with $c_0 = 10 ^{-8} > 0$.

\subsection{Sampling of $\z_i$}

The full conditional of $\z_i$ is obtained after introducing $\bth_i$. Here, we additionally sample the signs of latent variables to accelerate the convergence. 

We decompose $\t _i$ (or $\tilde{\t}_i$) into its signs and absolute values. That is, we set $s_{i,k} = \sgn (t_{i,k})$ and $r_{i,k}=|t_{i,k}|$ so that we write $\t _i = (\s_i \r_i)^{\z_i}$. Conditional on $\r_i$, we sample $(\z_i,\s_i)$ jointly. Note that $\tilde{\y}_i$ is dependent on $\t_i$, hence on $(\z_i,\s_i)$. Since $(\z_1,\s_1),\dots ,(\z_n,\s_n)$ are mutually independent, we focus on the sampling of $(\z_i,\s_i)$. The full conditional is, by writing the $k$-th element of $\X_i\bbe$ by $\x_{i,k}^{\top}\bbe$,  
\begin{equation*}
    \prod _{k=1}^p \left( \frac{\phi}{ r_{i,k}(1-\phi) } \right)^{z_{i,k}} \exp \left\{ -\frac{c}{2} \left( \psi _k \frac{y_{i,k} - \x_{i,k}^{\top}\bbe }{ ( s_{i,k}r_{i,k} ) ^{z_{i,k}} } - \theta_{i,k} \right) ^2 \right\},
\end{equation*}
which shows the independence of $(z_{i,1},s_{i,1}), \dots, (z_{i,p},s_{i,p})$ and identifies the discrete distribution of $(z_{i,k},s_{i,k}) \in \{ 0,1\} \times \{ 1,-1\}$.

\subsection{Sampling of $\t_i$ by HMC with latent variables}

In sampling $\t_i$, we use its reciprocal $\tilde{\t}_i = 1/ ( \t_i^{\z_i} )$. Since $\z_i$ is conditioned, we sample those with $z_{i,k}=1$, whose support is $\tilde{t}_{i,k} \in (-1,1)$. For those with $z_{i,k}=0$, the likelihood is free from $t_{i,k}$, so we can generate $t_{i,k}$ from the prior, or the unfolded log-Pareto distribution, following the procedure given in the Appendix of the main text. 

We use the slice sampler and the HMC method as explained in the main text, but introduce another set of latent variables to avoid the numerical errors at the HMC step. The $|\Zc(i)|$-dimensional vector of latent variables is denoted by $\bze _i$, based on the normal density and its integral, 
\begin{equation*}
    \begin{split}
        1 = \int {\rm{N}}_{|\Zc(i)|} ( \bze _i | \de \tilde{\t}_{i,\Zc(i)}, \de \I ) d\bze _i.
    \end{split}
\end{equation*}
The full conditional of $\tilde{t}_{i,\Zc(i)}$ is computed conditional on $\bze_i$. As in the main text, the log-Pareto density of $\tilde{t}_{i,k}$ is augmented as
\begin{equation*}
\begin{split}
    \pi (\tilde{t}_{i,k}) &= \mathbbm{1}[ |\tilde{t}_{i,k}|<1 ] \frac{\gamma / 2}{|\tilde{t}_{i,k}|} \{ 1 - \log |\tilde{t}_{i,k}| \}^{-(1+\ga )} \\
    &= \mathbbm{1}[ |\tilde{t}_{i,k}|<1 ] \frac{\gamma / 2}{|\tilde{t}_{i,k}|} \int \mathbbm{1}[ 0 < u_{i,k} < \{ 1 - \log |\tilde{t}_{i,k}| \}^{-(1+\ga )} ] du_{i,k} \\ 
    &= \frac{\gamma / 2}{|\tilde{t}_{i,k}|} \int \mathbbm{1}[ \ \exp \{ 1-u_{i,k}^{ -1/(1+\ga )} \} < |\tilde{t}_{i,k}| < 1 \ ] du_{i,k} . 
\end{split}
\end{equation*}
The full conditional of $u_{i,k}$ is the uniform distribution on the interval $(0 , \{ 1 - \log |\tilde{t}_{i,k}| \}^{-(1+\ga )})$.

Given latent variables $\bze_i$ and $\u_i$, the full conditional density of $\tilde{t}_{i,\Zc(i)}$ is the product of the following three components: 
\begin{itemize}
    \item Likelihood: with $\bPsi _i = \diag (\y_i - \X_i\bbe) \bSi^{-1} \diag (\y_i - \X_i\bbe)$, 
    \begin{equation*}
    \begin{split}
        {\rm N}_p( \y_i&|\X_i\bbe,\T_i\bSi\T_i ) \propto \prod _{k \in \Zc(i)} |\tilde{t}_{i,k}|\\
        &\times \exp \left\{ -\frac{1}{2} \left[ \tilde{\t}_{i,\Zc(i)}^{\top} (\bPsi _i)_{\Zc(i),\Zc(i)} \tilde{\t}_{i,\Zc(i)} %
        + 2 \tilde{\t}_{i,\Zc(i)}^{\top} (\bPsi _i)_{\Zc(i),\overbar{\Zc}(i)} \bm{1}^{ ( |\overbar{\Zc}(i)| )} \right] \right\} . 
    \end{split}
    \end{equation*}
    
    \item Log-Pareto prior: 
    \begin{equation*}
        \prod _{k\in \Zc(i)} |\tilde{t}_{i,k}|^{-1}\mathbbm{1}[ \ \exp \{ 1-u_{i,k}^{ -1/(1+\ga )} \} < |\tilde{t}_{i,k}| < 1 \ ].
    \end{equation*}
    \item Contribution from the augmented density: 
    \begin{equation*}
        \exp \left\{ -\frac{1}{2} \left[ \tilde{\t}_{i,\Zc(i)}^{\top} (\delta \I) \tilde{\t}_{i,\Zc(i)} - 2 \tilde{\t}_{i,\Zc(i)}^{\top} %
        \bze_i \right] \right\}. %
    \end{equation*}
\end{itemize}
Combining them all, we have the full conditional as 
\begin{equation} \label{eq:tmg}
\begin{split}
    &\prod _{k\in \Zc(i)} \mathbbm{1}[ \ \exp \{ 1-u_{i,k}^{ -1/(1+\ga )} \} < |\tilde{t}_{i,k}| < 1 \ ] \\
    & \times \exp \left\{ -\frac{1}{2} \left[ \tilde{\t}_{i,\Zc(i)}^{\top} [ (\bPsi _i)_{\Zc(i),\Zc(i)} + \delta \I ]\tilde{\t}_{i,\Zc(i)} %
    + 2 \tilde{\t}_{i,\Zc(i)}^{\top} [ (\bPsi _i)_{\Zc(i),\overbar{\Zc}(i)} \bm{1}^{ ( |\overbar{\Zc}(i)| )} %
    - \bze_i ] \right] \right\},
\end{split}
\end{equation}
which is the multivariate truncated normal distribution. We set $\delta = 1$ in our implementation in Section~5.

\subsection{The MCMC algorithm}

Each iteration of the proposed MCMC algorithm is summarized as follows.  
\begin{itemize}
\item
Update $\z _i$ as follows. 
\begin{itemize}
\item 
For each $k = 1, \dots , p$, sample each of $\tilde{\bth}_i = (c_0+\lambda_1) \H ^{\top} \bth_i$ in parallel and independently from 
\begin{align*}
{\rm{N}} (  ( c_0 + \la _1 - \la _k )  ( \H ^{\top } \tilde{\y}_i )_k , c_0 + \la _1 - \la _k ) \text{.} 
\end{align*}
\item
Sample $( s_{i,k}, z_{i, k} ) \in \{ 1, - 1 \} \times \{ 0, 1 \} $ from the discrete distribution, whose probability function is given for each $( s_{i,k}, z_{i, k} )$ by
\begin{align}
&\Big( { \phi \over r_{i,k}( 1 - \phi ) } \Big) ^{z_{i, k}} \exp \Big( %
( \H \tilde{\bth } _i )_k \psi _k {y_{i, k} - {\x _{i, k}}^{\top } \bbe \over \{ s_{i,k} r_{i,k} \} ^{z_{i, k}}} - {1 \over 2} ( \la _1 + c_0 ) {\psi _k}^2 \Big[ {y_{i, k} - {\x _{i, k}}^{\top } \bbe \over \{ s_{i,k} r_{i,k} \} ^{z_{i, k}}} %
\Big] ^2 \Big) \non 
\end{align}
after normalization, where $r_{i,k} = | t_{i, k} |$. Then, set $t_{i,k}=s_{i,k}r_{i,k}$ with newly-sampled $s_{i,k}$. 
\end{itemize}

\item Sample $(\bbe,\bSi)$. If the prior $p_0(\bbe,\bSi)$ is the independent normal and inverse-Wishart distribution, then sample $\bbe$ from the normal distribution conditional on $\bSi$, then sample $\bSi$ from the inverse-Wishart distribution conditional on $\bbe$. 
\item Sample $\phi$ from ${\rm{Beta}} \big( \phi \big| \sum_{i = 1}^{n} \sum_{k = 1}^{p} z_{i, k} + a_0, \sum_{i = 1}^{n} \sum_{k = 1}^{p} (1-z_{i, k}) + b_0 \big) $. 

\item
Update $\tilde{\t } _i$  as follows.
\begin{itemize}
\item
For each $k = 1, \dots , p$, sample 
\begin{align}
u_{i, k} &\sim {\rm{Uniform}} \Big( 0, {1 \over (1 - \log | \tilde{t} _{i, k} |)^{1 + \ga }} \Big) \text{.} \non 
\end{align}
\item Sample $\tilde{\t } _{i, \overbar{\Zc }_i}$ from the prior by generating $|\overbar{\Zc}_i|$ independent draws from the unfolded log-Pareto distribution and taking their reciprocals. 
\item
Sample $\tilde{\t } _{i, \Zc _i}$ by using the HMC method of Pakman and Paninski (2014). 
The full conditional of $\tilde{\t}_{i,\Zc_i}$ is the truncated multivariate normal distribution whose density kernel given in (\ref{eq:tmg}). 
\end{itemize}
\end{itemize}

\section{Additional numerical results} \label{sm:numerical}

We report simulation results under multivariate linear regression with $p=5$, given in Section~5.2 in the main text. 
The mean squared errors (MSE) of $\bbe$ and $\bSi$ averaged over 200 Monte Carlo replications are given in Figure~\ref{fig:MLR-mse-supp}, and the coverage probability (CP) and interval score (IS) of $95\%$ credible intervals of $\bbe$ and $\bSi$ are shown in Figure~\ref{fig:MLR-cp-supp}.
Overall, the relative performance of the four methods are similar to those under $p=10$.

\begin{figure}[htbp!]
\centering
\includegraphics[width=\linewidth]{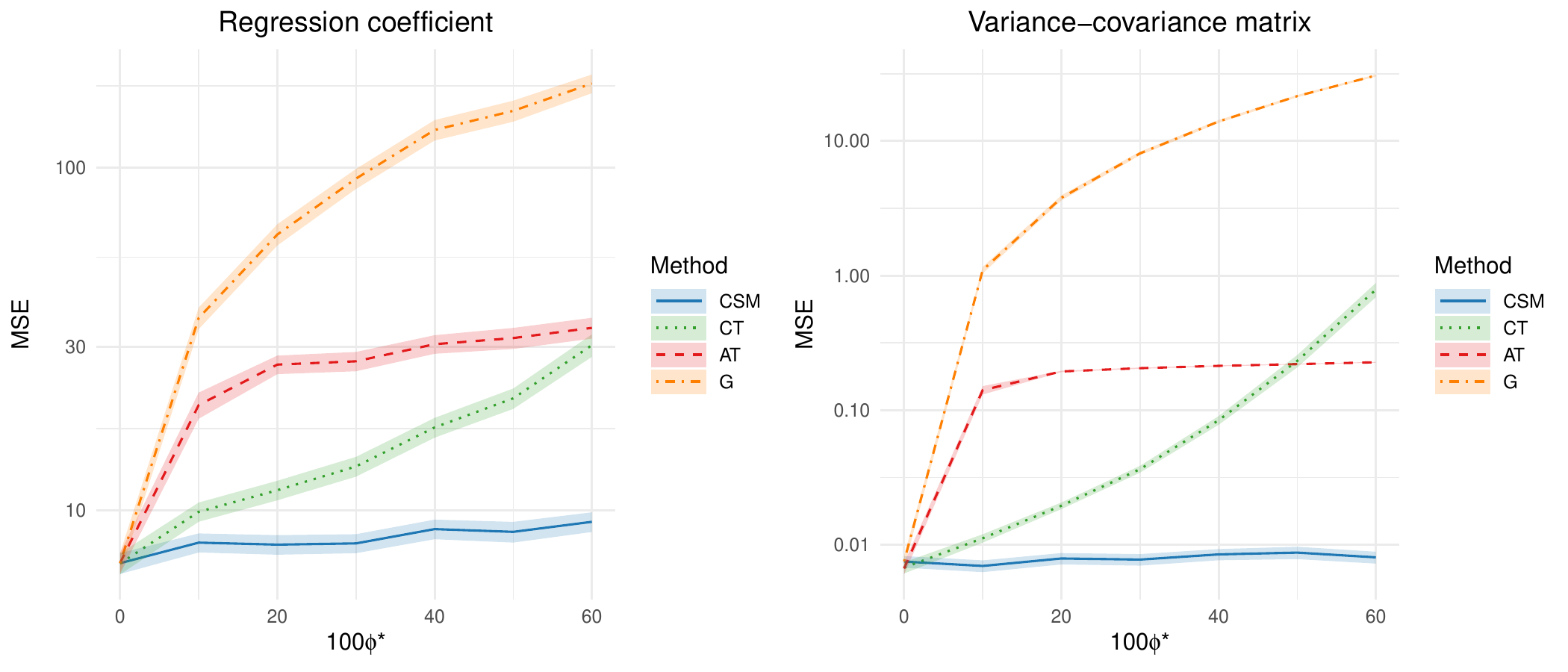} 
\caption{ Mean squared errors (MSE) of posterior means of regression coefficients $\bbe$ (left) and variance-covariance matrix $\bSi$ (right) in multivariate linear regression with $p=5$. The shaded region represents estimated Monte Carlo errors.
}
\label{fig:MLR-mse-supp}
\end{figure}

\begin{figure}[htbp!]
\centering
\includegraphics[width=\linewidth]{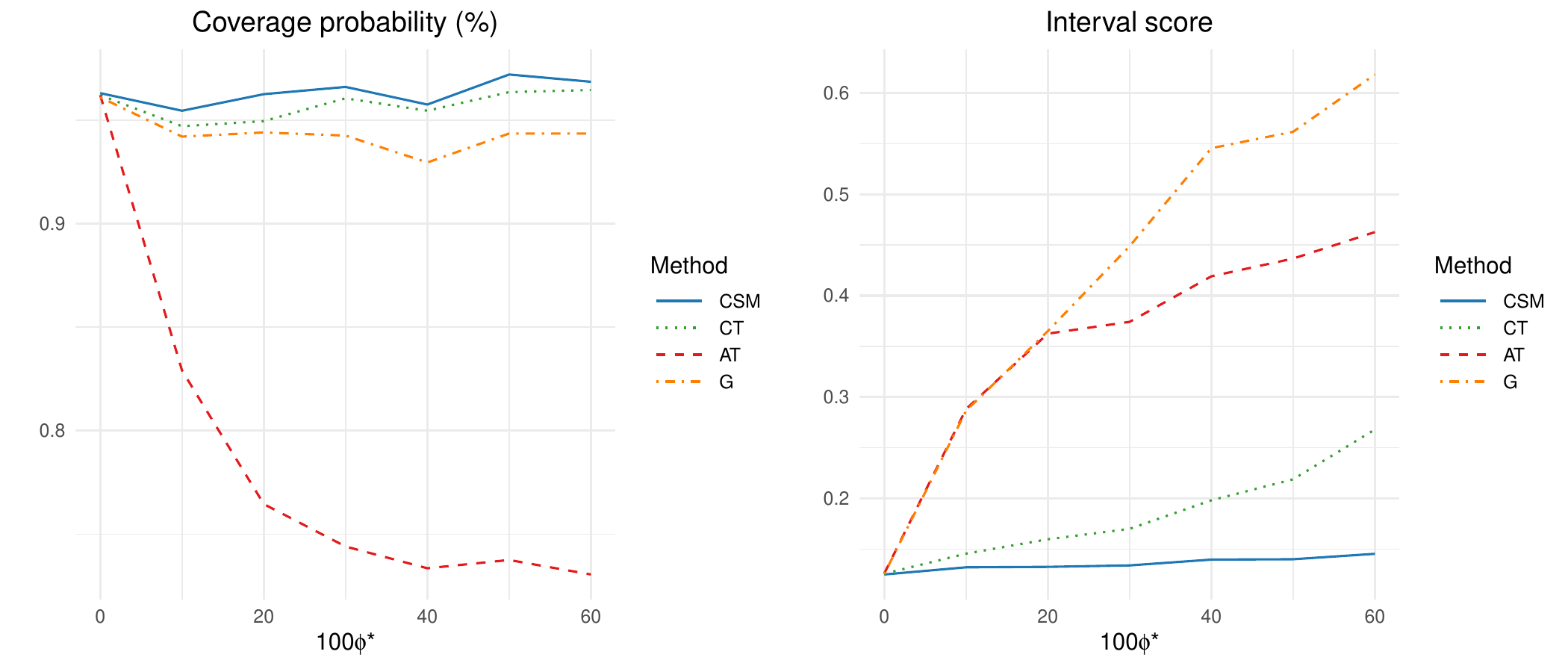}
\caption{ Coverage probability (CP) and interval score (IS) of $95\%$ credible intervals of $\bbe$ (left) and $\bSi$ (right) in multivariate linear regression with $p=5$.} 
\label{fig:MLR-cp-supp}
\end{figure}

\clearpage 

\section{Density functions of $\y_i$ and Figures 1 and 2} \label{sm:fig}

Figure~1 in the main text shows the contour plots of the joint density of $(y_{i1},y_{i2})$ with latent variables $(t_{i1},t_{i2})$ being marginalized out. This density function is written as the integral, 
\begin{equation*}
    \int \mathrm{N}(\y _i | \X_i\bbe, \T_i\bSi\T_i) \left[ \prod _{k=1}^p \pi (t_{i,k}) \right]  d\t_i.
\end{equation*}
We set $\bbe = \bm{0}$, $\Sigma _{11}=\Sigma _{22} = 1$ and $\Sigma _{12} = \Sigma _{21} = \rho$, so that the integral is simplified as 
\begin{equation*}
    \int\!\!\!\int \frac{1}{2\pi \sqrt{1-\rho^2}} \exp \left\{ -\frac{ 1 }{2(1-\rho^2)} \left[ \frac{y_{i,1}^2}{t_{i,1}^2} - 2 \rho \frac{y_{i,1}y_{i,2}}{t_{i,1}t_{i,2}} + \frac{y_{i,2}^2}{t_{i,2}^2} \right] \right\} \pi(t_{i,1}) \pi(t_{i,2}) dt_{i,1}dt_{i,2},
\end{equation*}
although the closed form is still difficult to obtain. The contour plots of the density function of this form can be found in  Choy et al. (2014) and Griffin and Hoff (2024) 
for a particular choice of $\pi(\cdot )$. As in Griffin and Hoff (2024), 
we approximated the integral above by the Monte Carlo method to draw Figure~1. 
To be precise, we generated $S$ random variables from $\pi (\cdot )$ independently (denoted by $t_{i,k}^{(s)}$ for $s=1,\dots ,S$) and approximated the integral by 
\begin{equation*}
    \frac{1}{S} \sum _{s=1}^S \frac{1}{2\pi \sqrt{1-\rho^2}} \exp \left\{ -\frac{ 1 }{2(1-\rho^2)} \left[ \frac{y_{i,1}^2}{( t_{i,1} ^{(s)})^2} - 2 \rho \frac{y_{i,1}y_{i,2}}{t_{i,1} ^{(s)}t_{i,2} ^{(s)}} + \frac{y_{i,2}^2}{(t_{i,2} ^{(s)})^2} \right] \right\} ,
\end{equation*}
for each value of $(y_{i1},y_{i2})$. We set the Monte Carlo size to $S=100000$. The hyperparameters used in Figure~1 are: $\rho = 0.9$, $\phi = 0.25$ and $\ga = 1$. For comparison, we also considered the shifted, unfolded inverse $\chi _{\nu}$ distribution with degree-of-freedom $\nu = 1$ as $\pi (\cdot )$, as detailed below. 

Figure~2 displays the density and distribution functions of the unfolded shifted log-Pareto distribution given in equation~(3) of the main text. Note that the distribution function of this probability distribution is analytically obtained as: for $t_0>1$,
\begin{equation*}
    \mathbb{P}[t_{i,k}\le -t_0] = \frac{1}{2} \{ 1+\log t_0 \}^{-\ga} ,
\end{equation*}
and $\mathbb{P}[t_{i,k}\le t_0] = 1 - \mathbb{P}[t_{i,k}\le -t_0]$. 
These functions are compared with those of the polynomial-tailed distribution. We chose the $\chi_{\nu}^2$ distribution for $t_{i,k}^{-2}$ with degree-of-freedom $\nu$ used in Choy et al. (2014), 
but shifted its location to exclude interval $(-1,1)$ from the support for a fair comparison. Its density function is obtained by the change of variable as  
\begin{equation*}
    \pi (t_{i,k}) = \frac{1}{\Gamma (\nu/2) (\nu /2)^{-\nu/2}} (|t_{i,k}|-1)^{-\nu -1} \exp \left\{ -\frac{\nu}{2 (|t_{i,k}|-1)^2} \right\} \mathbbm{1}[ |t_{i,k}| > 1].
\end{equation*}
Note that $(|t_{i,k}|-1)^{-2}$ is $\chi ^2_\nu$-distributed. The distribution function can be computed by using those of the inverse $\chi^2_{\nu}$ (or inverse gamma) distributions.

\end{document}